\pdfoutput=1

\documentclass{scrartcl}

\usepackage[utf8]{inputenc}
\usepackage{amsmath}
\usepackage{amssymb}
\usepackage{amsthm}
\usepackage{thmtools}
\usepackage{thm-restate}
\usepackage{csquotes}
\usepackage{enumitem}
\usepackage{mathtools}
\usepackage{mathrsfs} 

\usepackage{caption}
\usepackage{subcaption}
\usepackage{floatpag}

\usepackage{tikz}
\usetikzlibrary{shapes, matrix, decorations.pathreplacing, calc, positioning, patterns}

\pgfdeclaredecoration{penciline}{initial}{
  \state{initial}[width=+\pgfdecoratedinputsegmentremainingdistance,auto corner on length=1mm,]{
    \pgfpathcurveto%
    {
      \pgfqpoint{\pgfdecoratedinputsegmentremainingdistance}
      {\pgfdecorationsegmentamplitude}
    }
    {
      \pgfmathrand
      \pgfpointadd{\pgfqpoint{\pgfdecoratedinputsegmentremainingdistance}{0pt}}
      {\pgfqpoint{-\pgfdecorationsegmentaspect\pgfdecoratedinputsegmentremainingdistance}%
        {\pgfmathresult\pgfdecorationsegmentamplitude}
      }
    }
    {
      \pgfpointadd{\pgfpointdecoratedinputsegmentlast}{\pgfpoint{1pt}{1pt}}
    }
  }
  \state{final}{}
}

\usepackage[hidelinks, pdfstartview=FitH]{hyperref}

\usepackage{adjustbox}
\usepackage{tabularx}
\newcommand{\decProblem}[3][]{%
  \par\vspace{0.125cm plus 0.1cm minus 0.05cm}\adjustbox{valign=t}{\begin{tabularx}{\linewidth-2\parindent}{@{}lX}%
      \if\relax\detokenize{#1}\relax%
      \else%
      \textnormal{\textbf{Constant:}}&#1\\%
      \fi%
      \textnormal{\textbf{Input:}}&#2\\%
      \textnormal{\textbf{Question:}}&#3\\%
  \end{tabularx}}\vspace{0.125cm plus 0.1cm minus 0.05cm}\par%
}
\newcommand{\compProblem}[3][]{%
  \par\vspace{0.125cm plus 0.1cm minus 0.05cm}\adjustbox{valign=t}{\begin{tabularx}{\linewidth-2\parindent}{@{}lX}%
      \if\relax\detokenize{#1}\relax%
      \else%
      \textnormal{\textbf{Constant:}}&#1\\%
      \fi%
      \textnormal{\textbf{Input:}}&#2\\%
      \textnormal{\textbf{Output:}}&#3\\%
  \end{tabularx}}\vspace{0.125cm plus 0.1cm minus 0.05cm}\par%
}

\declaretheorem[numberwithin=section, style=plain]{theorem}
\declaretheorem[style=plain, sibling=theorem]{proposition}
\declaretheorem[style=plain, sibling=theorem]{lemma}
\declaretheorem[style=plain, sibling=theorem]{corollary}
\declaretheorem[style=plain, sibling=theorem]{fact}

\declaretheorem[style=definition, sibling=theorem]{definition}

\declaretheorem[style=remark, sibling=theorem]{example}
\declaretheorem[style=remark, numbered=no, name=Remark]{remark*}
\declaretheorem[style=remark, sibling=theorem]{remark}

\usepackage{todonotes}

\newcommand*{\Np}{\mathbb{N}^+} 

\newcommand*{\V}[1][V]{\boldsymbol{\mathrm{#1}}}
\newcommand*{\DAb}{\V[DA\!b]}

\newcommand*{\R}{\mathrel{\mathcal{R}}}
\renewcommand*{\L}{\mathrel{\mathcal{L}}}
\newcommand*{\J}{\mathrel{\mathcal{J}}}
\renewcommand*{\H}{\mathrel{\mathcal{H}}}

\DeclareMathOperator{\dom}{dom}
\DeclareMathOperator{\alphabet}{alph}

\newcommand*{\Xinfty}{X_{+\infty}}
\newcommand*{\Yinfty}{Y_{-\infty}}
\newcommand*{\freeProf}[1][\Sigma]{\ensuremath{\hat{#1}^*}}

\title{The Word Problem for $(\omega - 1)$-Terms\\over $\DAb$}
\usepackage[affil-it]{authblk}

\author{Jorge Almeida}
\affil{
  Centro de Matemática da Universidade do Porto (CMUP)\\
  Departamento de Matemática, Faculdade de Ciências, Universidade do Porto\\
  Rua do Campo Alegre s/n\\
  4169-007 Porto, Portugal
}
\author{Manfred Kuf\-leitner}
\affil{
  Institut für Formale Methoden der Informatik (FMI)\\
  Universität Stuttgart\\
  Universitätsstraße 38\\
  70569 Stuttgart, Germany
}
\author{Jan~Philipp~Wächter}
\affil{Department of Mathematics\\
  University of Manchester\\
  Oxford Road\\
  Manchester M13 9PL, UK}

\begin{document}

  \maketitle

  \begin{abstract}
    \noindent\textbf{Abstract.}
    We give a ranker-based description using finite-index congruences for the variety $\DAb$ of finite monoids whose regular $\mathcal{D}$-classes form Abelian groups. This combinatorial description yields a normal form for general pseudowords over $\DAb$. For $(\omega - 1)$-terms, this normal form is computable, which yields an algorithm for the word problem for $(\omega - 1)$-terms of $\DAb$.
    \\\textit{Keywords.} Profinite semigroup, word problem for $(\omega - 1)$-terms, ranker
    \\\textit{MSC.} 20M07, 20M35, 68Q45
  \end{abstract}

  \begin{section}{Introduction}
    Eilenberg's famous Correspondence Theorem \cite[Theorems~VII.3.4 and VII.3.4s]{eilenberg1976automata} is the cornerstone of modern aglebraic formal language theory. It connects \emph{varieties of finite monoids} (also referred to as \emph{pseudovarieties}) and \emph{varieties of formal languages} in a one-to-one way. The former are classes of finite monoids (a corresponding notion for semigroups also exists) closed under direct products, taking submonoids and homomorphic images. The latter are -- roughly speaking -- classes of regular languages closed under Boolean operations, inverse homomorphisms and language quotients. The connection between the two sides is given by the notion of monoids recognizing formal languages: A monoid $M$ \emph{recognizes} a language $L$ over some alphabet $\Sigma$ if there is a homomorphism $\varphi: \Sigma^* \to M$ such that the preimage of $\varphi(L)$ is exactly $L$.
    The smallest such $M$ is called the \emph{syntactic monoid} of~$L$
    (see e.\,g.\ \cite{pin86:short}).
    Eilenberg's theorem now states the following correspondence:
    to a variety $\V$ of finite monoids, we attach the variety $\mathscr{V}$ of those languages which are recognized by a monoid in $\V$. The inverse of this correspondence is given by mapping a variety $\mathscr{V}$ of formal languages to the variety $\V$ of finite monoids generated by the syntactic monoids of the languages from $\mathscr{V}$.
    
    Eilenberg's Correspondence Theorem shows a general principle underlying some historically significant previously known instances the correspondence: the (not less famous) Schü\-tzen\-ber\-ger Theorem states that the class of \emph{star-free} languages corresponds to the class of \emph{aperiodic} (or \emph{group-free}) monoids \cite{schutzenberger1965finite}. Similarly, the class of \emph{piecewise-testable} languages belongs to the class of \emph{$\mathcal{J}$-trivial} monoids \cite{sim75:short} (where $\mathcal{J}$ is one of \emph{Green's relations}, see e.\,g.\ \cite{Howie95}). Since both properties, being aperiodic and being $\mathcal{J}$-trivial, can be tested for any given finite monoid, these results, in particular, yield membership algorithms for the class of star-free and the class of piecewise-testable languages: compute the syntactic monoid of a given language and test whether it is aperiodic or $\mathcal{J}$-trivial, respectively.
    
    Testing whether a given finite monoid $M$ is aperiodic can here be done by testing whether it satisfies the equation $x^\omega = x^\omega x$. The idea is that we may replace the variable $x$ by any monoid element $s \in M$ and obtain $s^\omega$, the unique idempotent in the subsemigroup generated by $s$, for the left-hand side and the product of $s^\omega$ and $s$ on the right-hand side. The monoid \emph{satisfies} the equation if the respective monoid elements obtained for the left-hand and the right-hand sides are the same for all $s \in M$ (and, since $M$ is finite, this can be answered simply by an exhaustive search).
    
    This approach applies to a more general form of equations, for which we need to introduce the free profinite monoid over a finite alphabet $\Sigma$.\footnote{%
      The research on profinite monoids (and semigroups) is too vast
      to give an appropriate overview here; we refer the reader to
      books on the topic (such as \cite{alm94:short, alm20:short, rhodes2009qTheory}) instead.}
    One way to define it, is to consider the (ultra)metric which assigns to a pair of finite words $u, v \in \Sigma^*$ the distance $2^{-|M|}$ where $M$ is a smallest monoid in which $u$ and $v$ can be mapped to different images under a homomorphism $\Sigma^* \to M$ (see e.\,g.\ \cite{alm20:short} for more details and equivalent definitions). Now, the \emph{free profinite monoid} $\freeProf$ is the completion of $\Sigma^*$ as a
metric semigroup with respect to this metric (i.\,e.\ it contains the limit points of all Cauchy sequences of finite words). This is free in the sense that every map $\Sigma \to M$ (for a finite monoid $M$) 
extends uniquely into a continuous homomorphism $\freeProf \to M$. The elements of $\freeProf$ are called \emph{pseudowords}.
    
    An \emph{equation} (of pseudowords) is now simply a pair $\alpha = \beta$ with $\alpha, \beta \in \freeProf$. It is \emph{satisfied} by a finite monoid $M$ if $\alpha$ and $\beta$ are mapped to the same monoid element under any continuous homomorphism $\freeProf \to M$ (which are all uniquely defined by their restriction into maps $\Sigma \to M$).
    It is not difficult to check that the class of finite monoids satisfying some fixed (system of) equation(s) forms a variety (of finite monoids).
    For instance, the variety of aperiodic finite monoids may be defined using the equation $x^\omega = x^\omega x$ (where $x^\omega$ is the limit point of the Cauchy sequence $(x^{k!})_k$). In fact, Reiterman's theorem \cite{reiterman1982birkhoff} states that any variety is defined by a (not necessarily finite) system of equations. A variety (of finite monoids) \emph{satisfies} an equation/a system of equations if all its members do.
    
    Algorithmically, it is important to point out that there are uncountably many pseudowords (over a fixed finite alphabet) and that they, thus, cannot be used as inputs to algorithms. However, in most natural cases, it suffices to consider a countable subset. Typically, this subset is given by \emph{$\omega$-terms} (when considering aperiodic monoids) and \emph{$(\omega - 1)$-terms} (in monoids that contain non-trivial groups). These arise as the submonoid of $\freeProf$ generated by the letters $\Sigma$ and the operation $\alpha^\omega$ (which maps a Cauchy sequence $(u_k)_k$ to the Cauchy sequence $(u_k^{k!})_k$) or $\alpha^{\omega -1}$ (mapping $(u_k)_k$ to $(u_k^{k! - 1})_k$), respectively.
    The \emph{word problem for $\omega$-terms ($(\omega - 1)$-terms)} of a variety $\V$ asks whether a given equation of $\omega$-terms ($(\omega - 1)$-terms) is satisfied by all monoids in $\V$.
    
    Notable positive results on the decidability of the word problem for $\omega$-terms include the variety $\V[J]$ of $\mathcal{J}$-trivial finite monoids \cite{Almeida1991Implicit:short, almeida2002finite}, the variety $\V[R]$ of $\mathcal{R}$-trivial finite monoids \cite{almeida2007automata} and the variety $\V[A]$ of aperiodic finite monoids \cite{mccammond2001normal} as well as the variety $\V[DA]$ of finite monoids whose regular $\mathcal{D}$-classes form aperiodic semigroups \cite{moura2011word}. Although its definition might seem esoteric at first glance, the variety $\V[DA]$ appears in many scenarios and has a lot of beautiful descriptions \cite{tesson2002diamonds}. Noteworthy in the context of the current paper is a description based on \emph{rankers}, which are sequences of instructions of the form \enquote{go to the next $a$ on the left/right} (and are based on work by Weis and Immerman \cite{weis2009structure}, turtle programs \cite{schwentick2002partially} and unambiguous interval temporal logic \cite{lodaya2008marking}). Based on these rankers, one may define a family $\{ \mathcal{C}_k \}$ of finite-index congruences over $\Sigma^*$ such that a monoid is in $\V[DA]$ if and only if it is a homomorphic image of a submonoid of a monoid $\Sigma^* / \mathcal{C}_k$ given by a congruence in the family (for some $k$). This combinatorial approach can be used to re-formulate the proofs for the decidability of the word problem for $\omega$-terms for $\V[J]$ \cite{klima2011piecewise} and $\V[DA]$ \cite{kufleitner2018wordProblem}. In fact, it also naturally matches a hierarchy of varieties covering $\V[DA]$ (also naturally related to the quantifier alternation hierarchy of two-variable first-order logic), the \emph{Trotter-Weil hierarchy}, for which the word problem for $\omega$-terms is also known to be decidable \cite{kufleitner2018wordProblem}.
    
    Despite the success of the ranker-based approach inside $\V[DA]$, it has so far not been applied to varieties containing non-trivial groups (i.\,e.\ non-aperiodic varieties). In this paper, we will extend the approach to the variety $\DAb$ of finite monoids whose regular $\mathcal{D}$-classes form Abelian groups. This variety may also be defined using the equation $(xy)^{\omega - 1} = (yx)^{\omega - 1}$ (see \autoref{lem:DAbEquation}) and we will present a family of ranker-based combinatorial congruences that describe $\DAb$ (in a similar way as sketched above for $\V[DA]$; see \autoref{def:singleRankerPair} and \autoref{def:congruence}).
    This approach allows us to define a normal form for pseudowords over $\DAb$ (in the sense that, for every $\alpha$, there is a unique $\alpha'$ in normal form such that $\alpha = \alpha'$ is satisfied by $\DAb$; see \autoref{def:normalForm}). For $(\omega - 1)$-terms, this normal form will turn out to be computable, which immediately yields a decision algorithm for the word problem for $(\omega - 1)$-terms (\autoref{cor:WPdecidable}). The existence of the normal form also shows that the ranker-based approach allows for a deeper understanding of the variety $\DAb$ in general.
  \end{section}

  \begin{section}{Preliminaries}
    \paragraph*{Fundamentals and Notation.}
    We let the set of natural numbers $\mathbb{N}$ include $0$ and, if we explicitly want to exclude $0$, we write $\Np = \{ 1, 2, \dots \}$. Two integers $x, y \in \mathbb{Z}$ are \emph{congruent modulo $m$} for some number $m \in \Np$, written as $x \equiv y \bmod m$, if there is some $z \in \mathbb{Z}$ with $x = z \cdot m + y$. The relation $\cdot \equiv \cdot \bmod m$ is a congruence and we use $x \bmod m$ to denote the smallest non-negative representative of the congruence class of $x$.
    
    \paragraph*{Words.}
    Let $\Sigma$ be an arbitrary non-empty set, i.\,e.\ an \emph{alphabet}. The set of all finite words over $\Sigma$ including the empty word $\varepsilon$ is denoted by $\Sigma^*$. If we want to exclude the empty word, we write $\Sigma^+$ for $\Sigma^+ = \Sigma \setminus \{ \varepsilon \}$. Equipped with the operation of concatenation, $\Sigma^*$ has a natural monoid structure and $\Sigma^+$ is a semigroup. For any finite word $w = a_1 \dots a_n$ with letters $a_1, \dots, a_n \in \Sigma$, let $|w| = n$ be the \emph{length} of $w$. To denote sets of finite words of certain lengths, we use natural variations of the notation $\Sigma^*$. For example, we use $\Sigma^{\leq n}$ to denote the set of finite words of length at most $n$ and $\Sigma^n$ to denote the set of finite words of length exactly $n$.
    
    For a finite word $w$, let $\dom{w} = \{ -\infty, 1, \dots, n, +\infty \}$ be the set of positions in $w$ with two additional special \enquote{positions}: $-\infty$ denotes an imaginary position \emph{before the first letter} of $w$ and $+\infty$ denotes an imaginary position \emph{after the last letter} of $w$. For a word $w = a_1 \dots a_n$ with $a_1, \dots, a_n \in \Sigma$, we let $w(p) = a_p$ for all positions $p \in \{ 1, \dots, n \}$. By $|w|_a$ for $a \in \Sigma$ and $w \in \Sigma^*$, we denote the number counting how often the letter $a$ occurs in $w$ and by $\alphabet{w}$, we denote the set of letters appearing in $w$, i.\,e.\ the set $\alphabet w = \{ a \in \Sigma \mid |w|_a > 0 \}$. 
    
    A word $u$ is a \emph{prefix} of a word $w$ if there is some $y$ such that $w = uy$. It is a \emph{proper} prefix if, additionally, $y$ is non-empty. The word $u$ is a \emph{infix} of $w$ if there are $x$ and $y$ such that $w = xuy$ and it is a suffix of $w$ if there is some $x$ with $w = xu$. A suffix is \emph{proper} if $x$ is non-empty.
    
    \paragraph*{Rankers.}
    Let $X_{\Sigma} = \{X_a \mid a \in \Sigma \}$ and $Y_{\Sigma} = \{Y_a \mid a \in \Sigma \}$ and, additionally, define the special symbols $\Xinfty \not\in X_\Sigma$ and $\Yinfty \not\in Y_\Sigma$. A finite word $\rho \in (X_{\Sigma} \cup \{ \Xinfty \})^*$ is called an \emph{$X$-ranker} (over $\Sigma$) and a finite word $\lambda \in (Y_{\Sigma} \cup \{ \Yinfty \})^*$ is called a \emph{$Y$-ranker}. An $X$-ranker can be seen as a sequence of instructions of the form \enquote{go to the next $a$ on the right} and a $Y$-ranker as a sequence of instructions of the form \enquote{go to the next $a$ on the left}. Both kinds of sequences are evaluated left to right.
    
    More formally, for a word $w$, a position $p \in \dom{w}$ in $w$ and for all $a \in \Sigma$, define
    \begin{align*}
      X_a(w; p) &= \inf \{ q \in \dom{w} \mid p < q, w(q) = a \} \text{ and}\\
      Y_a(w; p) &= \sup \{ q \in \dom{w} \mid q < p, w(q) = a \}
      \intertext{and, for the special symbols $\Xinfty$ and $\Yinfty$, let}
      \Xinfty(w; p) &= +\infty \text{ and}\\
      \Yinfty(w; p) &= -\infty \text{.}
    \end{align*}
    If there is no $a$ to the right of $p$ in $w$, we have $X_a(w; p) = +\infty$ and, symmetrically, $Y_a(w; p) = -\infty$ if there is no $a$ to the left of $p$ in $w$. The special symbols $\Xinfty$ and $\Yinfty$ are used to jump to the end or to the beginning of the word, respectively.
    
    If $\rho$ is the empty $X$-ranker, we let $\rho(w; p) = p$ and, similarly, we let $\lambda(w; p) = p$ if $\lambda$ is the empty $Y$-ranker. For longer $X$- and $Y$-rankers, we let
    \begin{align*}
      \rho X(w; p) &= X(w; \rho(w; p)) \text{ and} \\
      \lambda Y(w; p) &= Y(w; \lambda(w; p))
    \end{align*}
    recursively for $X \in X_\Sigma \cup \{ \Xinfty \}$ and $Y \in Y_\Sigma \cup \{ \Yinfty \}$ where $\rho$ is any $X$-ranker and $\lambda$ is any $Y$-ranker.
    
    For an $X$-ranker $\rho$, let $\rho(w) = \rho(w; -\infty)$ and, for a $Y$-ranker $\lambda$, let $\lambda(w) = \lambda(w; +\infty)$; i.\,e.\ $X$-rankers start on the left of a word and $Y$-rankers on the right.
    
    A position $p \in \dom{w}$ is \emph{visited} by an $X$- or $Y$-ranker $\delta$ if there is a non-empty prefix $\delta'$ of $\delta$ with $\delta'(w) = p$ and the positions visited by $\delta$ are called the \emph{$\delta$-positions} in $w$. For a set $\Delta$ of rankers, the $\Delta$-positions in $w$ is the union of all $\delta$-positions with $\delta \in \Delta$. The \emph{$\Delta$-factorization} of a word $w$ is $w = z_0 c_1 z_1 \dots c_d z_d$ if the $c_i$ are exactly at the $\Delta$-positions in $w$ (with the possible exception of $-\infty$ and $+\infty$). For a singleton set $\Delta = \{ \delta \}$, we also speak of the $\delta$-factorization (which is the same as the $\{ \delta \}$-factorization).
    
    \begin{example}
      The $\rho$-factorization of $w$ with $\rho = X_{a_1} \dots X_{a_r}$ (if $\rho(w) < +\infty$) is the unique factorization $w = w_0 a_1 w_1 \dots a_r w_r$ with $a_i \not\in \alphabet{w_{i - 1}}$ (for $i = 1, \dots, r$) and the definition for $Y_{b_1} \dots Y_{b_l}$ is symmetric.
    \end{example}
    
    If the positions visited by an $X$-ranker $\rho$ and the positions visited by a $Y$-ranker $\lambda$ have the same ordering in some word $u$ as they have in some other word $v$, we say that $u$ and $v$ are $\rho$-$\lambda$-compatible. More formally, we define:
    \begin{definition}
      Let $\rho$ be an $X$-ranker and $\lambda$ a $Y$-ranker. We define $\operatorname{ord}[\rho, \lambda; w]$ for a finite word $w$ as one of the elements in $\{ <, =, > \}$ depending on whether we have $\rho(w) < \lambda(w)$, $\rho(w) = \lambda(w)$ or $\rho(w) > \lambda(w)$.
      
      Then, two finite words $u$ and $v$ are \emph{$\rho$-$\lambda$-compatible} if, for all prefixes $\rho'$ of $\rho\Xinfty$ and $\lambda'$ of $\lambda\Yinfty$, we have $\operatorname{ord}[\rho', \lambda'; v] = \operatorname{ord}[\rho', \lambda'; v]$.
    \end{definition}
    \noindent{}Note that we can, in particular, use empty rankers and the entire ranker $\rho\Xinfty$ or $\lambda\Yinfty$ to obtain the positions $-\infty$ and $+\infty$ in the comparison. Therefore, we have $\rho'(u) < +\infty$ if and only if we have $\rho'(v) < +\infty$ (and the same holds for $\lambda'$ with respect to $-\infty$). It is also easy to see that the relation of being $\rho$-$\lambda$-compatible is an equivalence.
    
    The idea is that the $\{ \rho, \lambda \}$-factorizations of two $\rho$-$\lambda$-compatible words have the same form: if the $\{ \rho, \lambda \}$-factorization of $u$ is $u = x_0 c_1 x_1 \dots c_d x_d$ and $v$ is $\rho$-$\lambda$-compatible to $u$, then the $\{ \rho, \lambda \}$-factorization of $v$ is also of the form $v = y_0 c_1 y_1 \dots c_d y_d$.
    
    Moreover, we can manipulate the factors $x_i$ of $u$ arbitrarily to obtain a word which is $\rho$-$\lambda$-compatible to $u$ as long as we do not increase the alphabet of the $x_i$ (see \autoref{fig:rhoLambdaCompatibleIllustration} or an illustration):
    \begin{fact}\label{fct:compatibleIfSmallerAlphabet}
      Let $\rho$ be an $X$-ranker, $\lambda$ be a $Y$-ranker and let $u = x_0 c_1 x_1 \dots c_d x_d \in \Sigma^*$ be given in its $\{ \rho, \lambda \}$-factorization. Furthermore, let $\tilde{x}_1, \dots, \tilde{x}_d \in \Sigma^*$ with $\alphabet \tilde{x}_i \subseteq \alphabet x_i$ for all $0 \leq i \leq d$. Then, $\tilde{u} = \tilde{x}_0 c_1 \tilde{x}_1 \dots c_d \tilde{x}_d$ is the $\{ \rho, \lambda \}$-factorization of $\tilde{u}$ and $\tilde{u}$ is $\rho$-$\lambda$-compatible to $u$.
    \end{fact}
  
    \begin{figure}\centering
      \begin{tikzpicture}
        \node[inner sep=0pt] (wLabel) {$u ={}$};
        \matrix [base right=0cm of wLabel, rectangle, draw, matrix of math nodes, ampersand replacement=\&, inner sep=2pt, text height=.8em, text depth=0pt] (w) {
          \makebox[1.5cm]{\textcolor{gray}{$x_0$}} \&
          a \&
          \makebox[2cm]{\textcolor{gray}{$x_1$}} \&
          b \&
          \makebox[0.75cm]{\textcolor{gray}{$x_2$}} \&
          c \&
          \makebox[1.0cm]{\textcolor{gray}{$x_3$}} \&
          b \&
          \makebox[0.75cm]{\textcolor{gray}{$x_4$}} \&
          a \&
          \makebox[1cm]{\textcolor{gray}{$x_5$}} \&
          a \&
          \makebox[1.75cm]{\textcolor{gray}{$x_6$}} \&
          \\
        };
        \foreach \x in {1,...,12}
          \draw ($(w-1-\x.east |- w.north)$) -- ($(w-1-\x.east |- w.south)$);
        
        \node[draw, below left=1.5cm and 0cm of w-1-5] (z2') {\makebox[1.5cm]{\textcolor{gray}{$\tilde{x}_2$}}};
        \draw[<->, very thick, color=gray, decorate, decoration={penciline}] (z2'.north) -- (w-1-5.south);
        
        \draw[->, shorten <= 1ex] ($(w.north west)+(0ex, 0.75)$) .. controls +(0.25cm, 0cm) and +(0cm,1cm) .. node[midway, above] {$X_a$} ($(w-1-2.north)+(0ex, 1ex)$);
        \draw[->, shorten <= 1ex] ($(w-1-2.north)+(0ex, 1ex)$) .. controls +(0.25cm, 0.75cm) and +(0cm,1cm) .. node[midway, above] {$X_b$} ($(w-1-4.north)+(0ex, 1ex)$);
        \draw[->, shorten <= 1ex] ($(w-1-4.north)+(0ex, 1ex)$) .. controls +(0.25cm, 0.75cm) and +(0cm,1cm) .. node[midway, above] {$X_a$} ($(w-1-10.north)+(0ex, 1ex)$);
        
        \draw[->, shorten <= 1ex] ($(w.south east)+(0ex, -0.75)$) .. controls +(-0.25cm, 0cm) and +(0cm,-1cm) .. node[midway, below] {$Y_a$} ($(w-1-12.south)+(0ex, -1ex)$);
        \draw[->, shorten <= 1ex] ($(w-1-12.south)+(0ex, -1ex)$) .. controls +(-0.25cm, -0.75cm) and +(0cm,-1cm) .. node[midway, below] {$Y_b$} ($(w-1-8.south)+(0ex, -1ex)$);
        \draw[->, shorten <= 1ex] ($(w-1-8.south)+(0ex, -1ex)$) .. controls +(-0.25cm, -0.75cm) and +(0cm,-1cm) .. node[midway, below] {$Y_c$} ($(w-1-6.south)+(0ex, -1ex)$);
        \draw[->, shorten <= 1ex] ($(w-1-6.south)+(0ex, -1ex)$) .. controls +(-0.25cm, -0.75cm) and +(0cm,-1cm) .. node[midway, below] {$Y_b$} ($(w-1-4.south)+(0ex, -1ex)$);
      \end{tikzpicture}
      \caption{Illustration of \autoref{fct:compatibleIfSmallerAlphabet}: If we replace $x_2$ by $\tilde{x}_2$ (for $\alphabet \tilde{x}_2 \subseteq \alphabet x_2$), the result will be $\rho$-$\lambda$-compatible to $u$ (for $\rho = X_a X_b X_a$ and $\lambda = Y_a Y_b Y_c Y_b$ in out example).}\label{fig:rhoLambdaCompatibleIllustration}
    \end{figure}
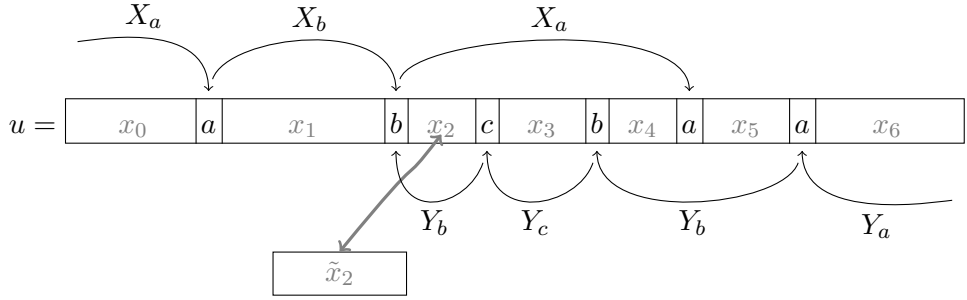
    
    \paragraph*{Semigroups, Monoids and Congruences.}
    Semigroups and monoids are the main object of this work, so we assume the reader to be familiar with this area. Details may be found in any introductory text book on the topic (such as \cite{Howie95}). Most of the time, the theory of semigroups is very similar to the one of monoids. Out of convenience, will usually work with monoids. Here, we denote the neutral element of a monoid by $1$.
    
    Congruences play an important role in (semigroup and) monoid theory. If $\equiv$ is a congruence over some monoid $M$, then the set $M / {\equiv}$ of congruences classes has a natural monoid structure itself (by letting $[s]_\equiv \cdot [t]_\equiv = [st]_\equiv$), whose neutral element is the congruence class of the neutral element of $M$. Often, we will consider monoids of the form $\Sigma^* / {\equiv}$ for some alphabet $\Sigma$ and some congruence $\equiv$ over $\Sigma^*$. Sometimes we will write $u = v$ \emph{in $M$} for $u \equiv v$ in this case for $u, v \in \Sigma^*$. Another common notation here is to write $\langle a_1, \dots, a_m \mid \ell_1 = r_1, \dots, \ell_n = r_n \rangle_{M}$ for the monoid $\Sigma^* / {\equiv}$ where $\Sigma = \{ a_1, \dots, a_m \}$ and $\equiv$ is the smallest congruence with $\ell_i \equiv r_i$ for all $1 \leq i \leq n$.
    
    We are mostly concerned with finite monoids. In such a monoid $M$, every element $s \in M$ has a proper idempotent power $s^E$ (i.\,e.\ $E > 0$ and $s^E \cdot s^E = s^E$). The smallest such power $E$ cannot be larger than $|M|$ (for reasons of cardinality) and we therefore have that $E$ divides $|M|!$. In particular, we have that $s^{|M|!}$ is idempotent for every $s \in M$. The number $|M|!$ is called the \emph{exponent} of $M$ and we will simply denote it by $M!$ in order to lighten our notation a bit.
    
    A \emph{monogenic} monoid is a monoid generated by a single element. All finite monogenic monoids are of the form $\langle a \mid a^I = a^{I + P} \rangle_M = \{ a^0, \dots, a^I, \dots a^{I + P - 1} \}$ for some $I \geq 0$ and $P > 0$. Graphically, it has a \emph{$\rho$-shape} where the elements $a^0, \dots, a^{I - 1}$ are on the tail and the elements $a^{I}, \dots, a^{I + P - 1}$ are on a cycle. In such a monoid, the smallest exponent $E > 0$ such that $a^E$ is idempotent is the unique $I \leq E < P + I$ which is a multiple of $P$. In particular, $s^{M!}$ is always on the cycle.
    
    The idempotent power $s^{M!}$ plays a central role in the theory of (semigroups and) monoids. In the context of this work, however, there is an almost equally important (and closely related) power of a monoid element $s$: it is $s^{M! - 1}$. We clearly have $s^{M! - 1} s = s^{M!}$ and it turns out that $s^{M! - 1}$ is the inverse element of $s^{M!} s$ in the group formed by the elements on the cycle of the monogenic submonoid generated by $s$. This follows from the following statement, that we will often make implicitly in the following.
    \begin{fact}
      For every element $s \in M$ of a finite monoid $M$, we have
      \[
        s^{M!} s^{M! - 1} = s^{M! - 1} = s^{M! - 1} s^{M!}
      \]
    \end{fact}
    \begin{proof}
      Since the statement must only hold in the monogenic subsemigroup of $M$ generated by $s$, it suffices to consider finite monogenic monoids, which are of the form $M = \langle s \mid s^I = s^{I + P} \rangle_M$. There is no monoid of size $0$ and the statement clearly holds in the monoid of size $1$. There are only two monoids of size $2$: the cyclic group of order two $\langle s \mid 1 = s^2 \rangle_M$ and the monoid $\langle s \mid s = s^2 \rangle_M$ (sometimes denote by $U_1$). In both of them, we have $s^2 s^1 = s^1$ and the statement holds.
      
      Thus, the only interesting monoids are the monogenic ones with $|M| \geq 3$. Here, we have $|M|! - 1 > |M|$, which shows that $s^{M! - 1}$ belongs to the elements on the cycle. Note that the cycle has length $P$ and that we have $s^n s^P = s^n$ for all elements $s^n$ on the cycle (i.\,e.\ those with $n \geq I$). Since $|M|!$ is a multiple of $P$, we obtain that multiplying by $s^{M!}$ does not change the elements on the cycle either and the statement follows.
    \end{proof}
    
    \paragraph*{Green's Relations.}
    Every monoid admits five relations which are very important for its structure. These are \emph{Green's Relations} and are defined as follows:\footnote{For details and proofs, we again refer to standard text books in the subject such as \cite{Howie95}.} we have
    \begin{itemize}[noitemsep]
      \item $s \L t \iff \exists x, x' \in M: xs = t \text{ and } s = x't$,
      \item $s \R t \iff \exists y, y' \in M: sy = t \text{ and } s = ty'$,
      \item $s \J t \iff \exists x, x', y, y' \in M: xsy = t \text{ and } s = x'ty'$,
      \item $s \H t \iff s \L t \text{ and } s \R t$ and
      \item $s \mathrel{\mathcal{D}} t \iff \exists r_1, \dots, r_n \in M: s \L r_1 \R r_2 \dots \L r_{n - 1} \R r_n \L t$
    \end{itemize}
    for two elements $s, t$ of a monoid $M$. All these relations are equivalences and, in finite monoids (and also most other monoids we deal with), the relations $\mathcal{D}$ and $\mathcal{J}$ coincide and can, thus, be used interchangeably. If a monoid is even a group, then all of the above relations have only one class containing the entire group. On the other hand, the elements of an $\mathcal{H}$-class containing an idempotent form a group and, thus, every $\mathcal{H}$-class contains at most one idempotent. A $\mathcal{D}$-class is called \emph{regular} if it contains an idempotent. A monoid element $s$ is called \emph{regular} if its $\mathcal{D}$-class is regular.
    
    \paragraph*{$\mathcal{R}$- and $\mathcal{L}$-Factorizations.}
    The \emph{$\mathcal{R}$-factorization} of a word $w \in \Sigma^*$ with respect to a homomorphism $\varphi: \Sigma^* \to M$ for a monoid $M$ is the (unique) factorization $w = w_0 a_1 w_1 \dots a_r w_r$ with
    \begin{itemize}[noitemsep, topsep=0pt]
      \item $1 \R \varphi(w_0)$, 
      \item $\varphi( w_0 a_1 w_1 \dots a_i ) \R \varphi( w_0 a_1 w_1 \dots a_i w_i )$ for $0 < i \leq r$ and
      \item $\varphi( w_0 a_1 w_1 \dots a_i w_i ) \not\R \varphi( w_0 a_1 w_1 \dots a_i w_i a_{i + 1} )$ for $0 \leq i < r$.
    \end{itemize}
    The \emph{positions} of such an $\mathcal{R}$-factorization are the positions of the $a_i $ in $w$ and we say that $w$ has an \emph{$\mathcal{R}$-decent} at these positions.
    
    Symmetrically, the \emph{$\mathcal{L}$-factorization} of $w$ with respect to $\varphi$ is the (unique) factorization $w = w_0 a_1 w_1 \dots a_\ell w_\ell$ with
    \begin{itemize}[noitemsep, topsep=0pt]
      \item $\varphi(w_\ell) \L 1$, 
      \item $\varphi( w_{i - 1} a_i w_i \dots a_\ell w_\ell ) \L \varphi( a_i w_i \dots a_\ell w_\ell )$ for $0 < i \leq \ell$ and
      \item $\varphi( a_i w_i \dots a_\ell w_\ell ) \not\L \varphi( w_i \dots a_\ell w_\ell )$ for $0 < i \leq \ell$.
    \end{itemize}
    Again, the \emph{positions} of such an $\mathcal{L}$-factorization are the positions of the $a_i $ in $w$ and we say that $w$ has an \emph{$\mathcal{L}$-decent} at these positions.
    
    \paragraph*{The Free Profinite Monoid, Pseudowords and Equations.}
    For two words $u, v \in \Sigma^*$, let $d(u, v) = 0$ if $u = v$ and, for $u \neq v$, let $d(u, v) = 2^{-|M|}$ where $M$ is the smallest monoid such that there is a homomorphism $\varphi: \Sigma^* \to M$ with $\varphi(u) \neq \varphi(v)$. It is easy to see that this defines an ultrametric on $\Sigma^*$, called the \emph{profinite metric}. It allows us to consider the topological completion $\freeProf$ of $\Sigma^*$ (i.\,e.\ the topological space in which every Cauchy sequence over $\Sigma^*$ converges). In this space, every element $\alpha \in \freeProf$ is a limit of a Cauchy sequence $(u_k)_{k \in \mathbb{N}}$, which we simply write as $\lim (u_k)_k$. We may endow $\freeProf$ with a continuous binary operation by letting $\lim (u_k)_k \cdot \lim (v_k)_k = \lim (u_k v_k)_k$. This results in the \emph{free profinite monoid} over $\Sigma$, whose elements are called \emph{pseudowords}. For details on (free) profinite monoids and pseudowords, we refer the reader to books on the topic (e.\,g.\ \cite{alm94:short, alm20:short}).
    
    Every homomorphism $\varphi: \Sigma^* \to M$ to a finite monoid $M$ uniquely extends to a continuous homomorphism\footnote{In fact, a homomorphism $\freeProf \to M$ is already uniquely defined by the images of the letters from $\Sigma$.} $\varphi: \freeProf \to M$ (where we endow $M$ with the discrete topology). In fact, for $\freeProf \ni \alpha = \lim (u_k)_k$, there is some $k_0$ such that $\varphi(u_{k_0}) = \varphi(u_{k_0 + 1}) = \dots$ (since $(u_k)_k$ is a Cauchy sequence with respect to the above metric) and we let $\varphi(\alpha) = \varphi(u_{k_0})$. In particular, we may extend the notion of the alphabet of a (finite) word to pseudowords by considering $\alphabet$ as a homomorphism $\Sigma^* \to 2^\Sigma$ to the monoid of subsets of $\Sigma$ with union as operation.
    
    Two pseudowords $\alpha, \beta$ form an \emph{equation} $\alpha = \beta$. A monoid $M$ \emph{satisfies} the equation $\alpha = \beta$ with $\alpha, \beta \in \freeProf$, which we write as $M \models \alpha = \beta$, if we have $\varphi(\alpha) = \varphi(\beta)$ for all homomorphisms $\varphi: \Sigma^* \to M$ (and, thus, all continuous homomorphisms $\varphi: \freeProf \to M$). A class $\mathcal{C}$ of monoids \emph{satisfies} an equation, which we write as $\mathcal{C} \models \alpha = \beta$, if every one of its elements satisfies it.
    
    \paragraph*{$\omega$- and $(\omega - 1)$-Powers, $(\omega - 1)$-Terms.}
    There are two very important operations on pseudwords, which we will need extensively in the following. The first one is the \emph{$\omega$-power}: for a pseudoword $\alpha = \lim (u_k)_k$, let $\alpha^\omega = \lim (u_k^{k!})_k$ (which is well-defined since $(u_k^{k!})_k$ can easily be shown to be a Cauchy sequence). Since, in a finite monoid $M$, $s^{M!}$ is idempotent for every $s \in M$, we have $\varphi\left( \alpha^\omega \right) = \varphi(\alpha)^{M!}$ for all continuous homomorphisms $\varphi: \freeProf \to M$.
    
    Closely related is the $(\omega - 1)$-power, where we let $\alpha^{\omega - 1} = \lim (u_k^{k! - 1})_k$ for a pseudoword $\alpha = \lim (u_k)_k$. Note that we have $\varphi(\alpha^{\omega - 1}) = \varphi(\alpha)^{M! - 1}$ for all continuous homomorphisms $\varphi: \freeProf \to M$ to a finite monoid $M$.
    
    An \emph{$(\omega - 1)$-term} is a pseudoword which arises from the letters in $\Sigma$ by only using concatenation and the $(\omega - 1)$-power.
    
    \paragraph*{Varieties and $\DAb$.}
    A \emph{variety of finite monoids} is a class of finite monoids closed under direct products, taking submonoids and homomorphic images. We will mostly work with varieties of finite monoids out of convenience and will simply call them \emph{varieties}. However, our results also work for the corresponding \emph{varieties of finite semigroups} (which are also closed under direct products and homomorphic images but we need to allow taking subsemigroups this time).
    
    The most important variety for this work is the variety $\DAb$, which consists of all finite monoids whose regular $\mathcal{D}$-classes form Abelian groups. One may verify directly that this is indeed a variety but this also follows by using the following defining relation.
    \begin{proposition}\label{lem:DAbEquation}
      A finite monoid $M$ is in $\DAb$ if and only if it satisfies $(xy)^{\omega - 1} = (yx)^{\omega - 1}$.
      Therefore, any monoid $M \in \DAb$ also satisfies\footnote{On a side note, we point out that $(xy)^{\omega} = (yx)^{\omega}$ is the defining relation for $\V[DG]$, the variety of all monoids whose regular $\mathcal{D}$-classes are monoids.} $(xy)^{\omega} = (yx)^{\omega}$.
    \end{proposition}
    \begin{proof}
      We first show that all monoids in $\DAb$ satisfy the equation. For this, let $M \in \DAb$ and $s, t \in M$ be arbitrary. We have to show
      \[
        (st)^{M! - 1} = (ts)^{M! - 1} \textbf{.}
      \]
      First, however, observe that we have
      \[
        (st)^{M!} \L t (st)^{M!} = (ts)^{M!} t \R (ts)^{M!}
      \]
      and that, thus, $(st)^{M!}$ and $(ts)^{M!}$ belong to the same $\mathcal{D}$-class, which must be regular (as both of the mentioned elements are idempotent). Since $M$ is in $\DAb$ this regular $\mathcal{D}$-class must be a group, which may only contain one idempotent, and we obtain
      \[
        (st)^{M!} = (ts)^{M!} \text{.}
      \]
      We use this to show $(st)^{M!} st = ts (st)^{M!}$. We have
      \[
        (st)^{M!} s \R (st)^{M!} = (ts)^{M!} \R (ts)^{M!} t = (st)^{M!} t
      \]
      and, thus, that $(st)^{M!} s$, $(st)^{M!}$ and $(st)^{M!} t$ all belong to the same regular $\mathcal{D}$-class and, therefore, commute. This yields
      \begin{align*}
        (st)^{M!} st &= (st)^{M!} \left( (st)^{M!} s \right) t = \left( (st)^{M!} s \right) \left( (st)^{M!} t \right) = \left( (st)^{M!} t \right) \left( (st)^{M!} s \right) \\
        &= (ts)^{M!} t (st)^{M!} s = ts (ts)^{2 \cdot M!} = ts (ts)^{M!} = ts (st)^{M!}
      \end{align*}
      and, thus,
      \begin{align*}
        (st)^{M! - 1} = (st)^{M! - 1} (st)^{M!} &= (ts) (st)^{M! - 2} (st)^{M!} = \dots = (ts)^{M! - 1} (st)^{M!} \\
        &= (ts)^{M! - 1} (ts)^{M!} = (ts)^{M! - 1} \text{.}
      \end{align*}
    
      For the converse, let $M$ be a finite monoid satisfying $(xy)^{\omega - 1} = (yx)^{\omega - 1}$ and, thus, also
      \[
        (xy)^\omega = \left( (xy)^{\omega - 1} \right)^{\omega - 1} = \left( (yx)^{\omega - 1} \right)^{\omega - 1} = (yx)^\omega \textbf{.}
      \]
      We first show that every regular $\mathcal{D}$-class is a group by showing that any $\mathcal{D}$-class contains at most one idempotent (which implies that, if it is regular, it consists of only one $\mathcal{R}$-class and only one $\mathcal{L}$-class \cite[Proposition~2.3.2]{Howie95} and, therefore, consists of a single $\mathcal{H}$-class, which forms a group as it contains an idempotent \cite[Corollary~2.2.6]{Howie95}). Let $e$ and $f$ be idempotents in the same $\mathcal{D}$-class. There must be some $s$ with $e \R s \L f$. Since $e$ and $s$ belong to the same $\mathcal{R}$-class and $e$ is idempotent, we get $es = s$ (as there must be some $x$ with $ex = s$ and, thus, $es = e(ex) = ex = s$). By Green's Lemma (see e.\,g.\ \cite[Lemma~2.2.1]{Howie95}), we obtain that right multiplication with $s$ is an $\mathcal{R}$-class preserving bijection form the $\mathcal{L}$-class of $e$ onto the $\mathcal{L}$-class of $f$. Thus, there is some $t$ with $ts = f$ and $t \R f$ (we can choose the pre-image of $f$ under this bijection since $s \L f$). Since $f$ is idempotent, the latter implies $ft = t$ and, again by Green's Lemma, that right multiplication with $t$ is an $\mathcal{R}$-class preserving bijection form the $\mathcal{L}$-class of $f$ to the $\mathcal{L}$-class of $e$. Thus, we have $st \R e$ and $st \L e$ or, in other words, $st \H e$. Additionally, $st$ is also idempotent (since we have $stst = sft = st$) and, thus, must be equal to $e$ (because there can only be one idempotent in every $\mathcal{H}$-class). This means that we have
      \[
        e = st = (st)^{M!} = (ts)^{M!} = ts = f
      \]
      where we used the identity $(xy)^\omega = (yx)^\omega$ in the middle.
      
      We have shown that every regular $\mathcal{D}$-class is a group and it remains to show that it is commutative. Let $g$ be an element of a regular $\mathcal{D}$-class $D$. Let $p$ be the order of $g$ as a group element. Clearly, $p$ divides $|D| \leq |M|$ and, thus, it divides $M!$. Therefore, $g^{M! - 1}$ is the inverse $g^{-1}$ of $g$ in the group $D$. For two group elements $g, h \in D$, we obtain
      \[
        (gh)^{-1} = (gh)^{M! - 1} = (hg)^{M! - 1} = (hg)^{-1}
      \]
      (where we have used the identity $(xy)^{\omega - 1} = (yx)^{\omega - 1}$ in the middle) and, thus, also $gh = hg$.
    \end{proof}
  
    We conclude this section by proving some results for $\DAb$, which we will need later on. First, we will show that the only possibility to leave an $\mathcal{H}$-class is by using a new letter.
    \begin{lemma}\label{lem:groupByAlphabet}
      Let $u, v \in \Sigma^*$ with $\alphabet v \subseteq \alphabet u$, $M \in \DAb$ and $\varphi: \freeProf \to M$ a continuous homomorphism. Then, we have $\varphi(u^{\omega}) \H \varphi(u^{\omega} v) = \varphi(u^{\omega} v u^{\omega}) = \varphi(v u^{\omega})$.
    \end{lemma}
    \begin{proof}
      For simplicity, we identify words from $\freeProf$ with their image under $\varphi$ in $M$.
      
      First, we handle the case $v = a \in \Sigma$. As we have $a \in \alphabet u$, we can write $u = u_0 a u_1$ for some $u_1, u_2 \in \Sigma^*$. Then, we have $u^{\omega} = (a u_1 u_0)^{\omega}$ by \autoref{lem:DAbEquation} and, therefore, trivially
      \[
        u^{\omega} a = (a u_1 u_0)^{\omega} a \R (a u_1 u_0)^{\omega} = u^{\omega} \text{.}
      \]
      As the $\mathcal{D}$-class of $u^{\omega}$ is trivially regular (since $u^{\omega}$ itself is an idempotent) and, therefore, a (commutative) group, this implies $u^{\omega} a \H u^{\omega}$ and, symmetrically, $u^{\omega} \H a u^{\omega}$. Since $u^{\omega}$ is idempotent and, thus, the neutral element of the group, we obtain
      \[
        u^{\omega} a = u^{\omega} a u^{\omega} = a u^{\omega}
      \]
      (where the equation on the right is symmetrical to the one on the left).
      
      For $v = a_1 \dots a_n$ with $a_1, \dots, a_n \in \Sigma$, we have $u^\omega \H u^\omega a_1$ by the above and, in particular, also $u^\omega \L u^\omega  a_1$. This yields $u^\omega a_2 \L u^\omega  a_1 a_2$, which, in the same way as above, implies $u^\omega a_2 \H u^\omega a_1 a_2$ and, thus, $u^\omega \H u^\omega a_1 a_2$. Iterating this, we end up with $u^\omega \H u^\omega v$. Since $u^\omega$ still is the neutral element of the group forming this $\mathcal{H}$-class, we also have $u^{\omega} v = u^{\omega} v u^{\omega} = v u^{\omega}$.
    \end{proof}
  
    This allows us to connect the $\mathcal{R}$-descents in a word to an $X$-ranker. Note that, due to symmetry, we immediately also get a symmetric dual connecting $\mathcal{L}$-descents to a $Y$-ranker.
    \begin{lemma}\label{lem:rankerVisitingFactorizationPositions}
      Let $M \in \DAb$ and let $w = w_0 a_1 w_1 a_2 w_2 \dots a_r w_r$ be the $\mathcal{R}$-factorization of $w \in \Sigma^*$ with respect to a homomorphism $\varphi: \Sigma^* \to M$. Then we have $a_i \not\in \alphabet w_{i - 1}$ for $i = 1, \dots, r$. Thus, the $X$-ranker $\rho = X_{a_1} X_{a_2} \dots X_{a_r}$ visits exactly  the positions of the $\mathcal{R}$-factorization and the $\mathcal{R}$-factorization coincides with the $\rho$-factorization. 
    \end{lemma}
    \begin{proof}
      Suppose we have $a_i \in \alphabet w_{i - 1}$ for some $1 \leq i \leq r$. We may uniquely extend $\varphi$ to a homomorphism $\freeProf \to M$ and identify pseudowords with their image under this extension in $M$. To lighten the notation even further, let $a_i = a$, $u = w_0 a_1 \dots w_{i - 2} a_{i - 1}$ and $v = w_{i - 1}$. In other words, we have $a \in \alphabet v$ and $u \R uv$ but $u v \not\R u v a$. From the former follows that there is some $x$ with $u = uvx$ and, by iterating this equation, $u = u(vx)^\omega$. Thus, we have
      \[
      u v = u (vx)^\omega v = u v (xv)^\omega \H u v (xv)^\omega a = u (vx)^\omega v a = u v a
      \]
      (where we have used \autoref{lem:groupByAlphabet} for the $\mathcal{H}$-equivalence). This constitutes a contradiction as we must have $uv \not\R uva$.
    \end{proof}
  
    The same connection between the alphabet of a word and $\mathcal{H}$-classes also allows us to make the following statements about permuting parts of a word without changing the attached monoid element.
    \begin{fact}\label{fct:permuteInHClass}
      Consider a monoid $M \in \DAb$ and a continuous homomorphism $\varphi: \freeProf \to M$ and let $u, v, \tilde{v} \in \Sigma^*$ with $\alphabet v \subseteq \alphabet u$ such that $\tilde{v}$ is a permutation of $v$. Then we have:
      \[
        \varphi(u^\omega v) = \varphi(u^\omega \tilde{v})
      \]
    \end{fact}
    \begin{proof}
      Let $v = a_1 \dots a_n$ for $a_1, \dots, a_n \in \Sigma$ and let $\tilde{v} = a_{\sigma(1)} \dots a_{\sigma(n)}$ for some permutation $\sigma$ of $\{ 1, \dots, n \}$ and, for all $\alpha \in \freeProf$, identity $\alpha$ with its image $\varphi(\alpha)$ in $M$.
      
      By \autoref{lem:groupByAlphabet}, we have $u^\omega \H u^\omega a_i$ for all $i = 1, \dots, n$. This $\mathcal{H}$-class contains an idempotent and, therefore, lies in a regular $\mathcal{D}$-class. Since the regular $\mathcal{D}$-class forms an Abelian group, it only consists of this single $\mathcal{H}$-class, which, thus, is an Abelian group whose neutral element is $u^\omega$. Therefore, we have
      \[
        u^\omega v = u^\omega (a_1 \dots a_n) = (u^\omega a_1) \dots (u^\omega a_n) = (u^\omega a_{\sigma(1)}) \dots (u^\omega a_{\sigma(n)}) = u^\omega \tilde{v} \textbf{.}\qedhere
      \]
    \end{proof}
  
    \begin{fact}\label{fct:permuteR}
      Consider a monoid $M \in \DAb$ and a homomorphism $\varphi: \Sigma^* \to M$ and let $u, v, \tilde{v} \in \Sigma^*$ such that $\tilde{v}$ is a permutation of $v$. Then we have:
      \[
        \varphi(u) \R \varphi(uv) \implies \varphi(uv) = \varphi(u \tilde{v})
      \]
    \end{fact}
    \begin{proof}
      Let $v = a_1 \dots a_n$ for $a_1, \dots, a_n \in \Sigma$ and let $\tilde{v} = a_{\sigma(1)} \dots a_{\sigma(n)}$ for some permutation $\sigma$ of $\{ 1, \dots, n \}$.
      
      For simplicity, we identity any pseudoword with its image in $M$ given by $\varphi$. Since we have $u \R uv$ (and since we can extend $\Sigma$ appropriately), there is some $x \in \Sigma^*$ with $u = uvx$. Iterating this yields $u = u (vx)^\omega$ and, by \autoref{fct:permuteInHClass}, we obtain
      $
        uv = u (vx)^\omega v = u (vx)^\omega \tilde{v} = u \tilde{v}\textbf{.}
      $
    \end{proof}
  
    In addition to permuting parts of a word, we may also add or remove blocks of the form $a^{M!}$ in certain positions without changing the monoid element. This is formalized in the following fact.
    \begin{fact}\label{fct:insertAM}
      Let $M \in \DAb$ be a monoid, $\varphi: \Sigma^* \to M$ be a homomorphism and $u, v \in \Sigma^*$ with $\varphi(u) \R \varphi(uv)$. Furthermore, let $m$ be a multiple of the orders of all groups in $M$. Then, we have
      \begin{enumerate}
        \item $\varphi(uv) = \varphi(u \, v_0 a^m v_1)$ if $v = v_0 v_1$ and
        \item $\varphi(uv) = \varphi(u \, v_0 v_1)$ if $v = v_0 a^m v_1$
      \end{enumerate}
      for $a \in \alphabet v$ and $v_0, v_1 \in \Sigma^*$.
    \end{fact}
    \begin{proof}
      We again identify pesudowords with their image in $M$ given by $\varphi$. As in the proof of \autoref{fct:permuteR}, there is some $x \in \Sigma^*$ with $u = u (vx)^\omega$ and, by \autoref{lem:groupByAlphabet}, we have $(vx)^\omega \H (vx)^\omega a$. Thus (since $(vx)^\omega$ is the neutral element of the group formed by this $\mathcal{H}$-class) we have $(vx)^\omega a = (vx)^\omega a (vx)^\omega$, which yields
      \[
        (vx)^\omega = \big( (vx)^\omega a \big)^m = (vx)^\omega a^m
      \]
      since the order of $(vx)^\omega a$ as a group element divides $m$. Therefore, we obtain
      \begin{enumerate}
        \item $uv = u (vx)^\omega v = u (vx)^\omega a^m v = u (vx)^\omega v_0 a^m v_1 = u v_0 a^m v_1$ if $v = v_0 v_1$ (where we have used \autoref{fct:permuteR} in the second last step) and
        \item $uv = u \, v_0 a^m v_1 = u \, a^m v_0 v_1 = u (vx)^\omega a^m v_0 v_1 = u (vx)^\omega v_0 v_1 = u v_0 v_1$ if $v = v_0 a^m v_1$ (where we have used \autoref{fct:permuteR} in the second step). \qedhere
      \end{enumerate}
    \end{proof}
  
    \autoref{fct:permuteR} and \autoref{fct:insertAM} will turn out to be central for the proof in the next section. It is important to point out that for both statements we also have left-right symmetric versions.
  \end{section}

  \begin{section}{A Congruence for $\DAb$}
    In this section, we will introduce a congruence of a combinatorial nature and show that it can be used to describe the finite monoids (and semigroups) in $\DAb$.
    
    First, we introduce some terminology on the factorization of a word given by an $X$- and a $Y$-ranker.
    \begin{definition}\label{def:rhoLambdaSplit}
      Let $w \in \Sigma^*$ be a finite word, $\rho$ an $X$-ranker and $\lambda$ a $Y$-ranker with $\rho(w) \leq \lambda(w)$. The \emph{$\rho$-$\lambda$-split} of $w$ is $(w_0, w_1, w_2)$ where
      \begin{itemize}[noitemsep]
        \item $w_0$ is the prefix of $w$ containing the positions smaller than $\rho(w)$,
        \item $w_1$ is the infix of $w$ containing the positions greater than $\rho(w)$ but smaller than $\lambda(w)$ and
        \item $w_2$ is the suffix of $w$ containing the positions greater than $\lambda(w)$.
      \end{itemize}
      We say that $w_0$ is the \emph{left} part, $w_1$ is the \emph{middle} part and $w_2$ is the \emph{right} part of the split.
      
      Furthermore, if $(w_0, w_1, \varepsilon)$ is the $\rho$-$\varepsilon$-split of $w$, we say that $(w_0, w_1)$ is the \emph{$\rho$-split} of $w$, whose \emph{left} part is $w_0$ and whose \emph{right} part is $w_1$.
      Symmetrically, if $(\varepsilon, w_1, w_2)$ is the $\varepsilon$-$\lambda$-split of $w$, we call $(w_1, w_2)$ the \emph{$\lambda$-split} of $w$ and refer to $w_1$ as its \emph{left} part and to $w_2$ as its \emph{right} part.
    \end{definition}
    \begin{remark*}
      \autoref{def:rhoLambdaSplit} covers quite a few cases (which are depicted in \autoref{fig:rhoLambdaSplit}):
      \begin{enumerate}[label=(\alph*)]
        \item $-\infty = \rho(w) = \lambda(w) < +\infty$:\\
          The left part and the middle part are both empty ($w_0 = w_1 = \varepsilon$) and the right part spans the whole word ($w_2 = w$).
        \item $-\infty = \rho(w) < \lambda(w) < +\infty$:\\
          The left part is empty ($w_0 = \varepsilon$) and we have $w = w_1 b w_2$ where $b$ is the label of the position $\lambda(w)$ in $w$.
        \item $-\infty = \rho(w) < \lambda(w) = +\infty$:\\
          The left part and the right part are both empty ($w_0 = w_2 = \varepsilon$) and the middle part spans the whole word ($w_1 = w$).
        \item $-\infty < \rho(w) = \lambda(w) < +\infty$:\\
          The middle part is empty ($w_1 = \varepsilon$) and we have $w = w_0 b w_2$ where $b$ is the label of the position $\rho(w) = \lambda(w)$ in $w$.
        \item $-\infty < \rho(w) < \lambda(w) < +\infty$:\\
          We have $w = w_0 b w_1 c w_2$ where $b$ is the label of the position $\rho(w)$ and $c$ is the label of the position $\lambda(w)$ in $w$.
        \item $-\infty < \rho(w) < \lambda(w) = +\infty$:\\
          The right part is empty ($w_2 = \varepsilon$) and we have $w = w_0 b w_1$ where $b$ is the label of the position $\rho(w)$ in $w$.
        \item $-\infty < \rho(w) = \lambda(w) = +\infty$:\\
          The middle part and the right part are both empty ($w_1 = w_2 = \varepsilon$) and the left part spans the whole word ($w_0 = w$).
      \end{enumerate}
    \end{remark*}
  
    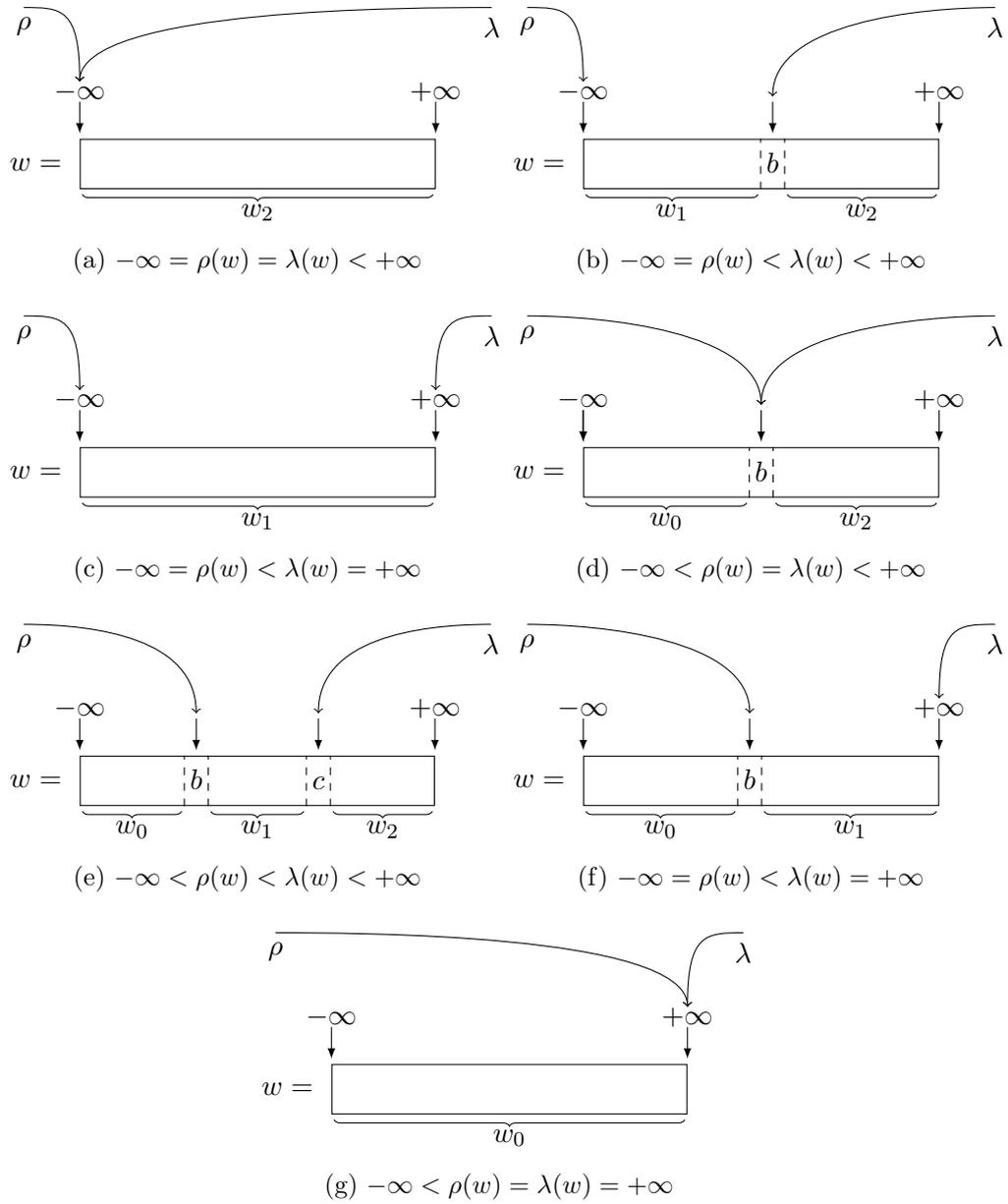
\begin{figure}\centering
      \newlength{\bwidth}
      \settowidth{\bwidth}{$b$}
      \addtolength{\bwidth}{4pt}
      \newlength{\cwidth}
      \settowidth{\cwidth}{$c$}
      \addtolength{\cwidth}{4pt}
      \begin{subfigure}{.45\linewidth}\centering
        \begin{tikzpicture}[baseline=(wLabel.base)]
          \node (wLabel) {$w = {}$};
          \matrix [right=0cm of wLabel, rectangle, draw, matrix of math nodes, ampersand replacement=\&, inner sep=2pt, text height=.8em, text depth=.2em] (w) {
            |[minimum width=\dimexpr4cm+\bwidth+\cwidth]| {} \\
          };
          \node[above=0.5cm of w.north west, inner sep=0pt] (-infty) {$-\infty$};
          \draw (-infty) edge[-latex, shorten >= 2pt] (w.north west);
          
          \node[above=0.5cm of w.north east, inner sep=0pt] (+infty) {$+\infty$};
          \draw (+infty) edge[-latex, shorten >= 2pt] (w.north east);
          
          \draw[decorate, decoration=brace] ([yshift=-2pt, xshift=-1pt]w.south east) -- node[midway, below] {$w_2$} ([yshift=-2pt, xshift=1pt]w.south west);
          
          \draw[->] ($(-infty.north -| w.west) + (-0.75cm, 1cm)$) .. controls +(0.5cm, 0cm) and +(0cm,1cm) .. node[below, pos=0] {$\rho$} (-infty.north);
          \draw[->] ($(+infty.north -| w.east) + (0.75cm, 1cm)$) .. controls +(-1cm, 0cm) and +(0cm,1cm) .. node[below, pos=0] {$\lambda$} (-infty.north);
        \end{tikzpicture}
        \caption{$-\infty = \rho(w) = \lambda(w) < +\infty$}
      \end{subfigure}
      \begin{subfigure}{.45\linewidth}\centering
        \begin{tikzpicture}[baseline=(wLabel.base)]
          \node (wLabel) {$w = {}$};
          \matrix [right=0cm of wLabel, rectangle, draw, matrix of math nodes, ampersand replacement=\&, inner sep=2pt, text height=.8em, text depth=.2em] (w) {
            \&[2cm+\cwidth] b \&[2cm] \\
          };
          \node[above=0.5cm of w.north west, inner sep=0pt] (-infty) {$-\infty$};
          \draw (-infty) edge[-latex, shorten >= 2pt] (w.north west);
          \path (-infty.south) -| coordinate (barrown) (w-1-2.north);
          \path (w.north) -| coordinate (barrows) (w-1-2.north);
          \draw (barrown) edge[-latex, shorten >= 2pt] (barrows);
          
          \node[above=0.5cm of w.north east, inner sep=0pt] (+infty) {$+\infty$};
          \draw (+infty) edge[-latex, shorten >= 2pt] (w.north east);
          
          \path (w.north) -| coordinate (bnw) (w-1-2.west);
          \path (w.south) -| coordinate (bsw) (w-1-2.west);
          \draw[dashed, dash phase=1.5pt] (bnw) -- (bsw);
          \path (w.north) -| coordinate (bne) (w-1-2.east);
          \path (w.south) -| coordinate (bse) (w-1-2.east);
          \draw[dashed, dash phase=1.5pt] (bne) -- (bse);
          
          \draw[decorate, decoration=brace] ([yshift=-2pt, xshift=-1pt]bsw) -- node[midway, below] {$w_1$} ([yshift=-2pt, xshift=1pt]w.south west);
          \draw[decorate, decoration=brace] ([yshift=-2pt, xshift=-1pt]w.south east) -- node[midway, below] {$w_2$} ([yshift=-2pt, xshift=1pt]bse);
          
          \draw[->] ($(-infty.north -| w.west) + (-0.75cm, 1cm)$) .. controls +(0.5cm, 0cm) and +(0cm,1cm) .. node[below, pos=0] {$\rho$} (-infty.north);
          \draw[->] ($(+infty.north -| w.east) + (0.75cm, 1cm)$) .. controls +(-1cm, 0cm) and +(0cm,1cm) .. node[below, pos=0] {$\lambda$} ($(barrown)+(0pt,2pt)$);
        \end{tikzpicture}
        \caption{$-\infty = \rho(w) < \lambda(w) < +\infty$}
      \end{subfigure}
      \\[1\baselineskip]
      \begin{subfigure}{.45\linewidth}\centering
        \begin{tikzpicture}[baseline=(wLabel.base)]
          \node (wLabel) {$w = {}$};
          \matrix [right=0cm of wLabel, rectangle, draw, matrix of math nodes, ampersand replacement=\&, inner sep=2pt, text height=.8em, text depth=.2em] (w) {
            |[minimum width=\dimexpr4cm+\bwidth+\cwidth]| {} \\
          };
          \node[above=0.5cm of w.north west, inner sep=0pt] (-infty) {$-\infty$};
          \draw (-infty) edge[-latex, shorten >= 2pt] (w.north west);
          
          \node[above=0.5cm of w.north east, inner sep=0pt] (+infty) {$+\infty$};
          \draw (+infty) edge[-latex, shorten >= 2pt] (w.north east);
          
          \draw[decorate, decoration=brace] ([yshift=-2pt, xshift=-1pt]w.south east) -- node[midway, below] {$w_1$} ([yshift=-2pt, xshift=1pt]w.south west);
          
          \draw[->] ($(-infty.north -| w.west) + (-0.75cm, 1cm)$) .. controls +(0.5cm, 0cm) and +(0cm,1cm) .. node[below, pos=0] {$\rho$} (-infty.north);
          \draw[->] ($(+infty.north -| w.east) + (0.75cm, 1cm)$) .. controls +(-0.5cm, 0cm) and +(0cm,1cm) .. node[below, pos=0] {$\lambda$} (+infty.north);
        \end{tikzpicture}
        \caption{$-\infty = \rho(w) < \lambda(w) = +\infty$}
      \end{subfigure}
      \begin{subfigure}{.45\textwidth}
        \begin{tikzpicture}[baseline=(wLabel.base)]
          \node (wLabel) {$w = {}$};
          \matrix [right=0cm of wLabel, rectangle, draw, matrix of math nodes, ampersand replacement=\&, inner sep=2pt, text height=.8em, text depth=.2em] (w) {
            \&[2cm+\cwidth/2] b \&[2cm+\cwidth/2] \\
          };
          \node[above=0.5cm of w.north west, inner sep=0pt] (-infty) {$-\infty$};
          \draw (-infty) edge[-latex, shorten >= 2pt] (w.north west);
          \node[above=0.5cm of w.north east, inner sep=0pt] (+infty) {$+\infty$};
          \draw (+infty) edge[-latex, shorten >= 2pt] (w.north east);
          
          \draw (-infty) edge[-latex, shorten >= 2pt] (w.north west);
          \path (-infty.south) -| coordinate (barrown) (w-1-2.north);
          \path (w.north) -| coordinate (barrows) (w-1-2.north);
          \draw (barrown) edge[-latex, shorten >= 2pt] (barrows);
          
          \path (w.north) -| coordinate (bnw) (w-1-2.west);
          \path (w.south) -| coordinate (bsw) (w-1-2.west);
          \draw[dashed, dash phase=1.5pt] (bnw) -- (bsw);
          \path (w.north) -| coordinate (bne) (w-1-2.east);
          \path (w.south) -| coordinate (bse) (w-1-2.east);
          \draw[dashed, dash phase=1.5pt] (bne) -- (bse);
          
          \draw[decorate, decoration=brace] ([yshift=-2pt, xshift=-1pt]bsw) -- node[midway, below, align=center] {$w_0$} ([yshift=-2pt, xshift=1pt]w.south west);
          \draw[decorate, decoration=brace] ([yshift=-2pt, xshift=-1pt]w.south east) -- node[midway, below] {$w_2$} ([yshift=-2pt, xshift=1pt]bse);
          
          \draw[->] ($(-infty.north -| w.west) + (-0.75cm, 1cm)$) .. controls +(1cm, 0cm) and +(0cm,1cm) .. node[below, pos=0] {$\rho$} ($(barrown)+(0pt,2pt)$);
          \draw[->] ($(+infty.north -| w.east) + (0.75cm, 1cm)$) .. controls +(-1cm, 0cm) and +(0cm,1cm) .. node[below, pos=0] {$\lambda$} ($(barrown)+(0pt,2pt)$);
        \end{tikzpicture}
        \caption{$-\infty < \rho(w) = \lambda(w) < +\infty$}
      \end{subfigure}
      \\[1\baselineskip]
      \begin{subfigure}{.45\linewidth}\centering
        \begin{tikzpicture}[baseline=(wLabel.base)]
          \node (wLabel) {$w = {}$};
          \matrix [right=0cm of wLabel, rectangle, draw, matrix of math nodes, ampersand replacement=\&, inner sep=2pt, text height=.8em, text depth=.2em] (w) {
            \&[1.33cm] b \&[1.33cm] c \&[1.33cm] \\
          };
          \node[above=0.5cm of w.north west, inner sep=0pt] (-infty) {$-\infty$};
          \draw (-infty) edge[-latex, shorten >= 2pt] (w.north west);
          \node[above=0.5cm of w.north east, inner sep=0pt] (+infty) {$+\infty$};
          \draw (+infty) edge[-latex, shorten >= 2pt] (w.north east);
          
          \path (-infty.south) -| coordinate (barrown) (w-1-2.north);
          \path (w.north) -| coordinate (barrows) (w-1-2.north);
          \draw (barrown) edge[-latex, shorten >= 2pt] (barrows);
          
          \path (+infty.south) -| coordinate (carrown) (w-1-3.north);
          \path (w.north) -| coordinate (carrows) (w-1-3.north);
          \draw (carrown) edge[-latex, shorten >= 2pt] (carrows);

          \path (w.north) -| coordinate (bnw) (w-1-2.west);
          \path (w.south) -| coordinate (bsw) (w-1-2.west);
          \draw[dashed, dash phase=1.5pt] (bnw) -- (bsw);
          \path (w.north) -| coordinate (bne) (w-1-2.east);
          \path (w.south) -| coordinate (bse) (w-1-2.east);
          \draw[dashed, dash phase=1.5pt] (bne) -- (bse);
          
          \path (w.north) -| coordinate (cnw) (w-1-3.west);
          \path (w.south) -| coordinate (csw) (w-1-3.west);
          \draw[dashed, dash phase=1.5pt] (cnw) -- (csw);
          \path (w.north) -| coordinate (cne) (w-1-3.east);
          \path (w.south) -| coordinate (cse) (w-1-3.east);
          \draw[dashed, dash phase=1.5pt] (cne) -- (cse);

          \draw[decorate, decoration=brace] ([yshift=-2pt, xshift=-1pt]bsw) -- node[midway, below] {$w_0$} ([yshift=-2pt, xshift=1pt]w.south west);
          \draw[decorate, decoration=brace] ([yshift=-2pt, xshift=-1pt]csw) -- node[midway, below] {$w_1$} ([yshift=-2pt, xshift=1pt]bse);
          \draw[decorate, decoration=brace] ([yshift=-2pt, xshift=-1pt]w.south east) -- node[midway, below] {$w_2$} ([yshift=-2pt, xshift=1pt]cse);

          \draw[->] ($(-infty.north -| w.west) + (-0.75cm, 1cm)$) .. controls +(1cm, 0cm) and +(0cm,1cm) .. node[below, pos=0] {$\rho$} ($(barrown)+(0pt,2pt)$);
          \draw[->] ($(+infty.north -| w.east) + (0.75cm, 1cm)$) .. controls +(-1cm, 0cm) and +(0cm,1cm) .. node[below, pos=0] {$\lambda$} ($(carrown)+(0pt,2pt)$);
        \end{tikzpicture}
        \caption{$-\infty < \rho(w) < \lambda(w) < +\infty$}
      \end{subfigure}
      \begin{subfigure}{.45\linewidth}\centering
        \begin{tikzpicture}[baseline=(wLabel.base)]
          \node (wLabel) {$w = {}$};
          \matrix [right=0cm of wLabel, rectangle, draw, matrix of math nodes, ampersand replacement=\&, inner sep=2pt, text height=.8em, text depth=.2em] (w) {
            \&[2cm] b \&[2cm+\cwidth] \\
          };
          \node[above=0.5cm of w.north west, inner sep=0pt] (-infty) {$-\infty$};
          \draw (-infty) edge[-latex, shorten >= 2pt] (w.north west);
          \path (-infty.south) -| coordinate (barrown) (w-1-2.north);
          \path (w.north) -| coordinate (barrows) (w-1-2.north);
          \draw (barrown) edge[-latex, shorten >= 2pt] (barrows);
          
          \node[above=0.5cm of w.north east, inner sep=0pt] (+infty) {$+\infty$};
          \draw (+infty) edge[-latex, shorten >= 2pt] (w.north east);
          
          \path (w.north) -| coordinate (bnw) (w-1-2.west);
          \path (w.south) -| coordinate (bsw) (w-1-2.west);
          \draw[dashed, dash phase=1.5pt] (bnw) -- (bsw);
          \path (w.north) -| coordinate (bne) (w-1-2.east);
          \path (w.south) -| coordinate (bse) (w-1-2.east);
          \draw[dashed, dash phase=1.5pt] (bne) -- (bse);
          
          \draw[decorate, decoration=brace] ([yshift=-2pt, xshift=-1pt]bsw) -- node[midway, below] {$w_0$} ([yshift=-2pt, xshift=1pt]w.south west);
          \draw[decorate, decoration=brace] ([yshift=-2pt, xshift=-1pt]w.south east) -- node[midway, below] {$w_1$} ([yshift=-2pt, xshift=1pt]bse);
          
          \draw[->] ($(-infty.north -| w.west) + (-0.75cm, 1cm)$) .. controls +(1cm, 0cm) and +(0cm,1cm) .. node[below, pos=0] {$\rho$} ($(barrown)+(0pt,2pt)$);
          \draw[->] ($(+infty.north -| w.east) + (0.75cm, 1cm)$) .. controls +(-0.5cm, 0cm) and +(0cm,1cm) .. node[below, pos=0] {$\lambda$} (+infty.north);
        \end{tikzpicture}
        \caption{$-\infty = \rho(w) < \lambda(w) = +\infty$}
      \end{subfigure}
      \\[1\baselineskip]
      \begin{subfigure}{.45\linewidth}\centering
        \begin{tikzpicture}[baseline=(wLabel.base)]
          \node (wLabel) {$w = {}$};
          \matrix [right=0cm of wLabel, rectangle, draw, matrix of math nodes, ampersand replacement=\&, inner sep=2pt, text height=.8em, text depth=.2em] (w) {
            |[minimum width=\dimexpr4cm+\bwidth+\cwidth]| {} \\
          };
          \node[above=0.5cm of w.north west, inner sep=0pt] (-infty) {$-\infty$};
          \draw (-infty) edge[-latex, shorten >= 2pt] (w.north west);
          
          \node[above=0.5cm of w.north east, inner sep=0pt] (+infty) {$+\infty$};
          \draw (+infty) edge[-latex, shorten >= 2pt] (w.north east);
          
          \draw[decorate, decoration=brace] ([yshift=-2pt, xshift=-1pt]w.south east) -- node[midway, below] {$w_0$} ([yshift=-2pt, xshift=1pt]w.south west);
          
          \draw[->] ($(-infty.north -| w.west) + (-0.75cm, 1cm)$) .. controls +(1cm, 0cm) and +(0cm,1cm) .. node[below, pos=0] {$\rho$} (+infty.north);
          \draw[->] ($(+infty.north -| w.east) + (0.75cm, 1cm)$) .. controls +(-0.5cm, 0cm) and +(0cm,1cm) .. node[below, pos=0] {$\lambda$} (+infty.north);
        \end{tikzpicture}
        \caption{$-\infty < \rho(w) = \lambda(w) = +\infty$}
      \end{subfigure}
      \caption{The $\rho$-$\lambda$-split $(w_0, w_1, w_2)$ of $w$ in various cases}\label{fig:rhoLambdaSplit}
    \end{figure}

    The most important part for our congruence is the following relation, which describes a connection between the $\rho$-$\lambda$-splits of two words.
    \begin{definition}\label{def:singleRankerPair}
      Let $u, v \in \Sigma^*$, $m \in \Np$, $\rho$ an $X$-ranker and $\lambda$ a $Y$-ranker. We let $u \sim_{m, \rho, \lambda} v$ if and only if
      \begin{enumerate}[noitemsep]
        \item $u$ and $v$ are $\rho$-$\lambda$-compatible and,
        \item for all prefixes $\rho'$ of $\rho\Xinfty$ and $\lambda'$ of $\lambda\Yinfty$ with $\rho'(u) \leq \lambda'(u)$ and $\rho'(v) \leq \lambda'(v)$, we have that, if $(u_0, u_1, u_2)$ is the $\rho'$-$\lambda'$-split of $u$ and $(v_0, v_1, v_2)$ is the $\rho'$-$\lambda'$-split of $v$, then $\alphabet u_1 = \alphabet v_1$ and
        \[
          \forall a \not\in \alphabet u_1 = \alphabet v_1: |u_0|_a \equiv |v_0|_a \bmod m \text{ and } |u_2|_a \equiv |v_2|_a \bmod m
        \]
        hold.
      \end{enumerate}
    \end{definition}
    \begin{figure}\centering
      \begin{subfigure}{0.5\linewidth}\centering
        \begin{tikzpicture}
          \node (uLabel) {$u = {}$};
          \matrix [right=0cm of uLabel, rectangle, draw, matrix of math nodes, ampersand replacement=\&, inner sep=2pt, text height=.8em, text depth=.2em] (u) {
            \&[1.5cm] b \&[1.5cm] c \&[1.5cm] \\
          };
          \node[above=0.5cm of u.north west, inner sep=0pt] (-infty) {$-\infty$};
          \draw (-infty) edge[-latex, shorten >= 2pt] (u.north west);
          \path (-infty.south) -| coordinate (aarrown) (u-1-2.north);
          \path (u.north) -| coordinate (aarrows) (u-1-2.north);
          \draw (aarrown) edge[-latex, shorten >= 2pt] (aarrows);
          
          \node[above=0.5cm of u.north east, inner sep=0pt] (+infty) {$+\infty$};
          \draw (+infty) edge[-latex, shorten >= 2pt] (u.north east);
          \path (+infty.south) -| coordinate (barrown) (u-1-3.north);
          \path (u.north) -| coordinate (barrows) (u-1-3.north);
          \draw (barrown) edge[-latex, shorten >= 2pt] (barrows);
          
          \path (u.north) -| coordinate (ubnw) (u-1-2.west);
          \path (u.south) -| coordinate (ubsw) (u-1-2.west);
          \draw[dashed, dash phase=1.5pt] (ubnw) -- (ubsw);
          \path (u.north) -| coordinate (ubne) (u-1-2.east);
          \path (u.south) -| coordinate (ubse) (u-1-2.east);
          \draw[dashed, dash phase=1.5pt] (ubne) -- (ubse);
          
          \path (u.north) -| coordinate (ucnw) (u-1-3.west);
          \path (u.south) -| coordinate (ucsw) (u-1-3.west);
          \draw[dashed, dash phase=1.5pt] (ucnw) -- (ucsw);
          \path (u.north) -| coordinate (ucne) (u-1-3.east);
          \path (u.south) -| coordinate (ucse) (u-1-3.east);
          \draw[dashed, dash phase=1.5pt] (ucne) -- (ucse);
          
          \draw[decorate, decoration=brace] ([yshift=-2pt, xshift=-1pt]ubsw) -- node[midway, below] (u0) {$u_0$} ([yshift=-2pt, xshift=1pt]u.south west);
          \draw[decorate, decoration=brace] ([yshift=-2pt, xshift=-1pt]ucsw) -- node[midway, below] {$u_1$} ([yshift=-2pt, xshift=1pt]ubse);
          \draw[decorate, decoration=brace] ([yshift=-2pt, xshift=-1pt]u.south east) -- node[midway, below] {$u_2$} ([yshift=-2pt, xshift=1pt]ucse);
          
          \draw[->] ($(-infty.north -| u.west) + (0cm, 1cm)$) .. controls +(1cm, 0cm) and +(0cm,1cm) .. node[below, pos=0] {$\rho'$} ($(aarrown)+(0pt,2pt)$);
          \draw[->] ($(+infty.north -| u.east) + (0cm, 1cm)$) .. controls +(-1cm, 0cm) and +(0cm,1cm) .. node[below, pos=0] {$\lambda'$} ($(barrown)+(0pt,2pt)$);
          
          \matrix [below=1cm of u, rectangle, draw, matrix of math nodes, ampersand replacement=\&, inner sep=2pt, text height=.8em, text depth=.2em] (v) {
            \&[1.75cm] b \&[2cm] c \&[0.75cm] \\
          };
          \node[left=0cm of v] (vLabel) {$v = {}$};
          
          \node[below=0.5cm of v.south west, inner sep=0pt] (-inftyV) {$-\infty$};
          \draw (-inftyV) edge[-latex, shorten >= 2pt] (v.south west);
          \path (-inftyV.north) -| coordinate (aarrowVs) (v-1-2.south);
          \path (v.south) -| coordinate (aarrowVn) (v-1-2.south);
          \draw (aarrowVs) edge[-latex, shorten >= 2pt] (aarrowVn);
          
          \node[below=0.5cm of v.south east, inner sep=0pt] (+inftyV) {$+\infty$};
          \draw (+inftyV) edge[-latex, shorten >= 2pt] (v.south east);
          \path (+inftyV.north) -| coordinate (barrowVs) (v-1-3.south);
          \path (v.south) -| coordinate (barrowVn) (v-1-3.south);
          \draw (barrowVs) edge[-latex, shorten >= 2pt] (barrowVn);
          
          \path (v.north) -| coordinate (vbnw) (v-1-2.west);
          \path (v.south) -| coordinate (vbsw) (v-1-2.west);
          \draw[dashed, dash phase=1.5pt] (vbnw) -- (vbsw);
          \path (v.north) -| coordinate (vbne) (v-1-2.east);
          \path (v.south) -| coordinate (vbse) (v-1-2.east);
          \draw[dashed, dash phase=1.5pt] (vbne) -- (vbse);
          
          \path (v.north) -| coordinate (vcnw) (v-1-3.west);
          \path (v.south) -| coordinate (vcsw) (v-1-3.west);
          \draw[dashed, dash phase=1.5pt] (vcnw) -- (vcsw);
          \path (v.north) -| coordinate (vcne) (v-1-3.east);
          \path (v.south) -| coordinate (vcse) (v-1-3.east);
          \draw[dashed, dash phase=1.5pt] (vcne) -- (vcse);
          
          \draw[decorate, decoration=brace] ([yshift=2pt, xshift=1pt]v.north west) -- node[midway, above] (v0) {$v_0$} ([yshift=2pt, xshift=-1pt]vbnw);
          \draw[decorate, decoration=brace] ([yshift=2pt, xshift=1pt]vbne) -- node[midway, above] {$v_1$} ([yshift=2pt, xshift=-1pt]vcnw);
          \draw[decorate, decoration=brace] ([yshift=2pt, xshift=1pt]vcne) -- node[midway, above] {$v_2$} ([yshift=2pt, xshift=-1pt]v.north east);
          
          \draw[->] ($(-inftyV.south -| v.west) - (0cm, 1cm)$) .. controls +(1cm, 0cm) and +(0cm,-1cm) .. node[above, pos=0] {$\rho'$} ($(aarrowVs)-(0pt,2pt)$);
          \draw[->] ($(+inftyV.south -| v.east) - (0cm, 1cm)$) .. controls +(-1cm, 0cm) and +(0cm,-1cm) .. node[above, pos=0] {$\lambda'$} ($(barrowVs)-(0pt,2pt)$);
          
          \path[pattern=north east lines] (ubne) rectangle (ucsw);
          \path[pattern=north east lines] (vbne) rectangle (vcsw);
        \end{tikzpicture}
      \end{subfigure}%
      \caption{Illustration of the relation $u \sim_{m, \rho, \lambda} v$. We need to have the same alphabet $\alphabet u_1 = \alphabet v_1$ in the middle and the respective left parts and the respective right parts need to have the same number of $a$s modulo $m$ for all $a \not\in \alphabet(u_1) = \alphabet(v_1)$.}
    \end{figure}
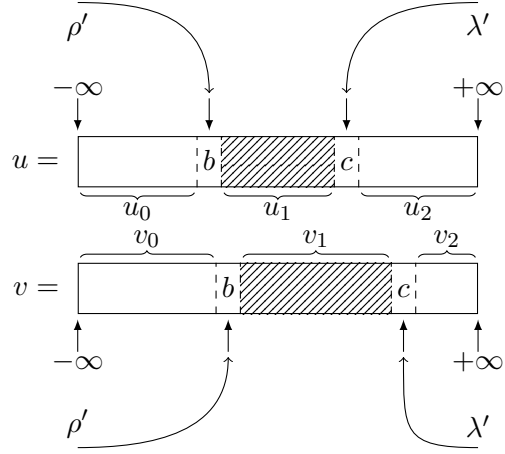
    \begin{remark}\label{rmk:simImpliesCongruentNumberOfA}
      By using the prefixes $\rho' = \rho \Xinfty$ and $\lambda' = \varepsilon$, we obtain that $u \sim_{m, \rho, \lambda} v$ implies $|u|_a \equiv |v|_a \bmod m$ for all letters $a$.
    \end{remark}
    \begin{fact}\label{fct:simRhoLambdaIsEquivalence}
      For all $X$-rankers $\rho$, $Y$-rankers $\lambda$ and $m \in \Np$, the relation $\sim_{m, \rho, \lambda}$ is an equivalence of finite index.
    \end{fact}
    \begin{proof}
      Clearly, $\sim_{m, \rho, \lambda}$ is an equivalence. The only interesting part is to show that it is of finite index. For this, we define a function $h$ mapping a word $w \in \Sigma^*$ to a triple whose first component is from $\Sigma^{\leq |\rho| + |\lambda|}$, whose second component consists of up to $|\rho| + |\lambda| + 1$ many subsets of $\Sigma$ and whose third component consists of up to $|\rho| + |\lambda| + 1$ many functions $\Sigma \to \{ 0, \dots, m - 1 \}$. Note that, in this way, we only have finitely many possible values for $h(w)$.
      
      To describe the values of the three components explicitly, let $w = z_0 c_1 z_1 \dots c_d z_d$ be the $\{ \rho, \lambda \}$-factorization of $w$. Then, the first component of $h(w)$ is given by $c_1 \dots c_d$. Note here that we have $d \leq |\rho| + |\lambda|$ as the position of every $c_k$ belongs to at least one (non-empty) prefix of $\rho$ or $\lambda$. The second component is $(\alphabet z_0, \dots, \alphabet z_d)$. Finally, the third component is $(f_0, \dots, f_d)$ for the functions $f_k: \Sigma \to \{ 0, \dots, m - 1 \}$, $a \mapsto |z_k|_a \bmod m$ (where $|z_k|_a \bmod m$ is the smallest non-negative representative of the residue class of $|z_k|_a$ modulo $m$).
      
      By choosing an arbitrary system of representatives of the classes of $\sim_{m, \rho, \lambda}$, we obtain a function from the set of these classes into a finite set by choosing $h(w)$ as the image of the class with the representative $w$. We will next prove that $u \not\sim_{m, \rho, \lambda} v$ implies $h(u) \neq h(v)$, which shows that this function is an injection and that, thus, $\sim_{m, \rho, \lambda}$ is of finite index.
      
      For this, let $u \not\sim_{m, \rho, \lambda} v$ and consider the $\{ \rho, \lambda \}$-factorizations $u = x_0 a_1 x_1 \dots a_d x_d$ and $v = y_0 b_1 y_1 \dots b_e y_e$ of $u$ and $v$ (where $a_1, \dots, a_d \in \Sigma$ and $b_1, \dots, b_e \in \Sigma$ are at the $\{ \rho, \lambda \}$-positions). We are done, if the first component $a_1 \dots a_d$ of $h(u)$ is different to the first component $b_1 \dots b_e$ of $h(v)$. So, assume that we have $a_1 \dots a_d = b_1 \dots b_e$. By \autoref{fct:compatibleIfSmallerAlphabet}, $u$ is $\rho$-$\lambda$-compatible to $a_1 \dots a_d$ and $v$ is $\rho$-$\lambda$-compatible to $b_1 \dots b_e = a_1 \dots a_d$. Thus, we have that $u$ and $v$ are $\rho$-$\lambda$-compatible (by transitivity) and that the $\{ \rho, \lambda \}$-factorization of $v$ is $v = y_0 a_1 y_1 \dots a_d y_d$.
      
      So, there must be a prefix $\rho'$ of $\rho\Xinfty$ and a prefix $\lambda'$ of $\lambda\Yinfty$ with $\rho'(u) \leq \lambda'(u)$ and $\rho'(v) \leq \lambda'(v)$ such that, for the $\rho'$-$\lambda'$-split $(u_0, u_1, u_2)$ of $u$ and the $\rho'$-$\lambda'$-split $(v_0, v_1, v_2)$ of $v$, we have $\alphabet u_1 \neq \alphabet v_1$ or we have $\alphabet u_1 = \alphabet v_1$ but there is some $a \not\in \alphabet u_1 = \alphabet v_1$ with $|u_0|_a \not\equiv |v_0|_a \bmod m$ or $|u_2|_a \not\equiv |v_2|_a \bmod m$. See \autoref{fig:rhoLambdaSplits}.
      
      By the definition of a $\{ \rho, \lambda \}$-factorization, $\rho'(u)$ must be the position of $a_i$ in $u$ for some $i$ (or $\pm\infty$). It is not difficult to see that $\rho'(v)$ is the position of the same $a_i$ in $v$ (or also $\pm\infty$). Together with a similar argumentation for $\lambda'$, we obtain, for the middle parts of the $\rho'$-$\lambda'$-splits, that there are $i$ and $j$ with $u_1 = x_i a_{i + 1} x_{i + 1} \dots a_j x_j$ and $v_1 = y_i a_{i + 1} y_{i + 1} \dots a_j y_j$.
      
      We are done if there is a difference in the second components of $h(u)$ and $h(v)$. Thus, we may assume $\alphabet x_k = \alphabet y_k$ for all $k$ and obtain
      \begin{align*}
        \alphabet u_1 &= \alphabet x_i a_{i + 1} x_{i + 1} \dots a_j x_j = \{ a_{i + 1}, \dots, a_j \} \cup \bigcup_{k = i}^j \alphabet x_k \\
        &= \{ a_{i + 1}, \dots, a_j \} \cup \bigcup_{k = i}^j \alphabet y_k = \alphabet y_i a_{i + 1} y_{i + 1} \dots a_j y_j = \alphabet v_1 \text{.}
      \end{align*}
      As stated above, there must be some $a \not\in \alphabet u_1 = \alphabet v_1$ with $|u_0|_a \not\equiv |v_0|_a \bmod m$ or $|u_2|_a \not\equiv |v_2|_a \bmod m$. We only show the former of the two cases as the latter is left-right symmetric. Observe that we have $u_0 = x_0 a_1 x_1 \dots a_{i - 1} x_{i - 1}$ and $v_0 = y_0 a_1 y_1 \dots a_{i - 1} y_{i - 1}$ (see \autoref{fig:rhoLambdaSplits} again). Let $(f_0, \dots, f_d)$ be the third component of $h(u)$ and $(g_0, \dots, g_d)$ be the third component of $h(v)$. Suppose that the two components are equal. Then, we have, in particular, $f_k(a) = g_k(a)$ for all $k$ and, therefore,
      \begin{align*}
        |u_0|_a &= |a_1 \dots a_{i - 1}|_a + \sum_{k = 0}^{i - 1} |x_k|_a = |a_1 \dots a_{i - 1}|_a + \sum_{k = 0}^{i - 1} f_k(a) \\
        &= |a_1 \dots a_{i - 1}|_a + \sum_{k = 0}^{i - 1} g_k(a) = |a_1 \dots a_{i - 1}|_a + \sum_{k = 0}^{i - 1} |y_k|_a = |v_0|_a \text{,}
      \end{align*}
      which constitutes a contradiction.
    \end{proof}
    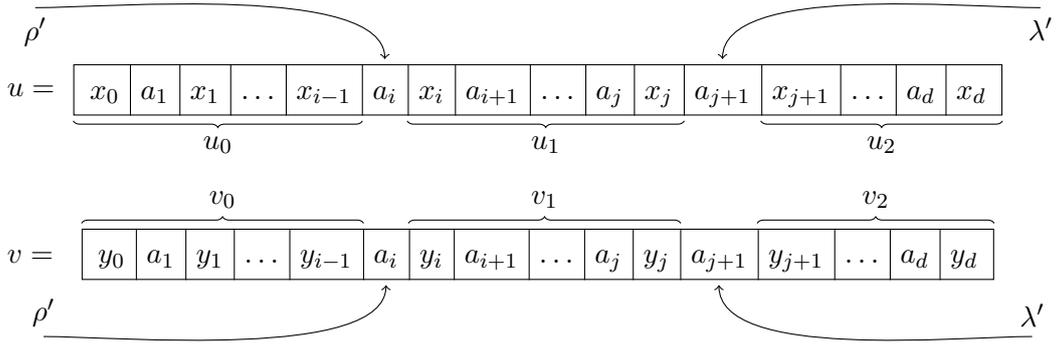
\begin{figure}\centering
      \begin{tikzpicture}
        \node (uLabel) {$u = {}$};
        \matrix [right=0cm of uLabel, rectangle, draw, matrix of math nodes, ampersand replacement=\&, inner ysep=2pt, inner xsep=4pt, text height=.8em, text depth=.2em] (u) {
          \hspace*{-2pt}x_0 \& a_1 \& x_1 \& \dots \& x_{i - 1} \& a_i \& x_i \& a_{i + 1} \& \dots \& a_j \& x_j \& a_{j + 1} \& x_{j + 1} \& \dots \& a_d \& x_d\hspace*{-2pt} \\
        };
        \foreach \i in {1,...,15} {
          \draw ($(u.north -| u-1-\i.north east)$) -- ($(u.south -| u-1-\i.south east)$);
        };
        \draw[decorate, decoration={brace}] ($(u.south -| u-1-5.south east)-(0pt,2pt)$) -- node[below, yshift=-2pt] {$u_0$} ($(u.south west)-(0pt,2pt)$);
        \draw[decorate, decoration={brace}] ($(u.south -| u-1-11.south east)-(0pt,2pt)$) -- node[below, yshift=-2pt] {$u_1$} ($(u.south -| u-1-7.south west)-(0pt,2pt)$);
        \draw[decorate, decoration={brace}] ($(u.south east)-(0pt,2pt)$) -- node[below, yshift=-2pt] {$u_2$} ($(u.south -| u-1-13.south west)-(0pt,2pt)$);
        \draw[->] ($(u.north west) + (-0.5cm, 0.75cm)$) .. controls +(1cm, 0cm) and +(0cm,1cm) .. node[below, pos=0] {$\rho'$} ($(u.north -| u-1-6.north)+(0pt,2pt)$);
        \draw[->] ($(u.north east) + (0.5cm, 0.75cm)$) .. controls +(-1cm, 0cm) and +(0cm,1cm) .. node[below, pos=0] {$\lambda'$} ($(u.north -| u-1-12.north)+(0pt,2pt)$);
      
        \matrix [below=1.5cm of u.south, rectangle, draw, matrix of math nodes, ampersand replacement=\&, inner ysep=2pt, inner xsep=4pt, text height=.8em, text depth=.2em] (v) {
          \hspace*{-2pt}y_0 \& a_1 \& y_1 \& \dots \& y_{i - 1} \& a_i \& y_i \& a_{i + 1} \& \dots \& a_j \& y_j \& a_{j + 1} \& y_{j + 1} \& \dots \& a_d \& y_d\hspace*{-2pt} \\
        };
        \node[anchor=base west] at ($(uLabel.base west |- v.base)$) {$v ={}$};
        \foreach \i in {1,...,15} {
          \draw ($(v.north -| v-1-\i.north east)$) -- ($(v.south -| v-1-\i.south east)$);
        };
        \draw[decorate, decoration={brace}] ($(v.north west)+(0pt,2pt)$) -- node[above, yshift=2pt] {$v_0$} ($(v.north -| v-1-5.north east)+(0pt,2pt)$);
        \draw[decorate, decoration={brace}] ($(v.north -| v-1-7.north west)+(0pt,2pt)$) -- node[above, yshift=2pt] {$v_1$} ($(v.north -| v-1-11.north east)+(0pt,2pt)$);
        \draw[decorate, decoration={brace}] ($(v.north -| v-1-13.north west)+(0pt,2pt)$) -- node[above, yshift=2pt] {$v_2$} ($(v.north east)+(0pt,2pt)$);
        
        \draw[->] ($(v.south west) + (-0.5cm, -0.75cm)$) .. controls +(1cm, 0cm) and +(0cm,-1cm) .. node[above, pos=0] {$\rho'$} ($(v.south -| v-1-6.south)-(0pt,2pt)$);
        \draw[->] ($(v.south east) + (0.5cm, -0.75cm)$) .. controls +(-1cm, 0cm) and +(0cm,-1cm) .. node[above, pos=0] {$\lambda'$} ($(v.south -| v-1-12.south)-(0pt,2pt)$);
      \end{tikzpicture}
      \caption{The $\rho'$-$\lambda'$-splits of $u$ and $v$.}\label{fig:rhoLambdaSplits}
    \end{figure}
    \begin{remark*}
      Despite \autoref{fct:simRhoLambdaIsEquivalence}, $\sim_{m, \rho, \lambda}$ is not a congruence in general. For example, consider the words $u = bc$ and $v = cb$ and the rankers $\rho = X_a X_b$ and $\lambda = \varepsilon$. The words are $\{ \rho, \lambda \}$-compatible and, for any prefix $\rho'$ of $\rho\Xinfty$ and $\lambda'$ of $\lambda\Yinfty$, the $\rho'$-$\lambda'$-split can only be $(u, \varepsilon, \varepsilon)$ for $u$ and $(v, \varepsilon, \varepsilon)$ for $v$, $(\varepsilon, u, \varepsilon)$ for $u$ and $(\varepsilon, v, \varepsilon)$ for $v$ or $(\varepsilon, \varepsilon, u)$ for $u$ and $(\varepsilon, \varepsilon, v)$ for $v$. Thus, we have $u \sim_{m, \rho, \lambda} v$. Clearly, we also have $a \sim_{m, \rho, \lambda} a$ but we do not have $au \sim_{m, \rho, \lambda} av$ as $\rho' = \rho = X_a X_b$ and $\lambda' = \lambda = \varepsilon$ split $u$ into $(a, c, \varepsilon)$ and $v$ into $(ac, \varepsilon, \varepsilon)$, so the alphabets of the middle parts ($c$ and $\varepsilon$) do not coincide.
    \end{remark*}
  
    However, we obtain a congruence -- and this is the congruence we will use to describe $\DAb$ -- by taking a finite intersection of some relations $\sim_{m, \rho, \lambda}$.
    \begin{definition}\label{def:congruence}
      Let $u, v \in \Sigma^*$ and $m \in \Np$, $n \in \mathbb{N}$. We let $u \approx_{m, n} v$ if $u \sim_{m, \rho, \lambda} v$ holds for all pairs of an $X$-ranker $\rho$ and a $Y$-ranker $\lambda$ with $|\rho| + |\lambda| \leq n$.
    \end{definition}

    \begin{fact}
      For all $m \in \Np$, $n \in \mathbb{N}$, the relation $\approx_{m, n}$ is a congruence of finite index.
    \end{fact}
    \begin{proof}
      By definition, $\approx_{m, n}$ is the intersection of all $\sim_{m, \rho, \lambda}$ for the pairs of an $X$-ranker $\rho$ and a $Y$-ranker $\lambda$ with $|\rho| + |\lambda| \leq n$. Since these finitely many $\sim_{m, \rho, \lambda}$ are equivalences of finite index (by \autoref{fct:simRhoLambdaIsEquivalence}), this also holds for $\approx_{m, n}$.
      
      Thus, the only interesting part of the proof is to show that $\approx_{m, n}$ is compatible. For this, let $u_\ell, v_\ell, u_r, v_r \in \Sigma^*$ with $u_\ell \approx_{m, n} v_\ell$ and $u_r \approx_{m, n} v_r$ and consider an $X$-ranker $\rho$ and a $Y$-ranker $\lambda$ with $|\rho| + |\lambda| \leq n$. Furthermore, let $\rho'$ be a prefix of $\rho\Xinfty$ and $\lambda'$ a prefix of $\lambda\Yinfty$. We have to show $\operatorname{ord}[\rho', \lambda'; u_\ell u_r] = \operatorname{ord}[\rho', \lambda'; v_\ell v_r]$
      and, if $\rho'(u_\ell u_r) \leq \lambda'(u_\ell u_r)$ (and, thus, also $\rho'(v_\ell v_r) \leq \lambda'(v_\ell, v_r)$), that, for the $\rho'$-$\lambda'$-split $(u_0, u_1, u_2)$ of $u_\ell u_r$ and the $\rho'$-$\lambda'$-split $(v_0, v_1, v_2)$ of $v_\ell v_r$, we have $\alphabet u_1 = \alphabet v_1$ and
      \[
        \forall a \not\in \alphabet u_1 = \alphabet v_1: |u_0|_a \equiv |v_0|_a \bmod m \text{ and } |u_2|_a \equiv |v_2|_a \bmod m \text{.}
      \]
      
      Let $\rho' = \rho'_0 \rho'_1$ where $\rho'_0$ is the longest prefix of $\rho'$ such that $\rho'_0(u_\ell) < +\infty$ (possibly empty). We first show
      \[
        \rho'(u_\ell u_r) = \begin{cases}
          \rho'(u_\ell) & \text{if } \rho'(u_\ell) < +\infty \\
          |u_\ell| + \rho'_1(u_r) & \text{otherwise}
        \end{cases}
      \]
      where we let $|u_\ell| + \infty = +\infty$. Note that $\rho'_1(u_r) = -\infty$ implies that $\rho'_1$ must be empty and that $\rho'_1$ being empty implies $\rho'(u_\ell) = \rho'_1(u_\ell) < +\infty$ (i.\,e.\ we are in the second case).
      
      The first case obviously holds and, in the second case, $\rho'_1$ cannot be empty. Thus, we may write $\rho' = \rho'_0 \Xinfty \rho''_1$ or $\rho' = \rho'_0 X_a \rho''_1$ for some $a \in \Sigma$ where $\rho''_1$ is some suffix of $\rho'$. In the former case, we $\rho'(u_\ell u_r) = +\infty = \Xinfty \rho''_1(u_r) = \rho'_1(u_r)$ and, in the latter case, we obtain the following situation:
      
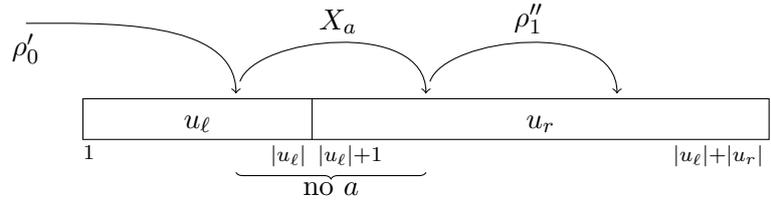
\begin{center}
        \begin{tikzpicture}[baseline=(ul.base)]
          \node[text height=0.8em, text depth=0.2em, inner sep=2pt, minimum width=3cm] (ul) {$u_\ell$};
          \node[text height=0.8em, text depth=0.2em, inner sep=2pt, minimum width=6cm, right=0pt of ul] (ur) {$u_r$};
          
          \draw (ul.north west) rectangle (ur.south east);
          \draw (ul.north east) -- (ul.south east);
          
          \node[below right=2pt and 0pt of ul.south west, inner sep=0pt] (1) {$\scriptstyle 1$};
          \node[below left=2pt and 2pt of ul.south east, inner sep=0pt] (ulpos) {$\scriptstyle |u_\ell|$};
          \node[below right=2pt and 2pt of ur.south west, inner sep=0pt] (ul+1pos) {$\scriptstyle |u_\ell| + 1$};
          \node[below left=2pt and 2pt of ur.south east, inner sep=0pt] (ul+urpos) {$\scriptstyle |u_\ell| + |u_r|$};
          
          \draw[->] ($(ul.north west) + (-0.75cm, 1cm)$) .. controls +(1cm, 0cm) and +(0cm,1cm) .. node[below, pos=0] {$\rho'_0$} ($(ul.north east)+(-1cm,2pt)$);
          \draw[->, shorten <= 1ex] ($(ul.north east)+(-1cm,2pt)$) .. controls +(0.25cm, 0.75cm) and +(0cm,1cm) .. node[midway, above] {$X_a$} ($(ur.north west)+(1.5cm,2pt)$);
          \draw[->, shorten <= 1ex] ($(ur.north west)+(1.5cm,2pt)$) .. controls +(0.25cm, 0.75cm) and +(0cm,1cm) .. node[midway, above] {$\rho''_1$} ($(ur.north east)+(-2cm,2pt)$);
          
          \draw[decorate, decoration={brace}] ($(ulpos.south -| ur.south west)+(1.5cm,-2pt)$) -- node[below] {no $a$} ($(ulpos.south -| ul.south east)+(-1cm,-2pt)$);
        \end{tikzpicture},
      \end{center}
      which shows that we have $\rho'(u_\ell u_r) = |u_\ell| + X_a\rho''_1(u_r) = |u_\ell| + \rho'_1(u_r)$.
      
      We also get an analogous statement for $v_\ell v_r$ and, by left-right symmetry, a statement for $\lambda$: Let $\lambda' = \lambda'_0 \lambda'_1$ where $\lambda'_0$ is the longest prefix of $\lambda'$ such that $-\infty < \lambda'_0(u_r)$. Then, we have
      \[
        \lambda'(u_\ell u_r) = \begin{cases}
          |u_\ell| + \lambda'(u_r) & \text{if } -\infty < \lambda'(u_r) \\
          \lambda'_1(u_\ell) & \text{otherwise}
        \end{cases}
      \]
      (where we again let $|u_\ell| + \infty = +\infty$) and an analogous statement for $v$.
      
      One consequence of these statements is that we have
      \[
        \rho'(u_\ell u_r) \leq |u_\ell| \iff \rho'(u_\ell) < +\infty
        \iff \rho'(v_\ell) < +\infty \iff \rho'(v_\ell v_r) \leq |v_\ell|
      \]
      where we have used that $u_\ell$ and $v_\ell$ are $\rho$-$\lambda$-compatible for the second equivalence (using $\rho'$ and the empty $Y$-ranker as the prefixes). The negation of this statement also yields $|u_\ell| < \rho'(u_\ell u_r) \iff |v_\ell| < \rho'(v_\ell v_r)$.
      
      In the same way, we obtain
      \[
        |u_\ell| < \lambda'(u_\ell u_r) \iff -\infty < \lambda'(u_r) \iff -\infty < \lambda'(v_r) \iff |v_\ell| < \lambda'(v_\ell v_r)
      \]
      and, by negation, $\lambda'(u_\ell u_r) \leq |u_\ell| \iff \lambda'(v_\ell v_r) \leq |v_\ell|$.
      
      \begin{figure}\centering
        \begin{subfigure}{0.5\textwidth}\centering
          \begin{tikzpicture}[baseline=(uLabel.base)]
            \matrix [rectangle, draw, matrix of math nodes, ampersand replacement=\&, inner sep=2pt, text height=.8em, text depth=.2em] (u) {
              |[minimum width=2cm, align=center]| $u_\ell / v_\ell$ \& |[minimum width=1.5cm, align=center]| $u_r / v_r$ \\
            };
            \draw ($(u.north -| u-1-1.north east)$) -- ($(u.south -| u-1-1.south east)$);
            
            \path[pattern=north west lines] ($(u.north west)+(0pt,1mm+2pt)$) rectangle ($(u.north -| u-1-1.north east)+(0pt,2pt)$);
            \path[pattern=north west lines] ($(u.south west)+(0pt,-2pt)$) rectangle ($(u.south -| u-1-1.north east)+(0pt,-1mm-2pt)$);
            
            \draw[->] ($(u.north west) + (-0.5cm, 0.85cm+2pt)$) .. controls +(0.5cm, 0cm) and +(0cm,0.75cm) .. node[below, pos=0] {$\rho'$} ($(u.north -| u-1-1.north) + (0pt, 1mm+4pt)$);
            \draw[->] ($(u.south east) + (0.5cm, -0.85cm-2pt)$) .. controls +(-0.5cm, 0cm) and +(0cm,-0.75cm) .. node[above, pos=0] {$\lambda'$} ($(u.south -| u-1-1.north) + (0pt, -1mm-4pt)$);
          \end{tikzpicture}
          \caption{$\rho'(u_\ell u_r) \leq |u_\ell|$, $\lambda'(u_\ell u_r) \leq |u_\ell|$}\label{sfig:congruenceCases:bothLeft}
        \end{subfigure}%
        \begin{subfigure}{0.5\textwidth}\centering
          \begin{tikzpicture}[baseline=(uLabel.base)]
            \matrix [rectangle, draw, matrix of math nodes, ampersand replacement=\&, inner sep=2pt, text height=.8em, text depth=.2em] (u) {
              |[minimum width=2cm, align=center]| $u_\ell / v_\ell$ \& |[minimum width=1.5cm, align=center]| $u_r / v_r$ \\
            };
            \draw ($(u.north -| u-1-1.north east)$) -- ($(u.south -| u-1-1.south east)$);
            
            \path[pattern=north west lines] ($(u.north -| u-1-2.north west)+(0pt,1mm+2pt)$) rectangle ($(u.north east)+(0pt,2pt)$);
            \path[pattern=north west lines] ($(u.south -| u-1-2.north west)+(0pt,-2pt)$) rectangle ($(u.south east)+(0pt,-1mm-2pt)$);
            
            \draw[->] ($(u.north west) + (-0.5cm, 0.85cm+2pt)$) .. controls +(0.5cm, 0cm) and +(0cm,0.75cm) .. node[below, pos=0] {$\rho'$} ($(u.north -| u-1-2.north) + (0pt, 1mm+4pt)$);
            \draw[->] ($(u.south east) + (0.5cm, -0.85cm-2pt)$) .. controls +(-0.5cm, 0cm) and +(0cm,-0.75cm) .. node[above, pos=0] {$\lambda'$} ($(u.south -| u-1-2.north) + (0pt, -1mm-4pt)$);
          \end{tikzpicture}
          \caption{$|u_\ell| < \rho'(u_\ell u_r)$, $|u_\ell| < \lambda'(u_\ell u_r)$}\label{sfig:congruenceCases:bothRight}
        \end{subfigure}\\
        \begin{subfigure}{0.5\textwidth}\centering
          \begin{tikzpicture}[baseline=(uLabel.base)]
            \matrix [rectangle, draw, matrix of math nodes, ampersand replacement=\&, inner sep=2pt, text height=.8em, text depth=.2em] (u) {
              |[minimum width=2cm, align=center]| $u_\ell / v_\ell$ \& |[minimum width=1.5cm, align=center]| $u_r / v_r$ \\
            };
            \draw ($(u.north -| u-1-1.north east)$) -- ($(u.south -| u-1-1.south east)$);

            \path[pattern=north west lines] ($(u.north -| u-1-2.north west)+(0pt,1mm+2pt)$) rectangle ($(u.north east)+(0pt,2pt)$);
            \path[pattern=north west lines] ($(u.south west)+(0pt,-2pt)$) rectangle ($(u.south -| u-1-1.north east)+(0pt,-1mm-2pt)$);
            
            \draw[->] ($(u.north west) + (-0.5cm, 0.85cm+2pt)$) .. controls +(0.5cm, 0cm) and +(0cm,0.75cm) .. node[below, pos=0] {$\rho'$} ($(u.north -| u-1-2.north) + (0pt, 1mm+4pt)$);
            \draw[->] ($(u.south east) + (0.5cm, -0.85cm-2pt)$) .. controls +(-0.5cm, 0cm) and +(0cm,-0.75cm) .. node[above, pos=0] {$\lambda'$} ($(u.south -| u-1-1.south) + (0pt, -1mm-4pt)$);
          \end{tikzpicture}
          \caption{$|u_\ell| < \rho'(u_\ell u_r)$, $\lambda'(u_\ell u_r) \leq |u_\ell|$}\label{sfig:congruenceCases:crossing}
        \end{subfigure}%
        \begin{subfigure}{0.5\textwidth}\centering
          \begin{tikzpicture}[baseline=(uLabel.base)]
            \matrix [rectangle, draw, matrix of math nodes, ampersand replacement=\&, inner sep=2pt, text height=.8em, text depth=.2em] (u) {
              |[minimum width=2cm, align=center]| $u_\ell / v_\ell$ \& |[minimum width=1.5cm, align=center]| $u_r / v_r$ \\
            };
            \draw ($(u.north -| u-1-1.north east)$) -- ($(u.south -| u-1-1.south east)$);

            \path[pattern=north west lines] ($(u.north west)+(0pt,1mm+2pt)$) rectangle ($(u.north -| u-1-1.north east)+(0pt,2pt)$);
            \path[pattern=north west lines] ($(u.south -| u-1-2.north west)+(0pt,-2pt)$) rectangle ($(u.south east)+(0pt,-1mm-2pt)$);
            
            \draw[->] ($(u.north west) + (-0.5cm, 0.85cm+2pt)$) .. controls +(0.5cm, 0cm) and +(0cm,0.75cm) .. node[below, pos=0] {$\rho'$} ($(u.north -| u-1-1.north) + (0pt, 1mm+4pt)$);
            \draw[->] ($(u.south east) + (0.5cm, -0.85cm-2pt)$) .. controls +(-0.5cm, 0cm) and +(0cm,-0.75cm) .. node[above, pos=0] {$\lambda'$} ($(u.south -| u-1-2.north) + (0pt, -1mm-4pt)$);
          \end{tikzpicture}
          \caption{$\rho'(u_\ell u_r) \leq |u_\ell|$, $|u_\ell| < \lambda'(u_\ell u_r)$}\label{sfig:congruenceCases:leftAndRight}
        \end{subfigure}
        \caption{The four cases for the positions of $\rho'$ and $\lambda'$}\label{fig:congruenceCases}
      \end{figure}
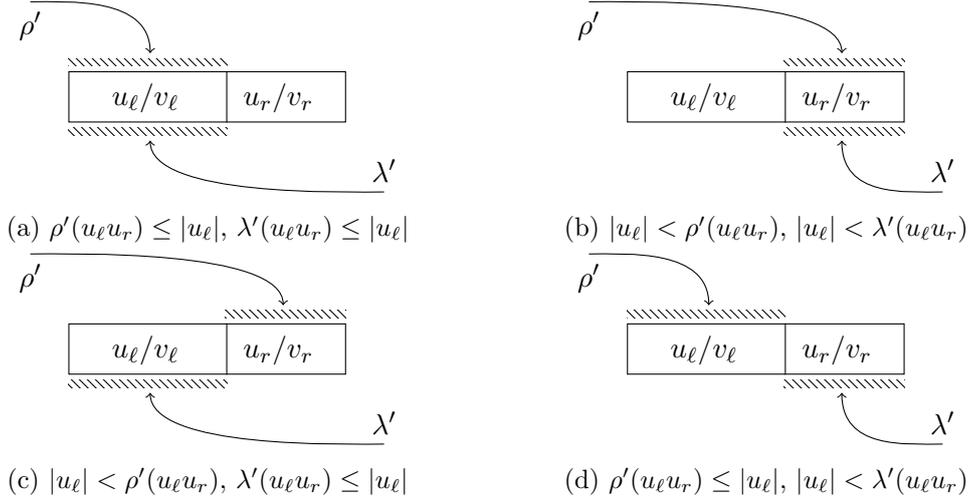
      
      For the main part, we have four cases (see \autoref{fig:congruenceCases}): First, consider the case $\rho'(u_\ell u_r) \leq |u_\ell|$ and $\lambda'(u_\ell u_r) \leq |u_\ell|$, which (by the above) is equivalent to $\rho'(u_\ell) < +\infty$ and $-\infty = \lambda'(u_r)$ (see \autoref{sfig:congruenceCases:bothLeft}). As just stated, we must also have $\rho'(v_\ell v_r) \leq |v_\ell|$, $\lambda'(v_\ell v_r) \leq |v_\ell|$, $\rho'(v_\ell) < +\infty$ and $-\infty = \lambda'(v_r)$. Thus, we have $\rho'(u_\ell u_r) = \rho'(u_\ell)$ and $\rho'(v_\ell v_r) = \rho'(v_\ell)$.
      
      Furthermore, we have $\lambda'(u_\ell u_r) = \lambda'_1(u_\ell)$ where $\lambda'_1$ is given by letting $\lambda' = \lambda'_0 \lambda'_1$ for the longest prefix $\lambda'_0$ of $\lambda'$ such that $-\infty < \lambda'_0(u_r)$. Note that $\lambda'_0$ is a non-empty prefix of $\lambda \Yinfty$ (as it is one of $\lambda'$) and that we, thus, have $-\infty < \lambda'_0(v_r)$ (by using that $u_r$ and $v_r$ are $\rho$-$\lambda$-compatible). Reversing the same argument also shows that $\lambda'_0$ is the longest prefix of $\lambda'$ such that $-\infty < \lambda'_0(v_r)$. Thus, we also have $\lambda'(v_\ell v_r) = \lambda'_1(v_\ell)$.
      
      Finally, we have $|\lambda'_1| \leq |\lambda|$ (since $|\lambda'_0| \geq 1$) and, therefore, $|\rho'| + |\lambda'_1| \leq n$ (as $\rho'$ cannot contain $\Xinfty$), which implies that $u_\ell$ and $v_\ell$ are $\rho'$-$\lambda'_1$-compatible (as we have $u_\ell \approx_{m, n} v_\ell$ and, thus, $u_\ell \sim_{m, \rho', \lambda'_1} v_\ell$). We may conclude:
      \[
        \operatorname{ord}[\rho', \lambda'; u_\ell u_r] = \operatorname{ord}[\rho', \lambda_1'; u_\ell] = \operatorname{ord}[\rho', \lambda_1'; v_\ell] = \operatorname{ord}[\rho', \lambda'; v_\ell v_r] \text{.}
      \]
      
      It remains to show the statement about the $\rho'$-$\lambda'$-splits of $u_\ell u_r$ and $v_\ell v_r$ if we have $\rho'(u_\ell u_r) \leq \lambda'(u_\ell u_r)$ (and, thus, also $\rho'(v_\ell v_r) \leq \lambda'(v_\ell v_r)$). The $\rho'$-$\lambda'$-split of $u_\ell u_r$ is $(u_{\ell, 0}, u_{\ell, 1}, u_{\ell, 2} u_r)$ if the $\rho'$-$\lambda'_1$-spit of $u_\ell$ is denoted by $(u_{\ell, 0}, u_{\ell, 1}, u_{\ell, 2})$. In the same way, the $\rho'$-$\lambda'$-split of $v_\ell v_r$ is $(v_{\ell, 0}, v_{\ell, 1}, v_{\ell, 2} v_r)$ if the $\rho'$-$\lambda'_1$-spit of $v_\ell$ is denoted by $(v_{\ell, 0}, v_{\ell, 1}, v_{\ell, 2})$. Due to $u_\ell \sim_{m, \rho', \lambda'_1} v_\ell$, we have $\alphabet u_{\ell, 1} = \alphabet v_{\ell, 1}$ and $|u_{\ell, 0}|_a \equiv |v_{\ell, 0}|_a \bmod m$ as well as $|u_{\ell, 2}|_a \equiv |v_{\ell, 2}|_a \bmod m$ for all $a \not\in \alphabet u_{\ell, 1} = \alphabet v_{\ell, 1}$. As we also have $|u_r|_a \equiv |v_r|_a \bmod m$ for all $a$ (see \autoref{rmk:simImpliesCongruentNumberOfA}), the latter implies $|u_{\ell, 2} u_r|_a \equiv |v_{\ell, 2} v_r|_a \bmod m$ for all $a \not\in \alphabet u_{\ell, 1} = \alphabet v_{\ell, 1}$, which concludes this case.

      The second case $|u_\ell| < \rho'(u_\ell u_r)$ and $|u_\ell| < \lambda'(u_\ell u_r)$ (and, thus, $|v_\ell| < \rho'(v_\ell v_r)$ and $|v_\ell| < \lambda'(v_\ell v_r)$, see \autoref{sfig:congruenceCases:bothRight}) is left-right symmetrical. In the third case $|u_\ell| < \rho'(u_\ell u_r)$ but $\lambda'(u_\ell u_r) \leq |u_\ell|$ (and, equivalently, $|v_\ell| < \rho'(v_\ell v_r)$ but $\lambda'(v_\ell v_r) \leq |v_\ell|$, see \autoref{sfig:congruenceCases:crossing}), we have
      \[
        \rho'(u_\ell u_r) > |u_\ell| \geq \lambda'(u_\ell u_r) \quad\text{and}\quad \rho'(v_\ell v_r) > |v_\ell| \geq \lambda'(v_\ell v_r)
      \]
      and there is nothing more to show.
      
      It remains the fourth case with $\rho'(u_\ell u_r) \leq |u_\ell|$ but $|u_\ell| < \lambda'(u_\ell u_r)$ (and, equivalently, $\rho'(v_\ell v_r) \leq |v_\ell|$ but $|v_\ell| < \lambda'(v_\ell v_r)$, see \autoref{sfig:congruenceCases:leftAndRight}). Here, we have
      \[
        \rho'(u_\ell u_r) \leq |u_\ell| < \lambda'(u_\ell u_r) \quad\text{and}\quad \rho'(v_\ell v_r) \leq |v_\ell| < \lambda'(v_\ell v_r)
      \]
      but we also need to show the statement about the $\rho'$-$\lambda'$-splits of $u_\ell u_r$ and $v_\ell v_r$. Note that we have $\rho'(u_\ell) < +\infty$ and $-\infty < \lambda'(u_r)$ (and, equivalently, $\rho'(v_\ell) < +\infty$ and $-\infty < \lambda'(v_r)$). Thus, we may consider the $\rho'$-$\varepsilon$-split $(u_{\ell, 0}, u_{\ell, 1}, \varepsilon)$ of $u_\ell$ and the $\varepsilon$-$\lambda'$-split $(\varepsilon, u_{r, 1}, u_{r, 2})$ of $u_r$ as well as the $\rho'$-$\varepsilon$-split $(v_{\ell, 0}, v_{\ell, 1}, \varepsilon)$ of $v_\ell$ and the $\varepsilon$-$\lambda'$-split $(\varepsilon, v_{r, 1}, v_{r, 2})$ of $v_r$, for which we have $\alphabet u_{\ell, 1} = \alphabet v_{\ell, 1}$ (since $|\rho'| + |\varepsilon| \leq |\rho| + |\lambda| \leq n$) and $\alphabet u_{r, 1} = \alphabet v_{r, 1}$ (since $|\varepsilon| + |\lambda'| \leq |\rho| + |\lambda| \leq n$) as well as $|u_{\ell, 0}|_a \equiv |v_{\ell, 0}|_a \bmod m$ for all $a \not\in \alphabet u_{\ell, 1} = \alphabet v_{\ell, 1}$ and $|u_{r, 2}|_a \equiv |v_{r, 2}|_a \bmod m$ for all $a \not\in \alphabet u_{r, 1} = \alphabet v_{r, 1}$.
      
      We may combine these splits and obtain that the $\rho'$-$\lambda'$-split of $u_\ell u_r$ is $(u_{\ell, 0}, u_{\ell, 1} u_{r, 1}, u_{r, 2})$ and that the $\rho'$-$\lambda'$-split of $v_\ell v_r$ is $(v_{\ell, 0}, v_{\ell, 1} v_{r, 1}, v_{r, 2})$ with $\alphabet u_{\ell, 1} u_{r, 1} = \alphabet u_{\ell, 1} \cup \alphabet u_{r, 1} = \alphabet v_{\ell, 1} \cup \alphabet v_{r, 1} = \alphabet v_{\ell, 1} v_{r, 1}$. Together with the above statements on the residual equivalence, this concludes the proof.
    \end{proof}
  
    We will show that a monoid is in $\DAb$ if and only if it is a homomorphic image of $\Sigma^* / {\approx_{m, n}}$ for some suitable $m, n$. The easier one of the two directions (namely that every homomorphic image is in $\DAb$) follows from the following proposition.
    \begin{proposition}\label{prop:congruenceInDAb}
      The monoid $\Sigma^* / {\approx_{m, n}}$ is in $\DAb$ for all $m \in \Np$, $n \in \mathbb{N}$.
    \end{proposition}
    \begin{proof}
      By \autoref{lem:DAbEquation}, we have to show that the equation $(xy)^{\omega - 1} = (yx)^{\omega - 1}$ holds in $M = \Sigma^* / {\approx_{m, n}}$. This is equivalent to showing that, for any words $x, y \in \Sigma^*$,
      \begin{align*}
        (xy)^{M! - 1} =_M (xy)^{n \cdot M!} (xy)^{M! - 1} (xy)^{n \cdot M!} \approx_{m, n} (yx)^{n \cdot M!} (yx)^{M! - 1} (yx)^{n \cdot M!} =_M (yx)^{M! - 1}
      \end{align*}
      holds. Let $u = u_0 u_1 u_2$ and $v = v_0 v_1 v_2$ with
      \begin{align*}
        u_0 &= (xy)^{n \cdot M!} \text{,}    & v_0 &= (yx)^{n \cdot M!} \text{,} \\
        u_1 &= (xy)^{M! - 1}     \text{,}    & v_1 &= (yx)^{M! - 1}     \text{,} \\
        u_2 &= (xy)^{n \cdot M!} \text{ and} & v_2 &= (yx)^{n \cdot M!}
      \end{align*}
      and let $\rho$ be an $X$-ranker and $\lambda$ a $Y$-ranker such that $|\rho| + |\lambda| \leq n$ holds.
      
      Let $\rho'$ be an arbitrary prefix of $\rho\Xinfty$. There are only three possible cases: if $\rho'$ is empty, it ends at $-\infty$ in $u$ and in $v$. Otherwise, it can only end within $u_0$ for $u$ and within $v_0$ for $v$ or at $+\infty$ for both words. Similarly, a prefix $\lambda'$ of $\lambda\Yinfty$ can only end at $+\infty$ for both words, within $u_2$ for $u$ and $v_2$ for $v$ or at $-\infty$ for both words (which is the case if and only if $\lambda'$ is empty). This proves that $u$ and $v$ are $\rho$-$\lambda$-compatible.
      
      Furthermore, it also shows that there are only three cases for the $\rho'$-$\lambda'$-splits of $u$ and $v$: we can have 
      \begin{itemize}[noitemsep]
        \item the split $(u, \varepsilon, \varepsilon)$ for $u$ and the split $(v, \varepsilon, \varepsilon)$ for $v$,
        \item the split $(\varepsilon, \varepsilon, u)$ for $u$ and $(\varepsilon, \varepsilon, v)$ for $v$ or
        \item a split of $u$ whose middle part contains both $x$ and $y$ as an infix and a split of $v$ for which the same holds.
      \end{itemize}
      In the first two cases, we have that the alphabets of the middle parts are empty. Since $x$ appears as often in $u$ as it does in $v$ and since the same holds for $y$, both $u$ and $v$ contain each letter the exact same number of times and there is nothing more to show. In the last case, the middle parts contain all the letters that appear in $u$ and $v$ (so their alphabets, in particular, coincide) and there is nothing to show.
    \end{proof}
  
    Proving the other direction (every monoid in $\DAb$ is a quotient of a monoid of the form $\Sigma^* / {\approx_{m, n}}$) is more involved and requires quite a bit of combinatorics.
    \begin{proposition}\label{prop:DAbImpliesQuotient}
      Let $M \in \DAb$ and $\varphi: \Sigma^* \to M$ a homomorphism. There are $m, n \in \Np$ such that $u \approx_{m, n} v \implies \varphi(u) = \varphi(v)$ for all $u, v \in \Sigma^*$.
    \end{proposition}
    \begin{proof}
      Let $R$ be the number of $\mathcal{R}$-classes of $M$ and let $L$ be the number of its $\mathcal{L}$-classes. Define $n = R + L$ and let $m$ be the least common multiple of the orders of the groups appearing in the regular $\mathcal{D}$-classes of $M$ (we could also chose $m = |M|!$).

      Now, assume that we have $u, v \in \Sigma^*$ with $u \approx_{m, n} v$. By \autoref{lem:rankerVisitingFactorizationPositions}, there is an $X$-ranker $\rho$ that visits exactly the positions of the $\mathcal{R}$-factorization of $u$ and a $Y$-ranker $\lambda$ that visits exactly the positions of the $\mathcal{L}$-factorizations of $v$. Note that we have $|\rho| + 1 \leq R$ and $|\lambda| + 1 \leq L$ and, thus, in particular, $|\rho| + |\lambda| \leq R + L = n$. Therefore, $u \approx_{m, n} v$, implies $u \sim_{m, \rho, \lambda} v$.
      
      We say a word $u' \in \Sigma^*$ is \emph{$u$-like} if we have
      \begin{enumerate}[label=(u-\Alph*), noitemsep, leftmargin=*]
        \item\label{itm:rhoLambdaSim} $u' \sim_{m, \rho, \lambda} u$ (and, thus, $u' \sim_{m, \rho, \lambda} v$) and,
        \item\label{itm:rhoRDescent} if $\rho'$ is a prefix of $\rho \Xinfty$ and $(u_0', u_1')$ is the $\rho$-split of $u'$ and $(u_0, u_1)$ is the $\rho$-split of $u$, we have $\varphi(u_0') = \varphi(u_0)$.
      \end{enumerate}
      Note that, by using $\rho \Xinfty$ as the prefix $\rho'$, we obtain, in particular, $\varphi(u') = \varphi(u)$. In fact, we obtain from \autoref{itm:rhoRDescent} that $\rho$ visits exactly the points of the $\mathcal{R}$-factorization of $u'$ and that, therefore, the $\{ \rho, \lambda \}$-factorization of $u'$ is a sub-factorization of its $\mathcal{R}$-factorization. Clearly, $u$ itself is $u$-like and we could replace $u$ in the above definition by any other $u$-like word while still obtaining the same class of words.
      
      We also introduce a symmetric version for $v$: a word $v' \in \Sigma^*$ is \emph{$v$-like} if we have
      \begin{enumerate}[label=(v-\Alph*), noitemsep, leftmargin=*]
        \item\label{itm:vLike:rhoLambdaSim} $v' \sim_{m, \rho, \lambda} v$ (and, thus, $v' \sim_{m, \rho, \lambda} u$) and,
        \item\label{itm:vLike:rhoRDescent} if $\lambda'$ is a prefix of $\lambda \Yinfty$ and $(u_0', u_1')$ is the $\lambda$-split of $v'$ and $(u_0, u_1)$ is the $\lambda$-split of $v$, we have $\varphi(u_1') = \varphi(u_1)$.
      \end{enumerate}
    
      The overall idea of the remaining proof is to start with $u$ and $v$ as a pair of a $u$-like and a $v$-like word and to successively apply transformations to them in order to obtain new pairs of $u$-like and $v$-like words where, in a certain sense, the $u$-like word and the $v$-like word become more and more similar. In the end, we want to obtain a word that is both $u$-like and $v$-like. Thus, it has in particular the same image as $u$ and the same image as $v$ under $\varphi$, which shows $\varphi(u) = \varphi(v)$.
      
      To describe the transformations we will use in this process, let $u'$ be a $u$-like word and let $v'$ be a $v$-like word (in the first step, we use $u' = u$ and $v' = v$). Furthermore, let
      \begin{equation}
        u' = x_0' c_1 x_1' \dots c_d x_d' \quad\text{and}\quad v' = y_0' c_1 y_1' \dots c_d y_d'\tag{$\ast$}\label{eqn:rhoLambdaFactorization1}
      \end{equation}
      be the $\{ \rho, \lambda \}$-factorizations of $u'$ and $v'$. Note here that $u'$ and $v'$ are both $\rho$-$\lambda$-compatible to $u$ as well as to $v$. Note also that the given factorization for $u'$ is a sub-factorization of its $\mathcal{R}$-factorization and that the given factorization for $v'$ is a sub-factorization of its $\mathcal{L}$-factorization.
      
      With this in mind, it is not surprising that we may permute the factors of these factorizations. In fact, this is the first transformation that we will use below.
      \begin{lemma}\label{fct:permutationRemainsULike}
        Let $\tilde{x}_k'$ be a permutation of $x_k'$. Then
        \[
          u'' = x_0' c_1 x_1' \dots c_{k - 1} x_{k - 1}' \  c_k \  \tilde{x}_k' \  c_{k + 1} x_{k + 1}' \dots c_d x_d'
        \]
        remains $u$-like.
        
        In the same way, if $\tilde{y}_k'$ is a permutation of $y_k'$, then
        \[
          v'' = y_0' c_1 y_1' \dots c_{k - 1} y_{k - 1}' \  c_k \  \tilde{y}_k' \  c_{k + 1} y_{k + 1}' \dots c_d y_d'
        \]
        remains $v$-like
      \end{lemma}
      \begin{proof}
        The two statements are left-right symmetric and we only show the one for $u''$.
        
        We prove \autoref{itm:rhoLambdaSim} by showing $u'' \sim_{m, \rho, \lambda} u'$ (which suffices since $u'$ is $u$-like and we, thus, have $u' \sim_{m, \rho, \lambda} u$). From \autoref{fct:compatibleIfSmallerAlphabet}, we obtain that $u'$ and $u''$ are $\rho$-$\lambda$-compatible (and that $u''$ is already stated in its $\{ \rho, \lambda \}$-factorization above).
        
        Next, let $\rho'$ be a prefix of $\rho\Xinfty$ and note that $\rho'$ must end at the position of some $c_i$ in $u'$ and in $u''$ (or at $\pm \infty$ in both words). Furthermore, let $\lambda'$ be a prefix of $\lambda\Yinfty$ with $\rho'(u') \leq \lambda'(u')$ (and, thus, $\rho'(u'') \leq \lambda(u'')$). Again, $\lambda'$ must end at the position of some $c_{j + 1}$ in $u'$ and in $u''$ (or at $\pm \infty$; but, otherwise, we are in a situation similar to the one depicted in \autoref{fig:rhoLambdaSplits}). Since passing from $x_k'$ to $\tilde{x}_k'$ neither changed the alphabet nor the number of $a$s (and since the other factors of the $\{ \rho, \lambda \}$-factorization have not changed at all), we have $\alphabet u_1'' = \alphabet u_1'$ and
        \[
          \forall a \not\in \alphabet u_1'' = \alphabet u_1': |u_0''|_a \equiv |u_0'|_a \bmod m \text{ and } |u_2''|_a \equiv |u_2'|_a \bmod m
        \]
        for the $\rho'$-$\lambda'$-split $(u_0'', u_1'', u_2'')$ of $u''$ and the $\rho'$-$\lambda'$-split $(u_0', u_1', u_2')$ of $u'$. This concludes the proof for \autoref{itm:rhoLambdaSim}.
        
        It remains to show that the left part of the $\rho'$-factorization of $u''$ has the same image under $\varphi$ as the left part of the $\rho'$-factorization of $u'$ (as it then also has the same image as the left part of the $\rho'$-factorization of $u$). This is obvious if $\rho'$ ends at some $c_i$ in $u''$ up to $c_k$ (or at $-\infty$). Thus, suppose that $\rho'$ ends at $c_i$ in $u''$ for some $i > k$ (or at $+\infty$).
        
        Since $u'$ is $u$-like, its $\{ \rho, \lambda \}$-factorization is a sub-factorization of its $\mathcal{R}$-factorization. Thus, we have
        \[
          \varphi(x_0' c_1 x_1' \dots c_{k - 1} x_{k - 1}' \  c_k) \R \varphi(x_0' c_1 x_1' \dots c_{k - 1} x_{k - 1}' \  c_k \  x_k') \text{.}
        \]
        By \autoref{fct:permuteR}, this implies
        \[
          \varphi(x_0' c_1 x_1' \dots c_{k - 1} x_{k - 1}' \  c_k \  x_k') = \varphi(x_0' c_1 x_1' \dots c_{k - 1} x_{k - 1}' \  c_k \  \tilde{x}_k')
        \]
        and, thus,
        \[
          \varphi(x_0' c_1 x_1' \dots c_{i - 1} x_{i - 1}') = \varphi(x_0' c_1 x_1' \dots c_{k - 1} x_{k - 1}' \  c_k \  \tilde{x}_k' c_{k + 1} x_{k + 1}' \dots c_{i - 1} x_{i - 1}') \text{.}\qedhere
        \]
      \end{proof}
      \noindent{}In the following, we will often use \autoref{fct:permutationRemainsULike} implicitly.
      
      By \autoref{fct:permutationRemainsULike}, we are done as soon as we have $|x_k'|_a = |y_k'|_a$ for all $0 \leq k \leq d$ and all $a \in \Sigma$. Therefore, let $\mu$ be minimal\footnote{So, in particular, we have $|x_k'|_b = |y_k'|_b$ for all $b \in \Sigma$ and $0 \leq k < \mu$.} such that there is some $a \in \Sigma$ with $|x_k'|_a \neq |y_k'|_a$.
      
      In this situation, we can state the following fact, which we will later use in multiple cases.
      \begin{lemma}\label{fct:droppingAMRemainsULike}
        If we have $x_\mu' = x_\mu'' a^{\beta m}$ (up to permutation) with $|x_\mu''|_a = |y_\mu'|_a$ for some $\beta > 0$ (i.\,e.\ if we have more $a$s in $x_\mu'$ than in $y_\mu'$ and the difference is a multiple of $m$), then
        \[
          u'' = x_0' c_1 x_1' \dots c_{\mu - 1} x_{\mu - 1}' \ c_\mu \ x_\mu'' \ c_{\mu + 1} x_{\mu + 1}' \dots c_d x_d'
        \]
        remains $u$-like.
        
        In the same way, if we have $y_\mu' = y_\mu'' a^{\beta m}$ (up to permutation) with $|y_\mu''|_a = |x_\mu'|_a$ for some $\beta > 0$ (i.\,e.\ if we have more $a$s in $y_\mu'$ than in $x_\mu'$ and the difference is a multiple of $m$), then
        \[
          v'' = y_0' c_1 y_1' \dots c_{\mu - 1} y_{\mu - 1}' \ c_\mu \ y_\mu'' \ c_{\mu + 1} y_{\mu + 1}' \dots c_d y_d'
        \]
        remains $v$-like.
      \end{lemma}
      \begin{proof}
        The statement for $u''$ and the one for $v''$ are not entirely symmetric this time. However, we still only show the statement for $u''$ as the other one can be shown using the same arguments.
        
        We show \autoref{itm:rhoLambdaSim} by showing $u'' \sim_{m, \rho, \lambda} u'$ (using $u' \sim_{m, \rho, \lambda} v'$). Since we have $\alphabet x_\mu'' \subseteq \alphabet x_\mu'$, we obtain that $u''$ is $\rho$-$\lambda$-compatible to $u'$ (and, thus, also to $v'$) and that $u''$ is stated above in its $\{ \rho, \lambda \}$-factorization by \autoref{fct:compatibleIfSmallerAlphabet}. In particular, we have that a prefix $\rho'$ of $\rho \Xinfty$ ends at some $c_i$ (or at $\pm \infty$) in the above factorization of $u''$, it also ends at the same $c_i$ (or also at $\pm \infty$) in the $\{ \rho, \lambda \}$-factorizations (\ref{eqn:rhoLambdaFactorization1}) of $u'$ and $v'$ and the same holds for prefixes $\lambda'$ of $\lambda \Yinfty$.
        
        Now, consider a pair $\rho'$ and $\lambda'$ of such prefixes with $\rho'(u') \leq \lambda'(u')$ (and, thus, also $\rho'(u'') \leq \lambda'(u'')$ and $\rho'(v') \leq \lambda'(v')$). Let $(u_0'', u_1'', u_2'')$ be the $\rho'$-$\lambda'$-split of $u''$, let $(u_0', u_1', u_2')$ be the $\rho'$-$\lambda'$-split of $u'$ and let $(v_0', v_1', v_2')$ be the $\rho'$-$\lambda'$-split of $v'$.
        
        First, consider the case that $\lambda'$ (and, thus, also $\rho'$) ends to the left of $x_\mu''$ in $u''$ (see \autoref{sfig:dropAMLeft}). Note that, in this case, $\lambda'$ also ends to the left of $x_\mu'$ in $u'$. We have $u_0'' = u_0'$, $u_1'' = u_1'$ and, thus, in particular, $\alphabet u_1'' = \alphabet u_1'$. Furthermore, we have $|u_2''|_a = |u_2'|_a - \beta m \equiv |u_2'|_a \bmod m$ (as well as $|u_2''|_b = |u_2'|_b$ for all $b \neq a$), which immediately proves this case. The case that $\rho'$ (and, thus, also $\lambda'$) ends to the right of $x_\mu''$ in $u''$ is symmetric.
        
        As $x_\mu''$ is a factor of the $\{ \rho, \lambda \}$-factorization of $u''$, it remains the case that $\rho'$ ends to the left of $x_\mu''$ in $u''$ and $\lambda'$ ends to its right. Note that this does not only imply that $\rho'$ ends to the left of $x_\mu'$ in $u'$ and that $\lambda'$ ends to its right but also that $\rho'$ ends to the left of $y_\mu'$ in $v'$ and $\lambda'$ ends to its right. The situation is depicted in \autoref{sfig:dropAMLeftAndRightLambda} and \autoref{sfig:dropAMLeftAndRightRho} but note that $\rho'(u'')$, $\rho'(u')$ and $\rho'(v)$ could also all be $-\infty$ and that $\lambda'(u'')$, $\lambda'(u')$ and $\lambda'(v)$ could all be $+\infty$.
        This means that there are $0 \leq i \leq \mu$ and $\mu \leq j \leq d$ such that $u_0' = x_0' c_1 x_1' \dots c_{i - 1} x_{i - 1}'$, $u_1' = x_{i}' c_{i + 1} x_{i + 1}' \dots c_j x_j'$ and $u_2' = x_{j + 1} c_{j + 2}' x_{j + 2}' \dots c_d x_d'$. The only difference between $u_1''$ and $u_1'$ is that, in $u_1''$, one of the factors has less $a$s (because we dropped $a^{\beta m}$). This means that we only have to show $a \in \alphabet u_1''$ (as this implies $\alphabet u_1'' = \alphabet u_1'$ and the statement about having the same number of letters modulo $m$ in the respective left parts and the respective right parts of the $\rho'$-$\lambda'$-factorizations of $u''$ and $u'$ already holds since we have $u_0'' = u_0'$ and $u_2'' = u_2'$).

        To do this, first consider the case that there is some prefix $\lambda_0$ of $\lambda$ such that $\lambda_0$ ends directly to the right of $x_\mu''$ in $u''$ (possibly at $+\infty$) and that it, thus, also ends directly to the right of $x_\mu'$ in $u'$ and directly to the right of $y_\mu'$ in $v'$ (see \autoref{sfig:dropAMLeftAndRightLambda}).
        Consider the $\rho'$-$\lambda_0$-splits $(\hat{u}_0'', \hat{u}_1'', \hat{u}_2'')$ of $u''$, $(\hat{u}_0', \hat{u}_1', \hat{u}_2')$ of $u'$ and $(\hat{v}_0', \hat{v}_1', \hat{v}_2')$ of $v'$. We have $\hat{u}_1' = x_i' c_{i + 1} x_{i + 1}' \dots c_\mu x_\mu'$ (with $x_\mu' = x_\mu'' a^{\beta m}$) and (thus) $\hat{u}_1' = \hat{u}_1'' a^{\beta m}$. We also have $\hat{v}_1' = y_i' c_{i + 1} y_{i + 1}' \dots c_\mu y_\mu'$. There is an $a$ in $\alphabet x_\mu' \subseteq \alphabet \hat{u}_1' = \alphabet \hat{v}_1'$ (where the last equality is due to $u' \sim_{m, \rho, \lambda} v'$). Since we have $|x_k'|_a = |y_k'|_a$ for all $k < \mu$ (by the choice of $\mu$) and also $|x_\mu''|_a = |y_\mu'|_a$, we obtain $|\hat{u}_1''|_a = |\hat{v}_1'|_a$ and, thus, $a \in \alphabet \hat{u}_1'' \subseteq \alphabet u_1''$.

        The second case is that there is no prefix of $\lambda$ ending directly to the right of $x_\mu''$ in $u''$ (not even the empty one), which means that there must be a prefix $\rho_0$ of $\rho$ which ends directly to the right of $x_\mu''$ in $u''$ (see \autoref{sfig:dropAMLeftAndRightRho}). Thus, $\rho_0$ also ends directly to the right of $x_\mu'$ in $u'$ and directly to the right of $y_\mu'$ in $v'$. Note that, this time, we have $\rho_0(u'') < \lambda'(u'') \leq +\infty$. Thus, $\rho_0$ ends at the position of $c_{\mu + 1}$ (in all words in question). We may factorize $u_1'' = \hat{u}_1'' c_{\mu + 1} \hat{u}_2''$ (where $c_{\mu + 1}$ belongs to the position reached by $\rho_0$) and, similarly, $u_1' = \hat{u}_1' c_{\mu + 1} \hat{u}_2'$ (and, thus, $\hat{u}_1' = \hat{u}_1'' a^{\beta m}$) as well as $v_1' = \hat{v}_1' c_{\mu + 1} \hat{v}_2'$. We are done if there is an $a$ in $\hat{u}_2'' = \hat{u}_2'$. So, assume $a \not\in \alphabet \hat{u}_2'$ and consider the $\rho_0$-$\lambda'$-splits of $u'$ and $v'$, whose middle parts are $\hat{u}_2'$ and $\hat{v}_2'$, respectively. Since we have $u' \sim_{m, \rho, \lambda} v'$, we also have $a \not\in \alphabet \hat{v}_2' = \alphabet \hat{u}_2'$. Furthermore, we have $c_{\mu + 1} \neq a$ since, otherwise, $\rho_0$ would already end at $a^{\beta m}$ in $u'$. However, we have $a \in \alphabet x_\mu'' a^{\beta m} \subseteq \alphabet u_1'$ and $\alphabet u_1' = \alphabet v_1'$ (because of $u' \sim_{m, \rho, \lambda} v')$. Therefore, we have $a \in \alphabet \hat{v}_1'$ and, thus, $a \in \alphabet \hat{u}_1'' \subseteq \alphabet u_1''$ (as we have $|\hat{u}_1''|_a = |\hat{v}_1'|_a$ just like in the case above).
        
        It remains to show \autoref{itm:rhoRDescent} and this part is very similar to the corresponding part of the proof for \autoref{fct:permutationRemainsULike}.
        
        We have
        \[
          \varphi(x_0' c_1 x_1' \dots c_{\mu - 1} x_{\mu - 1}' \  c_k) \R \varphi(x_0' c_1 x_1' \dots c_{\mu - 1} x_{\mu - 1}' \  c_\mu \  x_\mu')
        \]
        and, thus, by \autoref{fct:insertAM}:
        \[
          \varphi(x_0' c_1 x_1' \dots c_{\mu - 1} x_{\mu - 1}' \  c_\mu \  x_\mu') = \varphi(x_0' c_1 x_1' \dots c_{\mu - 1} x_{\mu - 1}' \  c_\mu \  x_\mu'')
        \]
        This shows that for any prefix $\rho'$ of $\rho\Xinfty$ the images of the left parts of the $\rho'$ splits of $u''$ and $u'$ coincide (which implies that they also coincide with the image of the left part of the $\rho'$ split of $u$).
      \end{proof}
      \begin{figure}\centering
        \begin{subfigure}{\linewidth}\centering
          \begin{tikzpicture}[baseline=(uLabel.base)]
            \node (uLabel) {$u'' = {}$};
            \matrix [right=0cm of uLabel, rectangle, draw, matrix of math nodes, ampersand replacement=\&, inner sep=2pt, text height=.8em, text depth=.2em] (u) {
              |[minimum width=3cm, align=center]| \dots \& x_\mu'' \textcolor{gray}{a^{\beta m}} \& |[minimum width=1cm, align=center]| \dots \\
            };
            \node[right=0pt of u, color=gray] {${} = u'$};
            \draw ($(u.north -| u-1-1.north east)$) -- ($(u.south -| u-1-1.south east)$);
            \draw ($(u.north -| u-1-3.north west)$) -- ($(u.south -| u-1-3.south west)$);
            
            \path[pattern=north west lines] ($(u.north west)+(0pt,1mm+2pt)$) rectangle ($(u.north -| u-1-1.north east)+(0pt,2pt)$);

            \draw[->] ($(u.north west) + (-0.75cm, 0.85cm+2pt)$) .. controls +(0.5cm, 0cm) and +(0cm,0.75cm) .. node[below, pos=0] {$\rho'$} ($(u.north west) + (1.25cm, 1mm+4pt)$);
            \draw[->] ($(u.north east) + (0.75cm, 0.85cm+2pt)$) .. controls +(-0.5cm, 0cm) and +(0cm,0.75cm) .. node[below, pos=0] {$\lambda'$} ($(u.north west) + (2.5cm, 1mm+4pt)$);
            
            \draw[decorate, decoration={brace}] ($(u.south west) + (1.25cm-2pt, -2pt)$) -- node[below, yshift=-2pt] {$\scriptstyle u_0''\textcolor{gray}{= u_0'}$} ($(u.south west) + (0pt, -2pt)$);
            \draw[decorate, decoration={brace}] ($(u.south west) + (2.5cm-2pt, -2pt)$) -- node[below, yshift=-2pt] {$\scriptstyle u_1''\textcolor{gray}{= u_1'}$} ($(u.south west) + (1.25cm+2pt, -2pt)$);
            \draw[decorate, decoration={brace}] ($(u.south east)+(0pt,-2pt)$) -- node[below, yshift=-2pt] {$\scriptstyle u_2''\textcolor{gray}{/ u_2'}$} ($(u.south west) + (2.5cm+2pt, -2pt)$);

          \end{tikzpicture}
          \caption{$\lambda'(u'')$ is to the left of $x_\mu''$ in $u''$}\label{sfig:dropAMLeft}
        \end{subfigure}\\
        \begin{subfigure}{\linewidth}\centering
          \begin{tikzpicture}[baseline=(uLabel.base)]
            \node (uLabel) {$u'' = {}$};
            \matrix [right=0cm of uLabel, rectangle, draw, matrix of math nodes, ampersand replacement=\&, inner sep=2pt, text height=.8em, text depth=.2em] (u) {
              |[minimum width=1cm, align=center]| \dots \& c_i \& |[minimum width=1cm, align=center]| \dots \& x_\mu'' \textcolor{gray}{a^{\beta m}} \& c_{\mu + 1} \& |[minimum width=2cm, align=center]| \dots \\
            };
            \node[right=0pt of u, color=gray] {${} = u'$};
            \foreach \i in {1,...,5} {
              \draw ($(u.north -| u-1-\i.north east)$) -- ($(u.south -| u-1-\i.south east)$);
            };
            
            \path[pattern=north west lines] ($(u.north west)+(0pt,1mm+2pt)$) rectangle ($(u.north -| u-1-3.north east)+(0pt,2pt)$);
            \path[pattern=north east lines] ($(u.north -| u-1-5.north west)+(0pt,1mm+2pt)$) rectangle ($(u.north east)+(0pt,2pt)$);
            
            \draw[->] ($(u.north west) + (-0.75cm, 0.85cm+2pt)$) .. controls +(0.5cm, 0cm) and +(0cm,0.75cm) .. node[below, pos=0] {$\rho'$} ($(u.north west) + (1.25cm, 1mm+4pt)$);
            \draw[->] ($(u.north east) + (0.75cm, 0.85cm+2pt)$) .. controls +(-0.5cm, 0cm) and +(0cm,0.75cm) .. node[below, pos=0] {$\lambda'$} ($(u.north east) + (-1.5cm, 1mm+4pt)$);
            
            \draw[->] ($(u.south east) + (0pt, -0.75cm)$) .. controls +(-0.5cm, 0cm) and +(0cm,-0.75cm) .. node[right, pos=0] {$\lambda_0$} ($(u.south -| u-1-5.south) + (0pt, -2pt)$);
            
            \draw[decorate, decoration={brace}] ($(u.north west) + (0pt, 1mm+4pt)$) -- node[above left, xshift=1ex, yshift=2pt] {$\scriptstyle u_0'' \textcolor{gray}{{}= u_0'}$} ($(u.north -| u-1-2.north west) + (0pt, 1mm+4pt)$);
            \draw[decorate, decoration={brace}] ($(u.north -| u-1-2.north east) + (0pt, 1mm+4pt)$) -- node[above, yshift=2pt] {$\scriptstyle u_1'' \textcolor{gray}{/ u_1'}$} ($(u.north east) + (-1.5cm-1ex, 1mm+4pt)$);
            \draw[decorate, decoration={brace}] ($(u.north east) + (-1.5cm+1ex, 1mm+4pt)$) -- node[above right, xshift=-1ex, yshift=2pt] {$\scriptstyle u_2'' \textcolor{gray}{{}= u_2'}$} ($(u.north east) + (0pt, 1mm+4pt)$);
            
            \draw[decorate, decoration={brace}] ($(u.south -| u-1-1.south east) + (0pt, -2pt)$) -- node[below, yshift=-2pt] {$\scriptstyle \hat{u}_0''\textcolor{gray}{= \hat{u}_0'}$} ($(u.south west) + (0pt, -2pt)$);
            \draw[decorate, decoration={brace}] ($(u.south -| u-1-4.south east) + (0pt, -2pt)$) -- node[below, yshift=-2pt] {$\scriptstyle \hat{u}_1''\textcolor{gray}{/ \hat{u}_1'}$} ($(u.south -| u-1-2.south east) + (0pt, -2pt)$);
            \draw[decorate, decoration={brace}] ($(u.south east)+(0pt,-2pt)$) -- node[below, yshift=-2pt] {$\scriptstyle \hat{u}_2''\textcolor{gray}{= \hat{u}_2'}$} ($(u.south -| u-1-5.south east) + (0pt, -2pt)$);
            
            \path[inner sep=2pt]
              let
                \p1=(u-1-4.east),\p2=(u-1-4.west),
                \n1={\x1-\x2-2*\pgfkeysvalueof{/pgf/inner xsep}-\pgflinewidth}
              in
              node[below right=1.5cm and 0cm of u.south west, rectangle, draw, matrix of math nodes, ampersand replacement=\&, text height=.8em, text depth=.2em] (v) {
                |[minimum width=1cm, align=center]| \dots \& c_i \& |[minimum width=1cm, align=center]| \dots \& |[text width=\n1, align=center]| y_\mu' \& c_{\mu + 1} \& |[minimum width=2cm, align=center]| \dots \\
              };
            \node[left=0pt of v] {$v' = {}$};
            \foreach \i in {1,...,5} {
              \draw ($(v.north -| v-1-\i.north east)$) -- ($(v.south -| v-1-\i.south east)$);
            };
            
            \path[pattern=north east lines] ($(v.south west)-(0pt,1mm+2pt)$) rectangle ($(v.south -| v-1-3.south east)-(0pt,2pt)$);

            \path[pattern=north west lines] ($(v.south -| v-1-5.south west)-(0pt,1mm+2pt)$) rectangle ($(v.south east)-(0pt,2pt)$);
            
            \draw[->] ($(v.south west) + (-0.75cm, -0.85cm-2pt)$) .. controls +(0.5cm, 0cm) and +(0cm,-0.75cm) .. node[above, pos=0] {$\rho'$} ($(v.south west) + (1.25cm, -1mm-4pt)$);
            \draw[->] ($(v.south east) + (0.75cm, -0.85cm-2pt)$) .. controls +(-0.5cm, 0cm) and +(0cm,-0.75cm) .. node[above, pos=0] {$\lambda'$} ($(v.south east) + (-1.5cm, -1mm-4pt)$);
            
            \draw[->] ($(u.south east) + (0pt, -0.75cm)$) .. controls +(-0.5cm, 0cm) and +(0cm,0.75cm) .. ($(v.north -| v-1-5.north) + (0pt, 2pt)$);
            
            \draw[decorate, decoration={brace}] ($(v.south -| v-1-2.south west) + (0pt, -1mm-4pt)$) -- node[below, yshift=-2pt] {$\scriptstyle v_0'$} ($(v.south west) + (0pt, -1mm-4pt)$);
            \draw[decorate, decoration={brace}] ($(v.south east) + (-1.5cm-1ex, -1mm-4pt)$) -- node[below, yshift=-2pt] {$\scriptstyle v_1'$} ($(v.south -| v-1-2.south east) + (0pt, -1mm-4pt)$);
            \draw[decorate, decoration={brace}] ($(v.south east) + (0pt, -1mm-4pt)$) -- node[below, yshift=-2pt] {$\scriptstyle v_2'$} ($(v.south east) + (-1.5cm+1ex, -1mm-4pt)$);
            
            \draw[decorate, decoration={brace}] ($(v.north west) + (0pt, 2pt)$) -- node[above, yshift=2pt] {$\scriptstyle \hat{v}_0'$} ($(v.north -| v-1-2.north west) + (0pt, 2pt)$);
            \draw[decorate, decoration={brace}] ($(v.north -| v-1-2.north east) + (0pt, 2pt)$) -- node[above, yshift=2pt] {$\scriptstyle \hat{v}_1'$} ($(v.north -| v-1-5.north west) + (0pt, 2pt)$);
            \draw[decorate, decoration={brace}] ($(v.north -| v-1-5.north east) + (0pt, 2pt)$) -- node[above, yshift=2pt] {$\scriptstyle \hat{v}_2'$} ($(v.north east)+(0pt,2pt)$);
          \end{tikzpicture}
          \caption{$\rho'(u'')$ is to the left of $x_\mu''$ in $u''$ and $\lambda'(u'')$ is to its right; $\lambda_0$ ends directly to the right of $x_\mu''$}\label{sfig:dropAMLeftAndRightLambda}
        \end{subfigure}\\
        \begin{subfigure}{\linewidth}\centering
          \begin{tikzpicture}[baseline=(uLabel.base)]
            \node (uLabel) {$u'' = {}$};
            \matrix [right=0cm of uLabel, rectangle, draw, matrix of math nodes, ampersand replacement=\&, inner sep=2pt, text height=.8em, text depth=.2em] (u) {
              |[minimum width=1cm, align=center]| \dots \& c_i \& |[minimum width=1cm, align=center]| \dots \& x_\mu'' \textcolor{gray}{a^{\beta m}} \& c_{\mu + 1} \& |[minimum width=1cm, align=center]| \dots \& c_{j + 1} \& |[minimum width=1cm, align=center]| \dots \\
            };
            \node[right=0pt of u, color=gray] {${} = u'$};
            \foreach \i in {1,...,7} {
              \draw ($(u.north -| u-1-\i.north east)$) -- ($(u.south -| u-1-\i.south east)$);
            };
            
            \path[pattern=north west lines] ($(u.north west)+(0pt,1mm+2pt)$) rectangle ($(u.north -| u-1-3.north east)+(0pt,2pt)$);
            \path[pattern=north east lines] ($(u.north -| u-1-6.north west)+(0pt,1mm+2pt)$) rectangle ($(u.north east)+(0pt,2pt)$);
            
            \draw[->] ($(u.north west) + (-0.75cm, 0.85cm+2pt)$) .. controls +(0.5cm, 0cm) and +(0cm,0.75cm) .. node[below, pos=0] {$\rho'$} ($(u.north west) + (1.25cm, 1mm+4pt)$);
            \draw[->] ($(u.north east) + (0.75cm, 0.85cm+2pt)$) .. controls +(-0.5cm, 0cm) and +(0cm,0.75cm) .. node[below, pos=0] {$\lambda'$} ($(u.north -| u-1-7.north) + (0pt, 1mm+4pt)$);
            
            \draw[->] ($(u.south west) + (0pt, -0.75cm)$) .. controls +(+0.5cm, 0cm) and +(0cm,-0.75cm) .. node[left, pos=0] {$\rho_0$} ($(u.south -| u-1-5.south) + (0pt, -2pt)$);
            
            \draw[decorate, decoration={brace}] ($(u.north west) + (0pt, 1mm+4pt)$) -- node[above left, xshift=1ex, yshift=2pt] {$\scriptstyle u_0'' \textcolor{gray}{{}= u_0'}$} ($(u.north -| u-1-2.north west) + (0pt, 1mm+4pt)$);
            \draw[decorate, decoration={brace}] ($(u.north -| u-1-2.north east) + (0pt, 1mm+4pt)$) -- node[above, yshift=2pt] {$\scriptstyle u_1'' \textcolor{gray}{/ u_1'}$} ($(u.north -| u-1-7.north west) + (0ex, 1mm+4pt)$);
            \draw[decorate, decoration={brace}] ($(u.north -| u-1-7.north east) + (0ex, 1mm+4pt)$) -- node[above right, xshift=-1ex, yshift=2pt] {$\scriptstyle u_2'' \textcolor{gray}{{}= u_2'}$} ($(u.north east) + (0pt, 1mm+4pt)$);
            
            \draw[decorate, decoration={brace}] ($(u.south -| u-1-1.south east) + (0pt, -2pt)$) -- node[below, yshift=-2pt] {$\scriptstyle \hat{u}_0''\textcolor{gray}{= \hat{u}_0'}$} ($(u.south west) + (0pt, -2pt)$);
            \draw[decorate, decoration={brace}] ($(u.south -| u-1-4.south east) + (0pt, -2pt)$) -- node[below, yshift=-2pt] {$\scriptstyle \hat{u}_1''\textcolor{gray}{/ \hat{u}_1'}$} ($(u.south -| u-1-2.south east) + (0pt, -2pt)$);
            \draw[decorate, decoration={brace}] ($(u.south -| u-1-6.south east) + (0pt, -2pt)$) -- node[below, yshift=-2pt] {$\scriptstyle \hat{u}_2''\textcolor{gray}{= \hat{u}_2'}$} ($(u.south -| u-1-6.south west) + (0pt, -2pt)$);
            \draw[decorate, decoration={brace}] ($(u.south east)+(0pt,-2pt)$) -- node[below, yshift=-2pt] {$\scriptstyle \hat{u}_3''\textcolor{gray}{= \hat{u}_3'}$} ($(u.south -| u-1-7.south east) + (0pt, -2pt)$);
            
            \path[inner sep=2pt]
              let
                \p1=(u-1-4.east),\p2=(u-1-4.west),
                \n1={\x1-\x2-2*\pgfkeysvalueof{/pgf/inner xsep}-\pgflinewidth}
              in
                node[below right=1.5cm and 0cm of u.south west, rectangle, draw, matrix of math nodes, ampersand replacement=\&, text height=.8em, text depth=.2em] (v) {
                |[minimum width=1cm, align=center]| \dots \& c_i \& |[minimum width=1cm, align=center]| \dots \& |[text width=\n1, align=center]| y_\mu' \& c_{\mu + 1} \& |[minimum width=1cm, align=center]| \dots \& c_{j + 1} \& |[minimum width=1cm, align=center]| \dots \\
              };
            \node[left=0pt of v] {$v' = {}$};
            \foreach \i in {1,...,7} {
              \draw ($(v.north -| v-1-\i.north east)$) -- ($(v.south -| v-1-\i.south east)$);
            };
            
            \path[pattern=north east lines] ($(v.south west)-(0pt,1mm+2pt)$) rectangle ($(v.south -| v-1-3.south east)-(0pt,2pt)$);
            
            \path[pattern=north west lines] ($(v.south -| v-1-6.south west)-(0pt,1mm+2pt)$) rectangle ($(v.south east)-(0pt,2pt)$);
            
            \draw[->] ($(v.south west) + (-0.75cm, -0.85cm-2pt)$) .. controls +(0.5cm, 0cm) and +(0cm,-0.75cm) .. node[above, pos=0] {$\rho'$} ($(v.south west) + (1.25cm, -1mm-4pt)$);
            \draw[->] ($(v.south east) + (0.75cm, -0.85cm-2pt)$) .. controls +(-0.5cm, 0cm) and +(0cm,-0.75cm) .. node[above, pos=0] {$\lambda'$} ($(v.south -| v-1-7.south) + (0pt, -1mm-4pt)$);
            
            \draw[->] ($(u.south west) + (0pt, -0.75cm)$) .. controls +(+0.5cm, 0cm) and +(0cm,0.75cm) .. ($(v.north -| v-1-5.north) + (0pt, 2pt)$);
            
            \draw[decorate, decoration={brace}] ($(v.south -| v-1-2.south west) + (0pt, -1mm-4pt)$) -- node[below, yshift=-2pt] {$\scriptstyle v_0'$} ($(v.south west) + (0pt, -1mm-4pt)$);
            \draw[decorate, decoration={brace}] ($(v.south -| v-1-7.south west) + (0ex, -1mm-4pt)$) -- node[below, yshift=-2pt] {$\scriptstyle v_1'$} ($(v.south -| v-1-2.south east) + (0pt, -1mm-4pt)$);
            \draw[decorate, decoration={brace}] ($(v.south east) + (0pt, -1mm-4pt)$) -- node[below, yshift=-2pt] {$\scriptstyle v_2'$} ($(v.south -| v-1-7.south east) + (0ex, -1mm-4pt)$);
            
            \draw[decorate, decoration={brace}] ($(v.north west) + (0pt, 2pt)$) -- node[above, yshift=2pt] {$\scriptstyle \hat{v}_0'$} ($(v.north -| v-1-2.north west) + (0pt, 2pt)$);
            \draw[decorate, decoration={brace}] ($(v.north -| v-1-2.north east) + (0pt, 2pt)$) -- node[above, yshift=2pt] {$\scriptstyle \hat{v}_1'$} ($(v.north -| v-1-5.north west) + (0pt, 2pt)$);
            \draw[decorate, decoration={brace}] ($(v.north -| v-1-6.north west) + (0pt, 2pt)$) -- node[above, yshift=2pt] {$\scriptstyle \hat{v}_2'$} ($(v.north -| v-1-6.north east) + (0pt, 2pt)$);
            \draw[decorate, decoration={brace}] ($(v.north -| v-1-7.north east) + (0pt, 2pt)$) -- node[above, yshift=2pt] {$\scriptstyle \hat{v}_3'$} ($(v.north east)+(0pt,2pt)$);
          \end{tikzpicture}
          \caption{$\rho'(u'')$ is to the left of $x_\mu''$ in $u''$ and $\lambda'(u'')$ is to its right; $\rho_0$ ends directly to the right of $x_\mu''$}\label{sfig:dropAMLeftAndRightRho}
        \end{subfigure}
        \caption{Multiple cases for showing $u'' \sim_{m, \rho, \lambda} u'$ in the proof of \autoref{fct:droppingAMRemainsULike}}
      \end{figure}
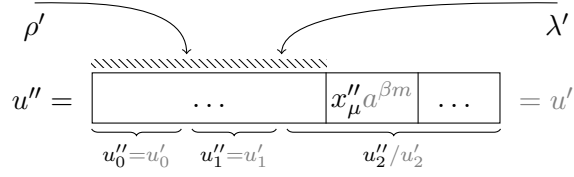
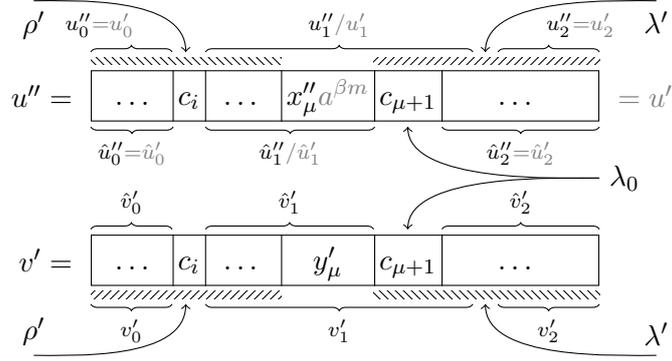
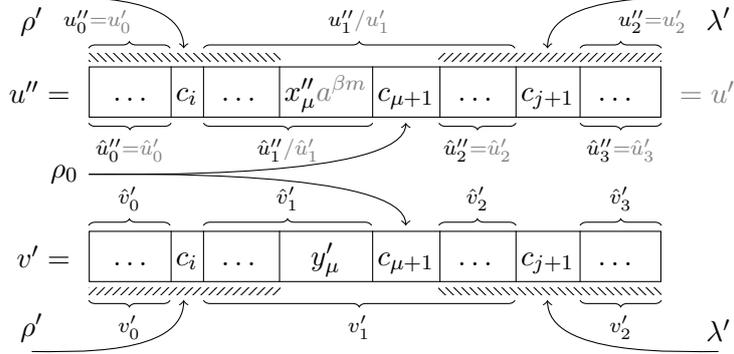
    
      With \autoref{fct:droppingAMRemainsULike} at hand, our goal is to find a $u$-like work $u''$ and a $v$-like word $v''$ with the $\{ \rho, \lambda \}$-factorizations
      \begin{equation*}
        u'' = x_0'' c_1 x_1'' \dots c_d x_d'' \quad\text{and}\quad 
        v'' = y_0'' c_1 y_1'' \dots c_d y_d''\tag{$\ast\ast$}\label{eqn:rhoLambdaFactorization2}
      \end{equation*}
      where we have $|x_k''|_a = |y_k''|_a$ not only for all $k < \mu$ but also for $k = \mu$. In fact, we let $x_k'' = x_k'$ and $y_k'' = y_k'$ for all $0 \leq k < \mu$ and only manipulate $x_k'$ and $y_k'$ for $k \geq \mu$. This way we continue to have $|x_k''|_b = |y_k''|_b$ for all $b \in \Sigma$ and $k < \mu$. Then, using an outer induction on $\mu$ and an inner induction of $\Sigma$, we immediately have a procedure to obtain the sought word that is both $u$-like and $v$-like.
      
      The rest of the proof heavily relies on case distinctions.
      \paragraph*{$\rho$-Position.}
      First, suppose that there is some prefix $\rho_0$ of $\rho\Xinfty$ such that $\rho_0$ ends directly to the right of $x_\mu'$ in $u'$ (either at the position of $c_{\mu + 1}$ or at $+\infty$ for $\mu = d$). Since $u'$ and $v'$ are $\rho$-$\lambda$-compatible, $\rho_0$ must end directly to the right of $y_k'$ in $v'$. Let $\lambda_0$ be the maximal prefix of $\lambda$ such that $\rho_0(u') \leq \lambda_0(u')$ (the prefix is possibly empty), let $(u_0', u_1', u_2')$ be the $\rho_0$-$\lambda_0$-split of $u'$ and let $(v_0', v_1', v_2')$ be the $\rho_0$-$\lambda_0$-split of $v'$. As we have $u' \sim_{m, \rho, \lambda} v'$, we obtain $\alphabet u_1' = \alphabet v_1'$. The situation is depicted in \autoref{fig:nextPosIsRhoPos}. Note, however, that one or both of $\rho_0$ and $\lambda_0$ may end at $+\infty$.
      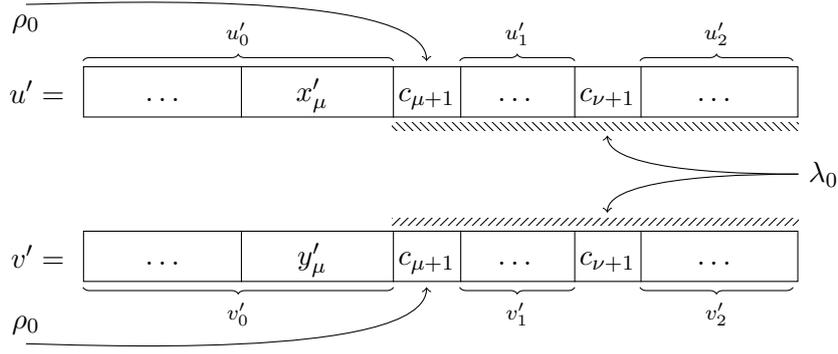
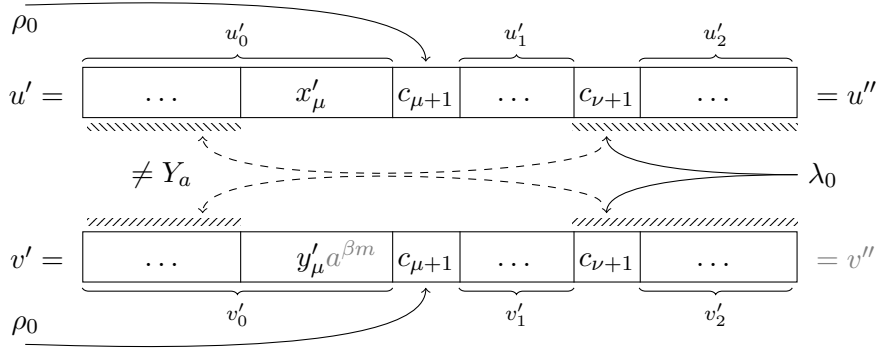
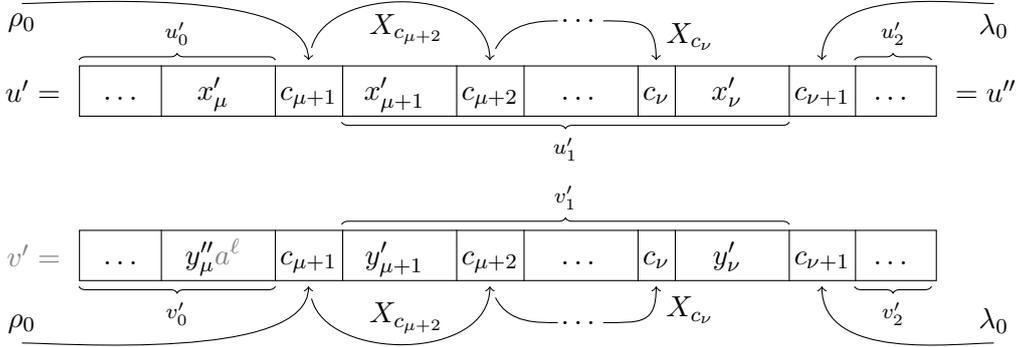
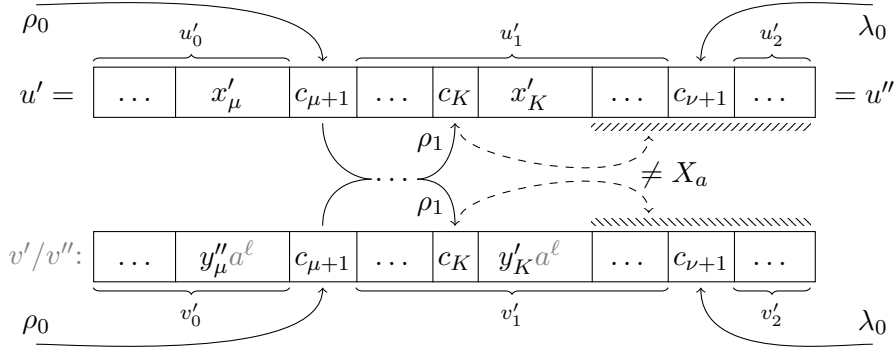
\begin{figure}\centering
        \thisfloatpagestyle{empty}
        \begin{subfigure}{\linewidth}\centering
          \begin{tikzpicture}[baseline=(uLabel.base)]
            \node (uLabel) {$u' = {}$};
            \matrix [right=0cm of uLabel, rectangle, draw, matrix of math nodes, ampersand replacement=\&, inner sep=2pt, text height=.8em, text depth=.2em] (u) {
              |[minimum width=2cm, align=center]| \dots \& |[minimum width=2cm, align=center]| $x_\mu'$ \& c_{\mu + 1} \& |[minimum width=1.5cm, align=center]| \dots \& c_{\nu + 1} \& |[minimum width=2cm, align=center]| \dots \\
            };
          
            \matrix [below right=1.5cm and 0cm of u.south west, rectangle, draw, matrix of math nodes, ampersand replacement=\&, inner sep=2pt, text height=.8em, text depth=.2em] (v) {
              |[minimum width=2cm, align=center]| \dots \& |[minimum width=2cm, align=center]| $y_\mu'$ \& c_{\mu + 1} \& |[minimum width=1.5cm, align=center]| \dots \& c_{\nu + 1} \& |[minimum width=2cm, align=center]| \dots \\
            };
            \node[left=0pt of v] {$v' = {}$};
            \node[right=0pt of v, color=gray] {\phantom{${} = v''$}};
  
            \foreach \i in {1,...,5} {
              \draw ($(u.north -| u-1-\i.north east)$) -- ($(u.south -| u-1-\i.south east)$);
              \draw ($(v.north -| v-1-\i.north east)$) -- ($(v.south -| v-1-\i.south east)$);
            };
          
            \draw[->] ($(u.north west) + (-0.75cm, 0.75cm+2pt)$) .. controls +(0.5cm, 0cm) and +(0cm,1cm) .. node[below, pos=0] {$\rho_0$} ($(u.north -| u-1-3.north)+(0pt,2pt)$);
            \draw[->] ($(v.south west) + (-0.75cm, -0.75cm-2pt)$) .. controls +(0.5cm, 0cm) and +(0cm,-1cm) .. node[above, pos=0] {$\rho_0$} ($(v.south -| v-1-3.south)+(0pt,-2pt)$);
           
            \path[pattern=north west lines] ($(u.south -| u-1-3.south west)+(0pt,-2pt)$) rectangle ($(u.south east)+(0pt,-1mm-2pt)$);
            \path[pattern=north east lines] ($(v.north -| v-1-3.north west)+(0pt,1mm+2pt)$) rectangle ($(v.north east)+(0pt,2pt)$);
            
            \draw[->] ($(u.south east) + (0pt, -0.75cm)$) .. controls +(-0.5cm, 0cm) and +(0cm,-0.5cm) .. node[right, pos=0] {$\lambda_0$} ($(u.south -| u-1-5.south) + (0pt, -1mm-4pt)$);
            \draw[->] ($(v.north east) + (0pt, 0.75cm)$) .. controls +(-0.5cm, 0cm) and +(0cm,0.5cm) .. ($(v.north -| v-1-5.north) + (0pt, 1mm+4pt)$);
  
            \draw[decorate, decoration={brace}] ($(u.north west) + (0pt, 2pt)$) -- node[above, yshift=2pt] {$\scriptstyle u_0'$} ($(u.north -| u-1-2.north east) + (0pt, 2pt)$);
            \draw[decorate, decoration={brace}] ($(u.north -| u-1-4.north west) + (0pt, 2pt)$) -- node[above, yshift=2pt] {$\scriptstyle u_1'$} ($(u.north -| u-1-4.north east) + (0pt, 2pt)$);
            \draw[decorate, decoration={brace}] ($(u.north -| u-1-6.north west) + (0pt, 2pt)$) -- node[above, yshift=2pt] {$\scriptstyle u_2'$} ($(u.north -| u-1-6.north east) + (0pt, 2pt)$);
            
            \draw[decorate, decoration={brace}] ($(v.south -| v-1-2.south east) + (0pt, -2pt)$) -- node[below, yshift=-2pt] {$\scriptstyle v_0'$} ($(v.south west) + (0pt, -2pt)$);
            \draw[decorate, decoration={brace}] ($(v.south -| v-1-4.south east) + (0pt, -2pt)$) -- node[below, yshift=-2pt] {$\scriptstyle v_1'$} ($(v.south -| v-1-4.south west) + (0pt, -2pt)$);
            \draw[decorate, decoration={brace}] ($(v.south -| v-1-6.south east) + (0pt, -2pt)$) -- node[below, yshift=-2pt] {$\scriptstyle v_2'$} ($(v.south -| v-1-6.south west) + (0pt, -2pt)$);
          \end{tikzpicture}
          \caption{General case}\label{sfig:nextPosIsRhoPos:general}
        \end{subfigure}\\
        \begin{subfigure}{\linewidth}\centering
          \begin{tikzpicture}[baseline=(uLabel.base)]
            \node (uLabel) {$u' = {}$};
            \matrix [right=0cm of uLabel, rectangle, draw, matrix of math nodes, ampersand replacement=\&, inner sep=2pt, text height=.8em, text depth=.2em] (u) {
                |[minimum width=2cm, align=center]| \dots \& |[minimum width=2cm, align=center]| $x_\mu'$ \& c_{\mu + 1} \& |[minimum width=1.5cm, align=center]| \dots \& c_{\nu + 1} \& |[minimum width=2cm, align=center]| \dots \\
              };
            \node[right=0pt of u] {${} = u''$};
          
            \matrix [below right=1.5cm and 0cm of u.south west, rectangle, draw, matrix of math nodes, ampersand replacement=\&, inner sep=2pt, text height=.8em, text depth=.2em] (v) {
                |[minimum width=2cm, align=center]| \dots \& |[minimum width=2cm, align=center]| $y_\mu'$\makebox[0pt][l]{\textcolor{gray}{$a^{\beta m}$}} \& c_{\mu + 1} \& |[minimum width=1.5cm, align=center]| \dots \& c_{\nu + 1} \& |[minimum width=2cm, align=center]| \dots \\
              };
            \node[left=0pt of v] {$v' = {}$};
            \node[right=0pt of v, color=gray] {${} = v''$};
  
            \foreach \i in {1,...,5} {
                \draw ($(u.north -| u-1-\i.north east)$) -- ($(u.south -| u-1-\i.south east)$);
                \draw ($(v.north -| v-1-\i.north east)$) -- ($(v.south -| v-1-\i.south east)$);
              };
          
            \draw[->] ($(u.north west) + (-0.75cm, 0.75cm+2pt)$) .. controls +(0.5cm, 0cm) and +(0cm,1cm) .. node[below, pos=0] {$\rho_0$} ($(u.north -| u-1-3.north)+(0pt,2pt)$);
            \draw[->] ($(v.south west) + (-0.75cm, -0.75cm-2pt)$) .. controls +(0.5cm, 0cm) and +(0cm,-1cm) .. node[above, pos=0] {$\rho_0$} ($(v.south -| v-1-3.south)+(0pt,-2pt)$);
           
            \path[pattern=north west lines] ($(u.south -| u-1-5.south west)+(0pt,-2pt)$) rectangle ($(u.south east)+(0pt,-1mm-2pt)$);
            \path[pattern=north east lines] ($(v.north -| v-1-5.north west)+(0pt,1mm+2pt)$) rectangle ($(v.north east)+(0pt,2pt)$);
            
            \draw[->] ($(u.south east) + (0pt, -0.75cm)$) .. controls +(-0.5cm, 0cm) and +(0cm,-0.5cm) .. node[right, pos=0] {$\lambda_0$} ($(u.south -| u-1-5.south) + (0pt, -1mm-4pt)$);
            \draw[->] ($(v.north east) + (0pt, 0.75cm)$) .. controls +(-0.5cm, 0cm) and +(0cm,0.5cm) .. ($(v.north -| v-1-5.north) + (0pt, 1mm+4pt)$);

            \path[pattern=north west lines] ($(u.south -| u-1-1.south west)+(0pt,-2pt)$) rectangle ($(u.south -| u-1-1.south east)+(0pt,-1mm-2pt)$);
            \path[pattern=north east lines] ($(v.north -| v-1-1.north west)+(0pt,1mm+2pt)$) rectangle ($(v.north -| v-1-1.north east)+(0pt,2pt)$);
            \draw[->, shorten <= 1ex, dashed] ($(u.south -| u-1-5.south) + (0pt, -1mm-4pt)$) .. controls +(-0.25cm, -0.5cm) and +(0cm,-0.75cm) .. ($(u.south -| u-1-1.south east)+(-0.5cm,-1mm-4pt)$);
            \draw[->, shorten <= 1ex, dashed] ($(v.north -| v-1-5.north) + (0pt, 1mm+4pt)$) .. controls +(-0.25cm, 0.5cm) and +(0cm,0.75cm) .. ($(v.north -| v-1-1.north east)+(-0.5cm,1mm+4pt)$);
            \node[left] at ($(u-1-1.south east)!0.5!(v-1-1.north east)+(-0.5cm,0pt)$) {$\neq Y_a$};

            \draw[decorate, decoration={brace}] ($(u.north west) + (0pt, 2pt)$) -- node[above, yshift=2pt] {$\scriptstyle u_0'$} ($(u.north -| u-1-2.north east) + (0pt, 2pt)$);
            \draw[decorate, decoration={brace}] ($(u.north -| u-1-4.north west) + (0pt, 2pt)$) -- node[above, yshift=2pt] {$\scriptstyle u_1'$} ($(u.north -| u-1-4.north east) + (0pt, 2pt)$);
            \draw[decorate, decoration={brace}] ($(u.north -| u-1-6.north west) + (0pt, 2pt)$) -- node[above, yshift=2pt] {$\scriptstyle u_2'$} ($(u.north -| u-1-6.north east) + (0pt, 2pt)$);
            
            \draw[decorate, decoration={brace}] ($(v.south -| v-1-2.south east) + (0pt, -2pt)$) -- node[below, yshift=-2pt] {$\scriptstyle v_0'$} ($(v.south west) + (0pt, -2pt)$);
            \draw[decorate, decoration={brace}] ($(v.south -| v-1-4.south east) + (0pt, -2pt)$) -- node[below, yshift=-2pt] {$\scriptstyle v_1'$} ($(v.south -| v-1-4.south west) + (0pt, -2pt)$);
            \draw[decorate, decoration={brace}] ($(v.south -| v-1-6.south east) + (0pt, -2pt)$) -- node[below, yshift=-2pt] {$\scriptstyle v_2'$} ($(v.south -| v-1-6.south west) + (0pt, -2pt)$);
          \end{tikzpicture}
          \caption{$a \in \alphabet u_1' = \alphabet v_1'$}\label{sfig:nextPosIsRhoPos:aInMiddle}
        \end{subfigure}\\
        \begin{subfigure}{\linewidth}\centering
          \begin{tikzpicture}[baseline=(uLabel.base)]
            \node (uLabel) {$u' = {}$};
            \matrix [right=0cm of uLabel, rectangle, draw, matrix of math nodes, ampersand replacement=\&, inner sep=2pt, text height=.8em, text depth=.2em] (u) {
                |[minimum width=1cm, align=center]| \dots \& |[minimum width=1.5cm, align=center]| $x_\mu'$ \& c_{\mu + 1} \& |[minimum width=1.5cm, align=center]| $x_{\mu + 1}'$ \& c_{\mu + 2} \& |[minimum width=1.5cm, align=center]| \dots \& c_\nu \& |[minimum width=1.5cm, align=center]| $x_{\nu}'$ \& c_{\nu + 1} \& |[minimum width=1cm, align=center]| \dots \\
              };
            \node[right=0pt of u] {${} = u''$};
          
            \matrix [below right=1.5cm and 0cm of u.south west, rectangle, draw, matrix of math nodes, ampersand replacement=\&, inner sep=2pt, text height=.8em, text depth=.2em] (v) {
                |[minimum width=1cm, align=center]| \dots \& |[minimum width=1.5cm, align=center]| $y_\mu''$\textcolor{gray}{$a^\ell$} \& c_{\mu + 1} \& |[minimum width=1.5cm, align=center]| $y_{\mu + 1}'$ \& c_{\mu + 2} \& |[minimum width=1.5cm, align=center]| \dots \& c_\nu \& |[minimum width=1.5cm, align=center]| $y_{\nu}'$ \& c_{\nu + 1} \& |[minimum width=1cm, align=center]| \dots \\
              };
            \node[left=0pt of v, color=gray] {$v' = {}$};
  
            \foreach \i in {1,...,9} {
                \draw ($(u.north -| u-1-\i.north east)$) -- ($(u.south -| u-1-\i.south east)$);
                \draw ($(v.north -| v-1-\i.north east)$) -- ($(v.south -| v-1-\i.south east)$);
              };
          
            \draw[->] ($(u.north west) + (-0.75cm, 0.75cm+2pt)$) .. controls +(0.5cm, 0cm) and +(0cm,1cm) .. node[below, pos=0] {$\rho_0$} ($(u.north -| u-1-3.north)+(0pt,2pt)$);
            \draw[->] ($(v.south west) + (-0.75cm, -0.75cm-2pt)$) .. controls +(0.5cm, 0cm) and +(0cm,-1cm) .. node[above, pos=0] {$\rho_0$} ($(v.south -| v-1-3.south)+(0pt,-2pt)$);
            
            \draw[->, shorten <= 1ex] ($(u.north -| u-1-3.north)+(0pt,2pt)$) .. controls +(0.25cm, 1cm) and +(0cm,1cm) .. node[below] {$X_{c_{\mu + 2}}$} ($(u.north -| u-1-5.north)+(0pt,2pt)$);
            \node[anchor=base, inner sep=1pt] at ($(u.north -| u-1-6.north)+(0pt,2pt+0.5cm)$) (dotsU) {$\dots$};
            \draw[-, shorten <= 1ex] ($(u.north -| u-1-5.north)+(0pt,2pt)$) .. controls +(0.25cm, 0.5cm) and +(-0.25cm,0cm) .. (dotsU.west);
            \draw[->] (dotsU.east) .. controls +(0.5cm, 0pt) and +(-0cm,0.5cm) .. node[above right, pos=1] {$X_{c_\nu}$} ($(u.north -| u-1-7.north)+(0pt,2pt)$);
            
            \draw[->, shorten <= 1ex] ($(v.south -| v-1-3.south)+(0pt,-2pt)$) .. controls +(0.25cm, -1cm) and +(0cm,-1cm) .. node[above] {$X_{c_{\mu + 2}}$} ($(v.south -| v-1-5.south)+(0pt,-2pt)$);
            \node[anchor=base, inner sep=1pt] at ($(v.south -| v-1-6.south)+(0pt,-2pt-0.5cm)$) (dotsV) {$\dots$};
            \draw[-, shorten <= 1ex] ($(v.south -| v-1-5.south)+(0pt,-2pt)$) .. controls +(0.25cm, -0.5cm) and +(-0.25cm,0cm) .. (dotsV.west);
            \draw[->] (dotsV.east) .. controls +(0.5cm, 0pt) and +(-0cm,-0.5cm) .. node[below right, pos=1] {$X_{c_\nu}$} ($(v.south -| v-1-7.south)+(0pt,-2pt)$);
           
            \draw[->] ($(u.north east) + (0.75cm, 0.75cm+2pt)$) .. controls +(-0.5cm, 0cm) and +(0cm,1cm) .. node[below, pos=0] {$\lambda_0$} ($(u.north -| u-1-9.north) + (0pt, +2pt)$);
            \draw[->] ($(v.south east) + (0.75cm, -0.75cm-2pt)$) .. controls +(-0.5cm, 0cm) and +(0cm,-1cm) .. node[above, pos=0] {$\lambda_0$} ($(v.south -| v-1-9.south) + (0pt, -2pt)$);
            
            \draw[decorate, decoration={brace}] ($(u.north west) + (0pt, 2pt)$) -- node[above, yshift=2pt] {$\scriptstyle u_0'$} ($(u.north -| u-1-2.north east) + (0pt, 2pt)$);
            \draw[decorate, decoration={brace}] ($(u.south -| u-1-8.south east) + (0pt, -2pt)$) -- node[below, yshift=-2pt] {$\scriptstyle u_1'$} ($(u.south -| u-1-4.south west) + (0pt, -2pt)$);
            \draw[decorate, decoration={brace}] ($(u.north -| u-1-10.north west) + (0pt, 2pt)$) -- node[above, yshift=2pt] {$\scriptstyle u_2'$} ($(u.north -| u-1-10.north east) + (0pt, 2pt)$);
            
            \draw[decorate, decoration={brace}] ($(v.south -| v-1-2.south east) + (0pt, -2pt)$) -- node[below, yshift=-2pt] {$\scriptstyle v_0'$} ($(v.south west) + (0pt, -2pt)$);
            \draw[decorate, decoration={brace}] ($(v.north -| v-1-4.north west) + (0pt, 2pt)$) -- node[above, yshift=2pt] {$\scriptstyle v_1'$} ($(v.north -| v-1-8.north east) + (0pt, 2pt)$);
            \draw[decorate, decoration={brace}] ($(v.south -| v-1-10.south east) + (0pt, -2pt)$) -- node[below, yshift=-2pt] {$\scriptstyle v_2'$} ($(v.south -| v-1-10.south west) + (0pt, -2pt)$);
          \end{tikzpicture}
          \caption{$a \in \alphabet u_1' = \alphabet v_1', |x_\mu'|_a < |y_\mu'|_a$, more positions reached by prefixes of $\rho$ drawn}\label{sfig:nextPosIsRhoPos:moreRhoPositions}
        \end{subfigure}\\
        \begin{subfigure}{\linewidth}\centering
          \begin{tikzpicture}[baseline=(uLabel.base)]
            \node (uLabel) {$u' = {}$};
            \matrix [right=0cm of uLabel, rectangle, draw, matrix of math nodes, ampersand replacement=\&, inner sep=2pt, text height=.8em, text depth=.2em] (u) {
                |[minimum width=1cm, align=center]| \dots \& |[minimum width=1.5cm, align=center]| $x_\mu'$ \& c_{\mu + 1} \& |[minimum width=1cm, align=center]| \dots \& c_{K} \& |[minimum width=1.5cm, align=center]| $x_K'$ \& |[minimum width=1cm, align=center]| \dots \& c_{\nu + 1} \& |[minimum width=1cm, align=center]| \dots \\
              };
            \node[right=0pt of u] {${} = u''$};
          
            \matrix [below right=1.5cm and 0cm of u.south west, rectangle, draw, matrix of math nodes, ampersand replacement=\&, inner sep=2pt, text height=.8em, text depth=.2em] (v) {
                |[minimum width=1cm, align=center]| \dots \& |[minimum width=1.5cm, align=center]| $y_\mu'' \textcolor{gray}{a^\ell}$ \& c_{\mu + 1} \& |[minimum width=1cm, align=center]| \dots \& c_{K} \& |[minimum width=1.5cm, align=center]| $y_K' \textcolor{gray}{a^\ell}$ \& |[minimum width=1cm, align=center]| \dots \& c_{\nu + 1} \& |[minimum width=1cm, align=center]| \dots \\
              };
            \node[left=0pt of v, color=gray] {$v'/v''$:~};
  
            \foreach \i in {1,...,8} {
                \draw ($(u.north -| u-1-\i.north east)$) -- ($(u.south -| u-1-\i.south east)$);
                \draw ($(v.north -| v-1-\i.north east)$) -- ($(v.south -| v-1-\i.south east)$);
              };
          
            \draw[->] ($(u.north west) + (-0.75cm, 0.75cm+2pt)$) .. controls +(0.5cm, 0cm) and +(0cm,1cm) .. node[below, pos=0] {$\rho_0$} ($(u.north -| u-1-3.north)+(0pt,2pt)$);
            \draw[->] ($(v.south west) + (-0.75cm, -0.75cm-2pt)$) .. controls +(0.5cm, 0cm) and +(0cm,-1cm) .. node[above, pos=0] {$\rho_0$} ($(v.south -| v-1-3.south)+(0pt,-2pt)$);

            \node[anchor=base, inner sep=1pt] at ($(u-1-4.south)!0.55!(v-1-4.north)$) (rho1dots) {$\dots$};
            \draw[-] ($(u.south -| u-1-3.south)+(0pt,-2pt)$) .. controls +(0cm, -0.5cm) and +(-0.25cm,0cm) .. (rho1dots.west);
            \draw[-] ($(v.north -| v-1-3.north)+(0pt,2pt)$) .. controls +(0cm, 0.5cm) and +(-0.25cm,0cm) .. (rho1dots.west);
            \draw[->] (rho1dots.east) .. controls +(0.25cm, 0pt) and +(0cm,-0.5cm) .. node[below left, pos=1] {$\rho_1$} ($(u.south -| u-1-5.south)+(0pt,-2pt)$);
            \draw[->] (rho1dots.east) .. controls +(0.25cm, 0pt) and +(0cm,0.5cm) .. node[above left, pos=1] {$\rho_1$} ($(v.north -| v-1-5.north)+(0pt,2pt)$);
            
            \path[pattern=north east lines] ($(u.south -| u-1-7.south west)+(0pt,-2pt)$) rectangle ($(u.south -| u-1-9.south east)+(0pt,-1mm-2pt)$);
            \path[pattern=north west lines] ($(v.north -| v-1-7.north west)+(0pt,+2pt)$) rectangle ($(v.north -| v-1-9.north east)+(0pt,1mm+2pt)$);
            \draw[->, shorten <= 1ex, dashed] ($(u.south -| u-1-5.south) + (0pt, -2pt)$) .. controls +(0.25cm, -0.5cm) and +(0cm,-0.75cm) .. ($(u.south -| u-1-7.south)+(0.25cm,-1mm-4pt)$);
            \draw[->, shorten <= 1ex, dashed] ($(v.north -| v-1-5.north) + (0pt, 2pt)$) .. controls +(0.25cm, 0.5cm) and +(0cm,0.75cm) .. ($(v.north -| v-1-7.north)+(0.25cm,1mm+4pt)$);
            \node[right] at ($(u-1-7.south east)!0.5!(v-1-7.north east)+(-0.5cm,0pt)$) {$\neq X_a$};
           
            \draw[->] ($(u.north east) + (0.75cm, 0.75cm+2pt)$) .. controls +(-0.5cm, 0cm) and +(0cm,1cm) .. node[below, pos=0] {$\lambda_0$} ($(u.north -| u-1-8.north) + (0pt, +2pt)$);
            \draw[->] ($(v.south east) + (0.75cm, -0.75cm-2pt)$) .. controls +(-0.5cm, 0cm) and +(0cm,-1cm) .. node[above, pos=0] {$\lambda_0$} ($(v.south -| v-1-8.south) + (0pt, -2pt)$);
            
            \draw[decorate, decoration={brace}] ($(u.north west) + (0pt, 2pt)$) -- node[above, yshift=2pt] {$\scriptstyle u_0'$} ($(u.north -| u-1-2.north east) + (0pt, 2pt)$);
            \draw[decorate, decoration={brace}] ($(u.north -| u-1-4.north west) + (0pt, 2pt)$) -- node[above, yshift=2pt] {$\scriptstyle u_1'$} ($(u.north -| u-1-7.north east) + (0pt, 2pt)$);
            \draw[decorate, decoration={brace}] ($(u.north -| u-1-9.north west) + (0pt, 2pt)$) -- node[above, yshift=2pt] {$\scriptstyle u_2'$} ($(u.north -| u-1-9.north east) + (0pt, 2pt)$);

            \draw[decorate, decoration={brace}] ($(v.south -| v-1-2.south east) + (0pt, -2pt)$) -- node[below, yshift=-2pt] {$\scriptstyle v_0'$} ($(v.south west) + (0pt, -2pt)$);
            \draw[decorate, decoration={brace}] ($(v.south -| v-1-7.south east) + (0pt, -2pt)$) -- node[below, yshift=-2pt] {$\scriptstyle v_1'$} ($(v.south -| v-1-4.south west) + (0pt, -2pt)$);
            \draw[decorate, decoration={brace}] ($(v.south -| v-1-9.south east) + (0pt, -2pt)$) -- node[below, yshift=-2pt] {$\scriptstyle v_2'$} ($(v.south -| v-1-9.south west) + (0pt, -2pt)$);
          \end{tikzpicture}
          \caption{$a \in \alphabet u_1' = \alphabet v_1', |x_\mu'|_a < |y_\mu'|_a, a \in \alphabet x_K' \cup \alphabet y_K'$}\label{sfig:nextPosIsRhoPos:aInAlphabetXYK}
        \end{subfigure}
        \caption{A prefix of $\rho\Xinfty$ ends directly to the right of $x_\mu'$/$y_\mu'$ in $u'$ and $v'$.}\label{fig:nextPosIsRhoPos}
      \end{figure}

      First, consider the case $a \not\in \alphabet u_1' = \alphabet v_1'$. This, in particular, includes the case $\lambda_0(u') = \rho_0(u')$ (and $\lambda_0(v') = \rho_0(v')$) and, thus, the case $\lambda_0(u') = \rho_0(u') = +\infty$ (and $\lambda_0(v') = \rho_0(v') = +\infty$). We obtain
      \begin{align}
        |x_0' c_1 \dots x_{\mu - 1}' c_\mu|_a + |x_\mu'|_a = |u'_0|_a \equiv |v'_0|_a
        &= |y_0' c_1 \dots y_{\mu - 1}' c_\mu|_a + |y_\mu'|_a \nonumber\\
        &= |x_0' c_1 \dots x_{\mu - 1}' c_\mu|_a + |y_\mu'|_a \mod m\tag{$\dagger$}\label{eqn:numberOfAs}
      \end{align}
      and, thus, $|x_\mu'|_a \equiv |y_\mu'|_a \bmod m$ by the definition of $\sim_{m, \rho, \lambda}$ and as we have $|x_k'|_a = |y_k'|_a$ for all $0 \leq k < \mu$. In other words, we either have $x_\mu' = x_\mu'' a^{\beta m}$ (up to permutation) for some $\beta > 0$ and $x_\mu''$ with $|x_\mu''|_a = |y_\mu'|_a$ (i.\,e.\ we have more $a$s in $x_\mu'$ than in $y_\mu'$) or $y_\mu' = y_\mu'' a^{\beta m}$ with $|y_\mu''|_a = |x_\mu'|_a$ (i.\,e.\ we have more $a$s in $y_\mu'$). In the former case, we already have defined $x_\mu''$ and let $x_k'' = x_k'$ for $k > \mu$. The resulting word is $u$-like by \autoref{fct:droppingAMRemainsULike} and we can let $v'' = v'$. Note that we now have $|x_\mu''|_a = |y_\mu''|_a$. The latter case (we have more $a$s in $y_\mu'$) can be handled in the same way.
      
      From now on, we may assume $a \in \alphabet u_1' = \alphabet v_1'$ without loss of generality. This implies, in particular, that $u_1'$ and $v_1'$ are not empty and, thus, $\rho_0(u') < \lambda_0(u')$ and $\rho_0(v') < \lambda_0(v')$.
      
      Here, we first handle the case that we have $|x_\mu'|_a > |y_\mu'|_a$. The idea is that we add $a^m$ blocks to $y_\mu'$ until we have more $a$s in $y_\mu''$ than in $x_\mu'$. More precisely, we let $y_\mu'' = y_\mu' a^{\beta m}$ for some $\beta$ such that $\beta m \geq |x_\mu'|_a - |y_\mu'|_a$, $y_k'' = y_k'$ for $\mu < k \leq d$ and $u'' = u'$ (which means that $u''$ remains $u$-like). This situation is depicted in \autoref{sfig:nextPosIsRhoPos:aInMiddle}. Note, however, that $\lambda_0$ may also end at $+\infty$ (in all words).
      
      We have to show that $v''$ is $v$-like. First, we show that $v''$ is $\rho$-$\lambda$-compatible to $v'$ (and, thus, also to $u'$, $u$ and $v$). Consider a prefix $\rho'$ of $\rho \Xinfty$ which ends at some $c_k$ in the $\{ \rho, \lambda \}$-factorization (\ref{eqn:rhoLambdaFactorization1}) of $v'$. We show that it ends at the same $c_k$ in factorization (\ref{eqn:rhoLambdaFactorization2}) of $v''$ and then we do the same for an arbitrary prefix $\lambda'$ of $\lambda \Yinfty$. The reader may observe that this shows that $v''$ and $v'$ are $\rho$-$\lambda$-compatible and that, in fact, the factorization in (\ref{eqn:rhoLambdaFactorization2}) is the $\{ \rho, \lambda \}$-factorization of $v''$ (the only other case is that $\rho'$ or $\lambda'$ ends at $\pm \infty$ both in $v'$ and in $v''$).
      
      This obviously is true if $\rho'$ ends to the left of $y_\mu'$ in $v'$ as this part is the same as the part to the left of $y_\mu'' = y_\mu' a^{\beta m}$ in $v''$. Note that there is an $a$ in $x_\mu'$ (since we are in the case $|x_\mu'|_a > |y_\mu'|_a \geq 0$) and that, therefore, $c_{\mu + 1} \neq a$ or, in other words, the next instruction after $\rho_0$ in $\rho \Xinfty$ cannot be an $X_a$. Thus, $\rho'$ also ends at the position of the same $c_k$ in $v''$ if it ends to the right of $y_\mu'$ in $v'$. The argument for $\lambda'$ is very similar: here it suffices to observe that the next instruction after $\lambda_0$ in $\lambda \Yinfty$ cannot be a $Y_a$ (as we must be to the left of $x_\mu'$ in $u'$, afterwards, by the definition of $\lambda_0$ and there again is an $a$ in $x_\mu'$).
      
      Having established that $v'$ and $v''$ are $\rho$-$\lambda$-compatible, we continue with showing $v'' \sim_{m, \rho, \lambda} v'$ (which shows \autoref{itm:vLike:rhoLambdaSim} of the definition of $v$-like words). For this, consider a prefix $\rho'$ of $\rho \Xinfty$ and a prefix $\lambda'$ of $\lambda \Yinfty$ with $\rho'(v') \leq \lambda'(v')$ (and, thus, $\rho'(v'') \leq \lambda'(v'')$). If $\lambda'$ ends to the left of $y_\mu'$ in $v'$ (and, thus, also to the left of $y_\mu'' = y_\mu' a^{\beta m}$ in $v''$), there is nothing to show as the left and the middle parts of the $\rho'$-$\lambda'$-splits of $v'$ and $v''$ are the same, respectively, and the right parts only differ because the one belonging to $v''$ has $\beta m$ many more $a$s (which is obviously a multiple of $m$, so the number of $a$s is the same modulo $m$). Symmetrically, we are also done if $\rho'$ ends to the right of $y_\mu'$ in $v'$ (and, thus, also to the right of $y_\mu'' = y_\mu a^{\beta m}$ in $v''$).
      
      Thus, it remains the case that $\rho'$ ends to the left of $y_\mu'$ in $v'$ (and of $y_\mu'' = y_\mu' a^{\beta m}$ in $v''$) and $\lambda'$ ends to its right. By the choice of $\lambda_0$, we then must have $\lambda_0(v') \leq \lambda'(v')$ (and, thus, also $\lambda_0(v'') \leq \lambda'(v'')$). In particular, the left parts of the $\rho'$-$\lambda'$-splits of $v'$ and $v''$ are the same word and the same is true for the right parts. The only difference is in the middle parts as the one belonging to $v''$ contains the additional $a^{\beta m}$ block. However, $v_1'$ is an infix of both middle parts and contains an $a$, so the alphabet of the middle part has not changed by adding the $a^{\beta m}$ block and we are done.
      
      We continue by showing \autoref{itm:vLike:rhoRDescent} of the definition of $v$-like words for $v''$. Let $\lambda'$ again be an arbitrary prefix of $\lambda \Yinfty$. If it ends to the right of $y_\mu'' = y_\mu' a^{\beta m}$ in $v''$, there is nothing to show as we have not changed that part; the right parts of the $\lambda'$-splits of $v'$ and $v''$ coincide and, thus, have the same image under $\varphi$. For the case that $\lambda'$ ends to the left of $y_\mu'' = y_\mu' a^{\beta m}$ in $v''$ and, thus, also to the left of $y_\mu'$ in $v'$, observe that we have
      \[
        y_\mu' c_{\mu + 1} v_1' \, c_{j + 1} y_{j + 1}' \dots c_d y_d' \L c_{j + 1} y_{j + 1}' \dots c_d y_d'
      \]
      (i.\,e.\ the suffix of $v'$ starting from $y_\mu'$ is $\mathcal{L}$-equivalent to the suffix starting from $c_{j + 1}$, where $\lambda_0$ ends -- or, if $\lambda_0$ is empty, we have $\mathcal{L}$-equivalence to $1$)
      because the $\{ \rho, \lambda \}$-factorization (\ref{eqn:rhoLambdaFactorization1}) of $v'$ is a sub-factorization of its $\mathcal{L}$-factorization. By (the left-right symmetric dual of) \autoref{fct:insertAM}, we obtain
      \[
        \varphi(y_\mu' c_{\mu + 1} v_1' \, c_{j + 1} y_{j + 1}' \dots c_d y_d') =
        \varphi(y_\mu' \, \textcolor{gray}{a^{\beta m}} \, c_{\mu + 1} v_1' \, c_{j + 1} y_{j + 1}' \dots c_d y_d')
      \]
      (i.\,e.\ adding the $a^{\beta m}$ block has not changed the image under $\varphi$ of the right hand side of any $\lambda'$-split of $v''$)
      and this establishes \autoref{itm:vLike:rhoRDescent}.
      
      From now on, we may assume $|x_\mu'|_a < |y_\mu'|_a$ (in addition to $a \in \alphabet u_1' = \alphabet v_1'$). In other words, we may write $y_\mu' = y_\mu'' a^\ell$ for some $y_\mu''$ and $\ell > 0$ with $|y_\mu''|_a = |x_\mu'|_a$ (up to permutation). This already defines $y_\mu''$. The other $y_k''$ will be defined below and we again let $u'' = u'$.
      
      As we have $a \in \alphabet u_1'$, the middle part $u_1'$ cannot be empty and there is some $\nu$ (with $\mu < \nu$) such that $u_1' = x_{\mu + 1}' c_{\mu + 2} x_{\mu + 2}' \dots c_\nu x_\nu'$ and $v_1' = y_{\mu + 1}' c_{\mu + 2} y_{\mu + 2}' \dots c_\nu y_\nu'$. Note that no prefix of $\lambda \Yinfty$ can end at the position of any $c_{\mu + 2}, \dots, c_\nu$ (by the choice of $\lambda_0$). Thus, for each of them, there is a prefix of $\rho \Xinfty$ ending at their position.
      The situation is depicted in \autoref{sfig:nextPosIsRhoPos:moreRhoPositions} but note that $\lambda_0$ may still also end at $+\infty$.
      
      If there is some $K$ with $\mu < K \leq \nu$ and $a \in \alphabet x_K' \cup \alphabet y_K'$, we move the excessive $a^\ell$ block from $y_\mu' = y_\mu'' a^\ell$ there, i.\,e.\ we let $y_K'' = y_K' a^\ell$ and keep $y_k'' = y_k'$ for $k \not \in \{ \mu, K \}$ (and $x_k'' = x_k'$ for all $k$ as stated above). Let $\rho_1$ be the prefix of $\rho \Xinfty$ ending at the position of $c_K$ (in $u'$ and $v'$). Note that $\rho_0$ is a prefix of $\rho_1$. Since there is an $a$ in $x_K'$ or in $y_K'$ (or both), we observe that the next instruction in $\rho \Xinfty$ after $\rho_1$ cannot be an $X_a$. The situation is depicted in \autoref{sfig:nextPosIsRhoPos:aInAlphabetXYK}. Note, however, that $\lambda_0$ still may end at $+\infty$ and that the next instruction after $\rho_1$ may end anywhere to the right of $x_K'$ in $u'$ and $y_K'$ in $v'$.
      
      We have to show that $v''$ remains $v$-like. By the above observation (on the next instruction in $\rho \Xinfty$ after $\rho_1$), we obtain that any prefix $\rho'$ of $\rho \Xinfty$ ends at the position of $c_k$ in $v'$ if and only if it does in $v''$ (and the same is true if it ends at $\pm\infty$). Using a very similar argumentation for the next instruction of $\lambda \Yinfty$ after $\lambda_0$ (we have an $a$ in $y_\mu' = y_\mu'' a^\ell$; compare to \autoref{sfig:nextPosIsRhoPos:aInMiddle} and the corresponding argument above), we also obtain that any prefix $\lambda'$ of $\lambda \Yinfty$ ends at the position of some $c_k$ in $v'$ if and only if it does in $v''$. This shows that $v'$ and $v''$ are $\rho$-$\lambda$-compatible and that the $\{ \rho, \lambda \}$-factorization of $v''$ is given by (\ref{eqn:rhoLambdaFactorization2}).
      
      This already shows the first half of $v'' \sim_{m, \rho, \lambda} v'$ (and, thus, \autoref{itm:vLike:rhoLambdaSim} in the definition of $v$-like words). To show the second part, let $\rho'$ be an arbitrary prefix of $\rho \Xinfty$ and let $\lambda'$ be an arbitrary prefix of $\lambda \Yinfty$ with $\rho'(v') \leq \lambda'(v')$ (and, thus, $\rho'(v'') \leq \lambda'(v'')$). Similar to the corresponding cases above, we are done if $\rho'$ ends to the right of $y_K'$ in $v'$ and, therefore, also to the right of $y_K''$ in $v''$ or if $\lambda'$ ends to the left of $y_\mu'$ in $v'$ and, therefore, also to the left of $y_\mu''$ in $v''$. In these cases, it suffices to observe that moving the $a^\ell$ block to the right happens entirely within the right or the left part of the $\rho'$-$\lambda'$-split of $v'$ and that the other parts remain unchanged.
      
      Thus, assume that $\rho'$ ends to the left of $y_K'$ in $v'$ and, therefore, to the left of $y_K'' = y_K' a^\ell$ in $v''$. By the choice of $\lambda_0$, this implies that $\lambda'$ must not end to the left of the position $\lambda_0$ ends at. In particular, $y_K'$ is an infix of the middle parts of the $\rho'$-$\lambda'$-splits of $v'$ and $v''$. If we have $a \in \alphabet y_K'$, we continue to have an $a$ in the alphabet of the middle parts of the $\rho'$-$\lambda'$-splits of $v'$ and $v''$ and there is nothing more to show. If we have $a \not\in \alphabet y_K'$, there must be an $a$ in $\alphabet x_K'$. Since $x_K'$ is an infix of the middle part of the $\rho'$-$\lambda'$-split of $u'$ (and since we have $v' \sim_{m, \rho, \lambda} u'$), there must also be an $a$ in the middle part of the $\rho'$-$\lambda'$-split of $v'$. Thus, moving the $a^\ell$ block to the right in $v''$ has not increased the alphabet of the middle part of its $\rho'$-$\lambda'$-split and, again, there is nothing more to show.
      
      This proves that $v''$ satisfies \autoref{itm:vLike:rhoLambdaSim} in the definition of $v$-like words and it remains to show \autoref{itm:vLike:rhoRDescent}. For this let $\lambda'$ be an arbitrary prefix of $\lambda \Yinfty$. If it ends to the right of $y_K'$ in $v'$ and, therefore, to the right of $y_K'' = y_K' a^\ell$ in $v''$, the right parts of the $\lambda'$-splits of $v'$ and $v''$ coincide and have, thus, the same image under $\varphi$. By the choice of $\lambda_0$, the only other case is that $\lambda'$ ends to the left of $y_\mu' = y_\mu'' a^\ell$ in $v'$ of $y_\mu''$ in $v'$. This also shows that $y_\mu' c_{\mu + 1} y_{\mu + 1} \dots c_K y_K'$ is an infix of a factor of the $\mathcal{L}$-factorization of $v'$. Since we may permute the letters in such a factor arbitrary without changing the image under $\varphi$ by (the left-right dual of) \autoref{fct:permuteR}, moving the $a^\ell$ block to the right has not changed the image under $\varphi$ and we are done.
      
      The remaining case is that we have $a \not\in \alphabet x_k' \cup \alphabet y_k'$ for all $\mu < k \leq \nu$. In this case, let $\rho_1 = \rho_0 X_{c_{\mu + 2}} \dots X_{c_\nu}$ and consider the $\rho_1$-$\lambda_0$-split of $u'$ and $v'$ (refer to \autoref{sfig:nextPosIsRhoPos:moreRhoPositions}). The middle parts of these splits are $x_\nu'$ and $y_\nu'$, respectively. In our case, we have, in particular, $a \not\in \alphabet x_\nu' = \alphabet y_\nu'$ (since we have $u' \sim_{m, \rho, \lambda} v'$) and obtain
      \begin{align*}
        |u_0' \, c_{\mu + 1} x_{\mu + 1}' \dots c_{\nu - 1} x_{\nu - 1}' c_\nu|_a \equiv
        |v_0' \, c_{\mu + 1} y_{\mu + 1}' \dots c_{\nu - 1} y_{\nu - 1}' c_\nu|_a \mod m
      \end{align*}
      by considering the left parts of the $\rho_1$-$\lambda_0$-splits, which implies $|u_0'|_a \equiv |v_0'|_a \bmod m$. Using the same calculation as in (\ref{eqn:numberOfAs}), we obtain $|x_\mu|_a \equiv |y_\mu'|_a = |y_\mu'' a^\ell|_a = |x_\mu'|_a + \ell \bmod m$ and, thus, that $\ell$ is a multiple of $m$. This allows us to simply drop the excessive $a^\ell$ block and obtain a $v$-like word $v''$ as stated in \autoref{fct:droppingAMRemainsULike} (very similar to what we did in the case with $a \not\in \alphabet u_1' = \alphabet v_1'$ above).

      \paragraph*{$\lambda$-Position.}
      This covers all cases where there is a prefix of $\rho \Xinfty$ which ends directly to the right of $x_\mu'$ in $u'$. So, assume from now on that there is no such prefix. Instead, there must be a prefix $\lambda_0$ of $\lambda \Yinfty$ which ends directly to the right of $x_\mu'$ in $u'$ and, thus, also directly to the right of $y_\mu'$ in $v'$. In fact, since there is indeed a prefix of $\rho \Xinfty$ ending at $+\infty$, we have that $\lambda_0$ is a non-empty, proper prefix of $\lambda \Yinfty$ and that we have $\mu < d$.
      
      The case that we have $|x_\mu'|_a < |y_\mu'|_a$ is very similar to the above case where we have $|x_\mu'|_a > |y_\mu'|_a$: we will just add $a^m$ blocks to $x_\mu'$ until we have more $a$s there, i.\,e.\ we let $x_\mu'' = x_\mu' a^{\beta m}$ for some $\beta$ such that $\beta m \geq |y_\mu'|_a - |x_\mu'|_a$ (and let $x_k'' = x_k'$ for all $k \neq \mu$ as well as $v'' = v'$). The situation is depicted in \autoref{sfig:nextPosIsLambdaPos:moreAsInV}.
      
      We have to show that $u''$ remains $u$-like. Let $\rho_0$ be the maximal prefix of $\rho \Xinfty$ such that $\rho_0$ ends to the left of $x_\mu'$ in $u'$ and, thus, to the left of $y_\mu'$ in $v'$ (in fact, $\rho_0$ must already be a prefix of $\rho$). As we have $0 \leq |x_\mu'|_a < |y_\mu'|_a$, there is an $a$ in $y_\mu'$ and, thus, in the middle part $u_1'$ of the $\rho_0$-$\lambda_0$-split of $v'$. In particular, the next instruction after $\rho_0$ in $\rho \Xinfty$ cannot be an $X_a$ and the next instruction after $\lambda_0$ in $\lambda_0 \Yinfty$ cannot be a $Y_a$. This shows that $u''$ remains $\rho$-$\lambda$-compatible to $u'$.
      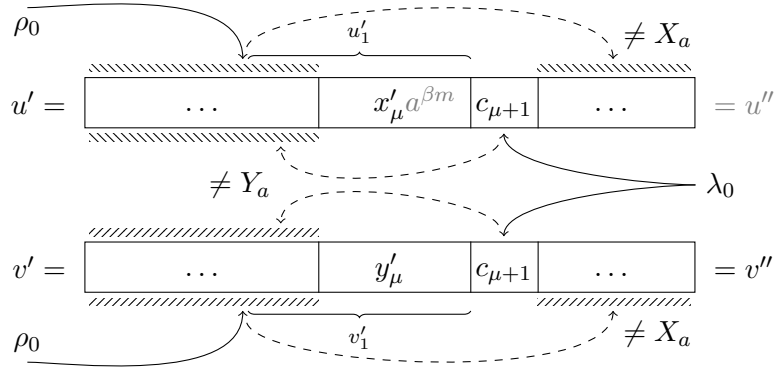
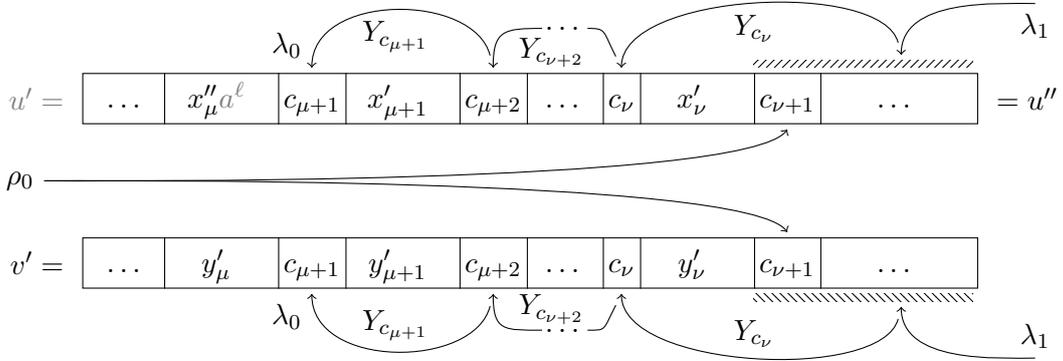
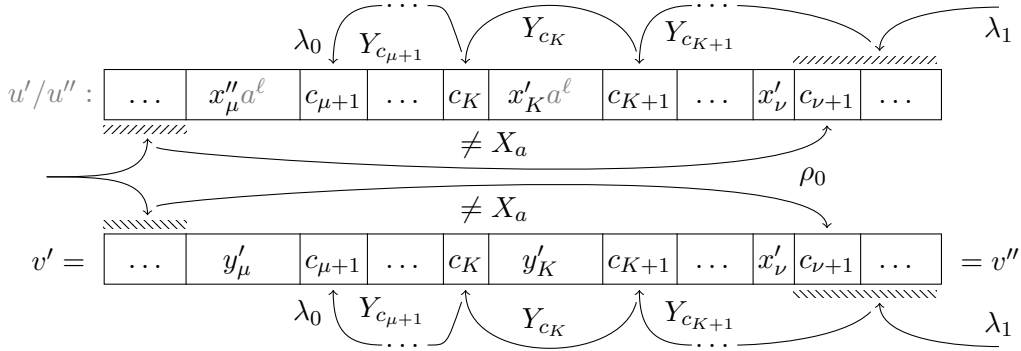
\begin{figure}\centering
        \begin{subfigure}{\linewidth}\centering
          \begin{tikzpicture}[baseline=(uLabel.base)]
            \node (uLabel) {$u' = {}$};
            \matrix [right=0cm of uLabel, rectangle, draw, matrix of math nodes, ampersand replacement=\&, inner sep=2pt, text height=.8em, text depth=.2em] (u) {
              |[minimum width=3cm, align=center]| \dots \& |[minimum width=2cm, align=center]| $x_\mu'$\makebox[0pt][l]{\textcolor{gray}{$a^{\beta m}$}} \& c_{\mu + 1} \& |[minimum width=2cm, align=center]| \dots \\
            };
            \node[right=0pt of u, color=gray] {${} = u''$};
            
            \matrix [below right=1.5cm and 0cm of u.south west, rectangle, draw, matrix of math nodes, ampersand replacement=\&, inner sep=2pt, text height=.8em, text depth=.2em] (v) {
              |[minimum width=3cm, align=center]| \dots \& |[minimum width=2cm, align=center]| $y_\mu'$ \& c_{\mu + 1} \& |[minimum width=2cm, align=center]| \dots \\
            };
            \node[left=0pt of v] {$v' = {}$};
            \node[right=0pt of v] {${} = v''$};
            
            \foreach \i in {1,...,3} {
              \draw ($(u.north -| u-1-\i.north east)$) -- ($(u.south -| u-1-\i.south east)$);
              \draw ($(v.north -| v-1-\i.north east)$) -- ($(v.south -| v-1-\i.south east)$);
            };
            
            \path[pattern=north west lines] ($(u.north -| u-1-1.north west)+(0pt,1mm+2pt)$) rectangle ($(u.north -| u-1-1.north east)+(0pt,2pt)$);
            \path[pattern=north east lines] ($(v.south -| v-1-1.south west)+(0pt,-1mm-2pt)$) rectangle ($(v.south -| v-1-1.south east)+(0pt,-2pt)$);
            
            \draw[->] ($(u.north west) + (-0.75cm, 0.85cm+2pt)$) .. controls +(0.5cm, 0cm) and +(0cm,1cm) .. node[below, pos=0] {$\rho_0$} ($(u.north -| u-1-1.north east)+(-1cm,1mm+4pt)$);
            \draw[->] ($(v.south west) + (-0.75cm, -0.85cm-2pt)$) .. controls +(0.5cm, 0cm) and +(0cm,-1cm) .. node[above, pos=0] {$\rho_0$} ($(v.south -| v-1-1.south east)+(-1cm,-1mm-4pt)$);
            
            \path[pattern=north west lines] ($(u.north -| u-1-4.north west)+(0pt,1mm+2pt)$) rectangle ($(u.north -| u-1-4.north east)+(0pt,2pt)$);
            \path[pattern=north east lines] ($(v.south -| v-1-4.south west)+(0pt,-1mm-2pt)$) rectangle ($(v.south -| v-1-4.south east)+(0pt,-2pt)$);
            
            \draw[->, shorten <= 1ex, dashed] ($(u.north -| u-1-1.north east)+(-1cm,1mm+4pt)$) .. controls +(0.25cm, 0.75cm) and +(0cm,1cm) .. node[pos=1, above right] {$\neq X_a$} ($(u.north -| u-1-4.north west)+(1cm,1mm+4pt)$);
            \draw[->, shorten <= 1ex, dashed] ($(v.south -| v-1-1.south east)+(-1cm,-1mm-4pt)$) .. controls +(0.25cm, -0.75cm) and +(0cm,-1cm) .. node[pos=1, below right] {$\neq X_a$} ($(v.south -| v-1-4.south west)+(1cm,-1mm-4pt)$);
            
            \draw[->] ($(u.south east) + (0pt, -0.75cm)$) .. controls +(-0.5cm, 0cm) and +(0cm,-0.5cm) .. node[right, pos=0] {$\lambda_0$} ($(u.south -| u-1-3.south) + (0pt, -2pt)$);
            \draw[->] ($(v.north east) + (0pt, 0.75cm)$) .. controls +(-0.5cm, 0cm) and +(0cm,0.5cm) .. ($(v.north -| v-1-3.north) + (0pt, 2pt)$);
            
            \path[pattern=north west lines] ($(u.south -| u-1-1.south west)+(0pt,-2pt)$) rectangle ($(u.south -| u-1-1.south east)+(0pt,-1mm-2pt)$);
            \path[pattern=north east lines] ($(v.north -| v-1-1.north west)+(0pt,1mm+2pt)$) rectangle ($(v.north -| v-1-1.north east)+(0pt,2pt)$);
            \draw[->, shorten <= 1ex, dashed] ($(u.south -| u-1-3.south) + (0pt, -2pt)$) .. controls +(-0.25cm, -0.5cm) and +(0cm,-0.75cm) .. ($(u.south -| u-1-1.south east)+(-0.5cm,-1mm-4pt)$);
            \draw[->, shorten <= 1ex, dashed] ($(v.north -| v-1-3.north) + (0pt, 2pt)$) .. controls +(-0.25cm, 0.5cm) and +(0cm,0.75cm) .. ($(v.north -| v-1-1.north east)+(-0.5cm,1mm+4pt)$);
            \node[left] at ($(u-1-1.south east)!0.5!(v-1-1.north east)+(-0.5cm,0pt)$) {$\neq Y_a$};
            
            \draw[decorate, decoration={brace}] ($(u.north -| u-1-1.north east) + (-1cm+2pt, 1mm+4pt)$) -- node[above, yshift=2pt] {$\scriptstyle u_1'$} ($(u.north -| u-1-2.north east) + (0pt, 1mm+4pt)$);
            
            \draw[decorate, decoration={brace}] ($(v.south -| v-1-2.south east) + (0pt, -1mm-4pt)$) -- node[below, yshift=-2pt] {$\scriptstyle v_1'$} ($(v.south -| v-1-1.south east) + (-1cm+2pt, -1mm-4pt)$);
          \end{tikzpicture}
          \caption{$|x_\mu'|_a < |y_\mu'|_a$}\label{sfig:nextPosIsLambdaPos:moreAsInV}
        \end{subfigure}\\
        \begin{subfigure}{\linewidth}\centering
          \begin{tikzpicture}[baseline=(uLabel.base)]
            \node[color=gray] (uLabel) {$u' = {}$};
            \matrix [right=0cm of uLabel, rectangle, draw, matrix of math nodes, ampersand replacement=\&, inner sep=2pt, text height=.8em, text depth=.2em] (u) {
                |[minimum width=1cm, align=center]| \dots \& |[minimum width=1.5cm, align=center]| $x_\mu''$\textcolor{gray}{$a^\ell$} \& c_{\mu + 1} \& |[minimum width=1.5cm, align=center]| $x_{\mu + 1}'$ \& c_{\mu + 2} \& |[minimum width=1cm, align=center]| \dots \& c_\nu \& |[minimum width=1.5cm, align=center]| $x_{\nu}'$ \& c_{\nu + 1} \& |[minimum width=2cm, align=center]| \dots \\
              };
            \node[right=0pt of u] {${} = u''$};
          
            \matrix [below right=1.5cm and 0cm of u.south west, rectangle, draw, matrix of math nodes, ampersand replacement=\&, inner sep=2pt, text height=.8em, text depth=.2em] (v) {
                |[minimum width=1cm, align=center]| \dots \& |[minimum width=1.5cm, align=center]| $y_\mu'$ \& c_{\mu + 1} \& |[minimum width=1.5cm, align=center]| $y_{\mu + 1}'$ \& c_{\mu + 2} \& |[minimum width=1cm, align=center]| \dots \& c_\nu \& |[minimum width=1.5cm, align=center]| $y_{\nu}'$ \& c_{\nu + 1} \& |[minimum width=2cm, align=center]| \dots \\
              };
            \node[left=0pt of v] {$v' = {}$};
  
            \foreach \i in {1,...,9} {
                \draw ($(u.north -| u-1-\i.north east)$) -- ($(u.south -| u-1-\i.south east)$);
                \draw ($(v.north -| v-1-\i.north east)$) -- ($(v.south -| v-1-\i.south east)$);
              };
          
            \draw[->] ($(u.south west) + (-0.5cm, -0.75cm)$) .. controls +(0.5cm, 0cm) and +(-0.5cm,-0.75cm) .. node[left, pos=0] {$\rho_0$} ($(u.south -| u-1-9.south)+(0pt,-2pt)$);
            \draw[->] ($(v.north west) + (-0.5cm, 0.75cm)$) .. controls +(0.5cm, 0cm) and +(-0.5cm,0.75cm) .. ($(v.north -| v-1-9.north)+(0pt,2pt)$);
            
            \path[pattern=north east lines] ($(u.north -| u-1-9.north west)+(0pt,1mm+2pt)$) rectangle ($(u.north -| u-1-10.north east)+(0pt,2pt)$);
            \draw[->] ($(u.north east) + (0.75cm, 0.85cm+2pt)$) .. controls +(-0.5cm, 0cm) and +(0cm,0.75cm) .. node[below, pos=0] {$\lambda_1$} ($(u.north -| u-1-9.north) + (1.5cm, 1mm+4pt)$);
            \draw[->, shorten <= 1ex] ($(u.north -| u-1-9.north)+(1.5cm,1mm+4pt)$) .. controls +(-0.25cm, 1cm) and +(0cm,1cm) .. node[below] {$Y_{c_{\nu}}$} ($(u.north -| u-1-7.north)+(0pt,2pt)$);
            \node[anchor=base, inner sep=1pt] at ($(u.north -| u-1-6.north)+(0pt,2pt+0.5cm)$) (dotsU) {$\dots$};
            \draw[-, shorten <= 1ex] ($(u.north -| u-1-7.north)+(0pt,2pt)$) .. controls +(-0.25cm, 0.5cm) and +(0.25cm,0cm) .. (dotsU.east);
            \draw[->] (dotsU.west) .. controls +(-0.5cm, 0pt) and +(-0cm,0.5cm) .. node[above right, pos=1, xshift=2ex, inner sep=0pt] {$Y_{c_{\nu + 2}}$} ($(u.north -| u-1-5.north)+(0pt,2pt)$);
            \draw[->, shorten <= 1ex] ($(u.north -| u-1-5.north)+(0pt,2pt)$) .. controls +(-0.25cm, 1cm) and +(0cm,1cm) .. node[below] {$Y_{c_{\mu + 1}}$} node[above left, pos=1] {$\lambda_0$} ($(u.north -| u-1-3.north)+(0pt,2pt)$);

            \path[pattern=north west lines] ($(v.south -| v-1-9.south west)+(0pt,-1mm-2pt)$) rectangle ($(v.south -| v-1-10.south east)+(0pt,-2pt)$);
            \draw[->] ($(v.south east) + (0.75cm, -0.85cm-2pt)$) .. controls +(-0.5cm, 0cm) and +(0cm,-0.75cm) .. node[above, pos=0] {$\lambda_1$} ($(v.south -| v-1-9.south) + (1.5cm, -1mm-4pt)$);
            \draw[->, shorten <= 1ex] ($(v.south -| v-1-9.south)+(1.5cm,-1mm-4pt)$) .. controls +(-0.25cm, -1cm) and +(0cm,-1cm) .. node[above] {$Y_{c_{\nu}}$} ($(v.south -| v-1-7.south)+(0pt,-2pt)$);
            \node[anchor=base, inner sep=1pt] at ($(v.south -| v-1-6.south)+(0pt,-2pt-0.5cm)$) (dotsV) {$\dots$};
            \draw[-, shorten <= 1ex] ($(v.south -| u-1-7.south)+(0pt,-2pt)$) .. controls +(-0.25cm, -0.5cm) and +(0.25cm,0cm) .. (dotsV.east);
            \draw[->] (dotsV.west) .. controls +(-0.5cm, 0pt) and +(-0cm,-0.5cm) .. node[below right, pos=1, xshift=2ex, inner sep=0pt] {$Y_{c_{\nu + 2}}$} ($(v.south -| v-1-5.south)+(0pt,-2pt)$);
            \draw[->, shorten <= 1ex] ($(v.south -| v-1-5.south)+(0pt,-2pt)$) .. controls +(-0.25cm, -1cm) and +(0cm,-1cm) .. node[above] {$Y_{c_{\mu + 1}}$} node[below left, pos=1] {$\lambda_0$} ($(v.south -| v-1-3.south)+(0pt,-2pt)$);
%
%
          \end{tikzpicture}
        \caption{$|x_\mu'|_a > |y_\mu'|_a$}\label{sfig:nextPosIsLambdaPos:moreAsInU}
        \end{subfigure}\\
        \begin{subfigure}{\linewidth}\centering
          \begin{tikzpicture}[baseline=(uLabel.base)]
            \node[color=gray] (uLabel) {$u'/u'':$~};
            \matrix [right=0cm of uLabel, rectangle, draw, matrix of math nodes, ampersand replacement=\&, inner sep=2pt, text height=.8em, text depth=.2em] (u) {
                |[minimum width=1cm, align=center]| \dots \& |[minimum width=1.5cm, align=center]| $x_\mu'' \textcolor{gray}{a^\ell}$ \& c_{\mu + 1} \& |[minimum width=1cm, align=center]| \dots \& c_{K} \& |[minimum width=1.5cm, align=center]| $x_K' \textcolor{gray}{a^\ell}$ \& c_{K + 1} \& |[minimum width=1cm, align=center]| \dots \& x_\nu' \& c_{\nu + 1} \& |[minimum width=1cm, align=center]| \dots \\
              };

            \matrix [below right=1.5cm and 0cm of u.south west, rectangle, draw, matrix of math nodes, ampersand replacement=\&, inner sep=2pt, text height=.8em, text depth=.2em] (v) {
                |[minimum width=1cm, align=center]| \dots \& |[minimum width=1.5cm, align=center]| $y_\mu'$ \& c_{\mu + 1} \& |[minimum width=1cm, align=center]| \dots \& c_{K} \& |[minimum width=1.5cm, align=center]| $y_K'$ \& c_{K + 1} \& |[minimum width=1cm, align=center]| \dots \& x_\nu' \& c_{\nu + 1} \& |[minimum width=1cm, align=center]| \dots \\
              };
            \node[left=0pt of v] {$v' = {}$};
            \node[right=0pt of v] {${} = v''$};
  
            \foreach \i in {1,...,10} {
                \draw ($(u.north -| u-1-\i.north east)$) -- ($(u.south -| u-1-\i.south east)$);
                \draw ($(v.north -| v-1-\i.north east)$) -- ($(v.south -| v-1-\i.south east)$);
              };
          
            \path[pattern=north east lines] ($(u.south west)+(0pt,-2pt)$) rectangle ($(u.south -| v-1-1.south east)+(0pt,-1mm-2pt)$);
            \draw[->] ($(u.south west) + (-0.75cm, -0.75cm)$) .. controls +(0.75cm, 0cm) and +(0cm,-0.5cm) .. ($(u.south -| u-1-1.south)+(0pt,-1mm-4pt)$);

            \path[pattern=north west lines] ($(v.north west)+(0pt,2pt)$) rectangle ($(v.north -| v-1-1.north east)+(0pt,1mm+2pt)$);
            \draw[->] ($(v.north west) + (-0.75cm, 0.75cm)$) .. controls +(0.75cm, 0cm) and +(0cm,0.5cm) .. ($(v.north -| v-1-1.north)+(0pt,1mm+4pt)$);

            \draw[->, shorten <= 1ex] ($(u.south -| u-1-1.south)+(0pt,-1mm-4pt)$) .. controls +(0.25cm, -0.25cm) and +(0cm,-1cm) .. node[above] {$\neq X_a$} ($(u.south -| u-1-10.south)+(0cm,-2pt)$);
            \draw[->, shorten <= 1ex] ($(v.north -| v-1-1.north)+(0pt,1mm+4pt)$) .. controls +(0.25cm, 0.25cm) and +(0cm,1cm) .. node[below] {$\neq X_a$} ($(v.north -| v-1-10.north)+(0cm,2pt)$);
            \node[right] at ($(u-1-10.south)!0.5!(v-1-10.north)+(-0.5cm,0pt)$) {$\rho_0$};
           
            \path[pattern=north east lines] ($(u.north -| u-1-10.north west)+(0pt,2pt)$) rectangle ($(u.north -| u-1-11.north east)+(0pt,1mm+2pt)$);
            \path[pattern=north west lines] ($(v.south -| v-1-10.south west)+(0pt,-2pt)$) rectangle ($(v.south -| v-1-11.south east)+(0pt,-1mm-2pt)$);
            \draw[->] ($(u.north east) + (0.75cm, 0.75cm+2pt)$) .. controls +(-0.5cm, 0cm) and +(0cm,0.5cm) .. node[below, pos=0] {$\lambda_1$} ($(u.north -| u-1-10.north east) + (0.25cm, 1mm+4pt)$);
            \draw[->] ($(v.south east) + (0.75cm, -0.75cm-2pt)$) .. controls +(-0.5cm, 0cm) and +(0cm,-0.5cm) .. node[above, pos=0] {$\lambda_1$} ($(v.south -| v-1-10.south east) + (0.25cm, -1mm-4pt)$);
            
            \node[anchor=base, inner sep=1pt] at ($(u.north -| u-1-8.north)+(0pt,2pt+0.75cm)$) (dotsUr) {$\dots$};
            \draw[-, shorten <= 1ex] ($(u.north -| u-1-10.north east) + (0.25cm, 1mm+4pt)$) .. controls +(-0.5cm, 0.5cm) and +(0.25cm,0cm) .. (dotsUr.east);
            \draw[->] (dotsUr.west) .. controls +(-0.5cm, 0pt) and +(-0cm,0.5cm) .. node[below right, xshift=1ex, pos=0.5, inner sep=0pt] {$Y_{c_{K + 1}}$} ($(u.north -| u-1-7.north)+(0pt,2pt)$);
            \draw[->, shorten <= 1ex] ($(u.north -| u-1-7.north)+(0pt,2pt)$) .. controls +(-0.25cm, 1cm) and +(0cm,1cm) .. node[below] {$Y_{c_{K}}$} ($(u.north -| u-1-5.north)+(0pt,2pt)$);
            \node[anchor=base, inner sep=1pt] at ($(u.north -| u-1-4.north)+(0pt,2pt+0.75cm)$) (dotsUl) {$\dots$};
            \draw[-, shorten <= 1ex] ($(u.north -| u-1-5.north)+(0pt,2pt)$) .. controls +(-0.25cm, 0.75cm) and +(0.25cm,0cm) .. (dotsUl.east);
            \draw[->] (dotsUl.west) .. controls +(-0.5cm, 0pt) and +(-0cm,0.5cm) .. node[above right, pos=1, xshift=2ex, inner sep=0pt] {$Y_{c_{\mu + 1}}$} node[above left, pos=1] {$\lambda_0$} ($(u.north -| u-1-3.north)+(0pt,2pt)$);
            
            \node[anchor=base, inner sep=1pt] at ($(v.south -| v-1-8.south)+(0pt,-2pt-0.75cm)$) (dotsVr) {$\dots$};
            \draw[-, shorten <= 1ex] ($(v.south -| v-1-10.south east) + (0.25cm, -1mm-4pt)$) .. controls +(-0.5cm, -0.5cm) and +(0.25cm,0cm) .. (dotsVr.east);
            \draw[->] (dotsVr.west) .. controls +(-0.5cm, 0pt) and +(-0cm,-0.5cm) .. node[above right, xshift=1ex, pos=0.5, inner sep=0pt] {$Y_{c_{K + 1}}$} ($(v.south -| v-1-7.south)+(0pt,-2pt)$);
            \draw[->, shorten <= 1ex] ($(v.south -| v-1-7.south)+(0pt,-2pt)$) .. controls +(-0.25cm, -1cm) and +(0cm,-1cm) .. node[above] {$Y_{c_{K}}$} ($(v.south -| v-1-5.south)+(0pt,-2pt)$);
            \node[anchor=base, inner sep=1pt] at ($(v.south -| v-1-4.south)+(0pt,-2pt-0.75cm)$) (dotsVl) {$\dots$};
            \draw[-, shorten <= 1ex] ($(v.south -| v-1-5.south)+(0pt,-2pt)$) .. controls +(-0.25cm, -0.75cm) and +(0.25cm,0cm) .. (dotsVl.east);
            \draw[->] (dotsVl.west) .. controls +(-0.5cm, 0pt) and +(-0cm,-0.5cm) .. node[below right, pos=1, xshift=2ex, yshift=-0.5ex, inner sep=0pt] {$Y_{c_{\mu + 1}}$} node[below left, pos=1] {$\lambda_0$} ($(v.south -| v-1-3.south)+(0pt,-2pt)$);
          \end{tikzpicture}
          \caption{$|x_\mu'|_a > |y_\mu'|_a, c_K \neq a$}\label{sfig:nextPosIsLambdaPos:cKIsA}
        \end{subfigure}
        \caption{Only a prefix of $\lambda\Yinfty$ ends directly to the right if $x_\mu'$/$y_\mu'$ in $u'$ and $v'$.}
      \end{figure}
      
      In fact, since we have $u' \sim_{m, \rho, \lambda} v'$, there must already have been an $a$ in the middle part $u_1'$ of the $\rho_0$-$\lambda_0$-split of $u'$. Thus, adding the $a^{\beta m}$ to $u''$ cannot change the alphabet of the middle part of any $\rho'$-$\lambda'$-split of $u''$ (where $\rho'$ is a prefix of $\rho \Xinfty$ and $\lambda'$ is a prefix of $\lambda \Yinfty$ with $\rho'(u'') \leq \lambda'(u'')$). Together with the fact that we have only added a number of $a$s which is a multiple of $m$, this shows $u'' \sim_{m, \rho, \lambda} u'$ and, thus, \autoref{itm:rhoLambdaSim} of the definition of $u$-like words.
      
      To show \autoref{itm:rhoRDescent}, it suffices to observe that $u_1'$ is a prefix of some block in the $\mathcal{R}$-factorization of $u'$ containing an $a$. Thus, by \autoref{fct:insertAM}, we have not changed the image under $\varphi$.

      From now on, we may assume that we have $|x_\mu'|_a > |y_\mu'|_a$, i.\,e.\ there is some $\ell > 0$ such that $x_\mu' = x_\mu'' a^\ell$ (up to permutation) with $|x_\mu''|_a = |y_\mu'|_a$. Let $\rho_0$ be the minimal prefix of $\rho \Xinfty$ such that $\rho_0$ ends to the right of $x_\mu' = x_\mu'' a^\ell$. Because $\rho_0$ cannot end at the position of $c_{\mu + 1}$, there is some $\nu > \mu + 1$ such that $\rho_0$ ends directly to the right of $x_\nu'$ in $u'$ and, therefore, $y_\nu'$ in $v'$ (this position may well be $+\infty$). Let $\lambda_1$ be the maximal prefix of $\lambda \Yinfty$ with $\rho_0(u') \leq \lambda_1(u')$ (both possibly $+\infty$). This implies $\lambda_0 = \lambda_1 Y_{c_\nu} \dots Y_{c_{\mu + 1}}$. The situation is depicted in \autoref{sfig:nextPosIsLambdaPos:moreAsInU} (note, however, that $\rho_0$ and $\lambda_1$ may end at $+\infty$).
      
      First, consider the case $c_k = a$ for all $\mu < k \leq \nu$. In particular, we have that all instruction after $\lambda_1$ in $\lambda \Yinfty$ up to $\lambda_0$ are $Y_a$ and that, thus, no $x_k'$ or $y_k'$ with $\mu < k \leq \nu$ may contain an $a$. In addition, the middle parts of the $\rho_0$-$\lambda_1$-splits of $u'$ and $v'$ do not contain an $a$ (this includes the case that they are both empty) and, for the left parts, we obtain:
      \begin{align*}
        &|x_0' \dots x_{\mu - 1}'|_a + |c_1 \dots c_\mu|_a + |x_\mu'|_a + |c_{\mu + 1} \dots c_\nu|_a + |x_{\mu + 1}' \dots x_\nu'|_a\\
        ={}&|x_0' c_1 x_1' \dots c_\mu x_\mu' \, c_{\mu + 1} x_{\mu + 1}' \dots c_\nu x_\nu'|_a\\
        \equiv{}&
        |y_0' c_1 y_1' \dots c_\mu y_\mu' \, c_{\mu + 1} y_{\mu + 1}' \dots c_\nu y_\nu'|_a \mod m\\
        ={}& |y_0' \dots y_{\mu - 1}'|_a + |c_1 \dots c_\mu|_a + |y_\mu'|_a + |c_{\mu + 1} \dots c_\nu|_a + |y_{\mu + 1}' \dots y_\nu'|_a
      \end{align*}
      Since we have $|x_0' \dots x_{\mu - 1}'|_a = |y_0' \dots y_{\mu - 1}'|_a$ (by the choice of $\mu$) and $|x_{\mu + 1}' \dots x_\nu'|_a = |y_{\mu + 1}' \dots y_\nu'|_a = 0$, we obtain $|x_\mu'|_a \equiv |y_\mu'|_a \bmod m$ and, thus, that $\ell$ is a multiple of $m$. Thus, we may apply \autoref{fct:droppingAMRemainsULike} to drop the $a^\ell$ block and to obtain $u''$ as a $u$-like word.
                        
      It remains the case that there is some $K$ with $\mu < K \leq \nu$ and $c_K \neq a$. The idea is that we may move the excessive $a^\ell$ block to $x_K'$ without any disruptions in that case, i.\,e.\ we let $x_K'' = x_K' a^\ell$ and $x_k'' = x_k'$ for all other $k \not\in \{ \mu, K \}$ (as well as $v'' = v'$).
      
      As always, we have to show that $u''$ remains $u$-like. Since we have $c_K \neq a$ (and, thus, $Y_{c_K} \neq Y_a$), all prefixes of $\lambda\Yinfty$ visit the same $c_k$ in the factorization (\ref{eqn:rhoLambdaFactorization1}) of $u'$ as they do in the factorization (\ref{eqn:rhoLambdaFactorization2}) of $u''$. With regard to $\rho \Xinfty$, we observe that a prefix $\rho'$ of $\rho \Xinfty$ may only end either to the left of $x_\mu' = x_\mu'' a^\ell$ in $u'$ or to the right of $x_\nu'$ (by the choice of $\rho_0$ and, thus, $\nu$). The latter is the case if and only if $\rho_0$ is prefix of $\rho'$. Since there is an $a$ in $x_\mu' = x_\mu'' a^\ell$, the last instruction of $\rho_0$ cannot be an $X_a$. Thus, the additional $a^\ell$ block in $x_K''$ does not interfere and $\rho_0$ also ends directly to the right of $x_\nu''$ in $u''$. This shows that $u''$ remains $\rho$-$\lambda$-compatible to $u'$ (and, thus, $u$). The situation is depicted in \autoref{sfig:nextPosIsLambdaPos:cKIsA}. Note, however, that $\rho_0$ and $\lambda_1$ may also end at $+\infty$, however.
      
      For \autoref{itm:rhoLambdaSim} of the definition of $u$-like words, we show $u' \sim_{m, \rho, \lambda} u''$ (by using $u' \sim_{m, \rho, \lambda} v'$). For this, consider an arbitrary prefix $\rho'$ of $\rho \Xinfty$ and an arbitrary prefix $\lambda'$ of $\lambda \Yinfty$ with $\rho'(u') \leq \lambda'(u')$ (and, thus, also $\rho'(u'') \leq \lambda'(u'')$). There is nothing to show if $\rho_0(u') \leq \rho'(u')$ since the numbers of $a$s in the left part of the $\rho'$-$\lambda'$-split of $u''$ is the same as in the left part of the $\rho'$-$\lambda'$-split of $u'$ (we have only moved the $a^\ell$ block). The only alternative is, however, that $\rho'$ ends to the left of $x_\mu' = x_\mu'' a^\ell$ in $u'$. In this case, consider the $\rho'$-$\lambda_0$-splits of $u'$, $u''$ and $v'$. Note that there is an $a$ in the middle part for the one belonging to $u'$ (as $x_\mu' = x_\mu'' a^\ell$ is one of its suffixes). Since we have $u' \sim_{m, \rho, \lambda} v'$, there must also be one in the middle part of the split of $v'$. Since we have $|x_k'|_a = |y_k|_a$ for all $k < \mu$ and $|x_\mu''|_a = |y_\mu'|_a$, it follows that there still must be an $a$ in the middle part of the split of $u''$ (even after moving the $a^\ell$ block to the right) and there is nothing more to show. This proves \autoref{itm:rhoLambdaSim}.
      
      Finally, we have to show \autoref{itm:rhoRDescent}. However, since there is no position where any prefix of $\rho \Xinfty$ ends between them, $x_\mu' = x_\mu'' a^\ell$ and $x_K'$ belong to the same block of the $\mathcal{R}$-factorization of $u'$. Thus, moving the $a^\ell$ block is just a permutation within the same block of the $\mathcal{R}$-factorization, which does not change the image under $\varphi$ by \autoref{fct:permuteR}.
    \end{proof}

    We may combine \autoref{prop:congruenceInDAb} and \autoref{prop:DAbImpliesQuotient} into the following theorem.
    \begin{theorem}\label{thm:mainResult}
      Let $M$ be a monoid generated by $\Sigma$. Then $M$ is in $\DAb$ if and only if there are $m, n \in \Np$ and a surjective homomorphism $\Sigma^* / {\approx_{m, n}} \to M$.
    \end{theorem}
    \begin{proof}
      Let $\varphi: \Sigma^* \to M$ be a surjective homomorphism.
      For $M \in \DAb$, we obtain that there are $m, n \in \Np$ such that $u \approx_{m, n} v$ implies $\varphi(u) = \varphi(v)$ by \autoref{prop:DAbImpliesQuotient}. In other words, $[w]_{\approx_{m, n}} \mapsto \varphi(w)$ is a well-defined homomorphism, which is surjective because $\varphi$ is.
      
      Now, let $M$ be a homomorphic image of $\Sigma^* / {\approx_{m, n}}$ for some $m, n \in \Np$. Since $\Sigma^* / {\approx_{m, n}}$ is in $\DAb$ by \autoref{prop:congruenceInDAb} and since $\DAb$ is closed under taking homomorphic images as a variety, we obtain $M \in \DAb$.
    \end{proof}
  
    \paragraph*{Semigroups.}
    We also obtain the same result for semigroups. For this, let $S^1 = S \cup \{ 1 \}$ be the monoid with $1 s = s 1 = s$ for all $s \in S^1$. We have:
    \begin{fact}\label{fct:isomorphismSemigroupMonoid}
      Let $m \in \Np$, $n \in \mathbb{N}$. Then $\Sigma^* / {\approx_{m, n}}$ is isomorphic to $\left( \Sigma^+ / {\approx_{m, n}} \right)^1$.
    \end{fact}
    \begin{proof}
      The isomorphism is given by mapping the class of $\varepsilon$ to $1$ and every other class to itself. This is a well-defined isomorphism since we have $\varepsilon \not\approx_{m, n} w$ for all $\Sigma^* \ni w \neq \varepsilon$.
      Indeed, for the empty $X$-ranker $\rho$ and the empty $Y$-ranker $\lambda$, the $\rho$-$\lambda$-split of $\varepsilon$ is $(\varepsilon, \varepsilon, \varepsilon)$ and that of $w$ is $(\varepsilon, w, \varepsilon)$. Since we have $\alphabet \varepsilon = \emptyset \neq \alphabet w$, we have $\varepsilon \not\sim_{m, \varepsilon, \varepsilon} w$.
    \end{proof}
    
    Let $\DAb_S$ be the variety of finite semigroups whose regular $\mathcal{D}$-classes form Abelian groups.
    \begin{corollary}
      Let $S$ be a semigroup generated by $\Sigma$. Then $S$ is in $\DAb_S$ if and only if there are $n, m \in \Np$ and a surjective homomorphism $\Sigma^+ / {\approx_{m, n}} \to S$.
    \end{corollary}
    \begin{proof}
      Since $S$ is generated by $\Sigma$, there is a surjective homomorphism $\varphi: \Sigma^+ \to S$. Suppose $S$ is in $\DAb_S$. Then, $S^1$ is in $\DAb$ and we may extend $\varphi$ into a surjective homomorphism of monoids $\Sigma^* \to S^1$ by mapping $\varepsilon$ to $1$. By \autoref{prop:DAbImpliesQuotient}, there are $m, n \in \Np$ such that $u \approx_{m, n} v$ implies $\varphi(u) = \varphi(v)$. This shows that mapping $[w]_{\approx_{m, n}} \mapsto \varphi(w)$ defines a well-defined monoid homomorphism $\varphi': \Sigma^* / {\approx_{m, n}} \to S^1$, which is surjective since $\varphi$ is. Also note that we have $\varphi'([w]_{\approx_{m, n}}) = \varphi(w) \in S = S^1 \setminus \{ 1 \}$ for all $w \in \Sigma^+$. Thus, we may restrict $\varphi'$ into a surjective homomorphism $\Sigma^+ / {\approx_{m, n}} \to S$ of semigroups.
      
      For the other direction, suppose that $S$ is a homomorphic image of $\Sigma^+ / {\approx_{m, n}}$ for some $m, n \in \Np$. Thus, it is a homomorphic image of a subsemigroup of $(\Sigma^+ / {\approx_{m, n}})^1$, which is isomorphic to $\Sigma^* / {\approx_{m, n}}$ (by \autoref{fct:isomorphismSemigroupMonoid}) and, thus, in $\DAb \subseteq \DAb_S$ by \autoref{prop:congruenceInDAb}.
    \end{proof}
  \end{section}

  \begin{section}{A Normal Form for Pseudowords over $\DAb$}
    \begin{subsection}{Definition of the Normal Form}
      In this section, we will show that there is a normal form for pseudowords with respect to $\DAb$. For this, we first need to define when a pseudoword is in normal form, then we need to show that, for every pseudoword $\alpha$, there is one $\alpha'$ in normal form with $\DAb \models \alpha = \alpha'$. Finally, we need to show that, for two distinct pseudowords $\alpha$ and $\beta$ in normal form, there is a finite monoid $M \in \DAb$ and a continuous homomorphism $\varphi: \freeProf \to M$ with $\varphi(\alpha) \neq \varphi(\beta)$ (i.\,e.\ that $\DAb \not\models \alpha = \beta$).
      
      \paragraph*{Profinite Natural Numbers and Integers.}
      Similar to the profinite metric from above, let $d'(n, n) = 0$ for $n \in \mathbb{N}$ and $d'(m, n) = 2^{-|M|}$ for $n, m \in \mathbb{N}$ with $n \neq m$ where $M$ is the smallest monogenic monoid generated by some $a$ with $a^n \neq a^m$. Just like the profinite metric, $d'$ is an ultrametric and we may consider the completion $\hat{\mathbb{N}}$ of $\mathbb{N}$ with respect to this metric. This completion is the set of \emph{profinite natural numbers}. By continuing the addition and multiplication of natural numbers (i.\,e.\ we let $\lim (n_k)_k + \lim (m_k)_k = \lim (n_k + m_k)_k$ and $\lim (n_k)_k \cdot \lim (m_k)_k = \lim (n_k \cdot m_k)_k$), we obtain a natural, commutative semiring structure. While we need a few results on $\hat{\mathbb{N}}$, we will not give any proofs but refer the reader for example to \cite[Section~2.3]{alm20:short} for details.
      
      In $\hat{\mathbb{N}}$, there are two idempotents: the first one is $0$ and the second idempotent is $\omega = \lim (k!)_k$. It turns out that $\omega$ plays an important role in $\hat{\mathbb{N}}$. In fact, we have $\hat{\mathbb{N}} = \mathbb{N} \cup K$ for $K = \{ \omega + \hat{n} \mid \hat{n} \in \hat{\mathbb{N}} \}$, where $\mathbb{N}$ and $K$ are disjoint \cite[Proposition~2.3.3]{alm20:short}. We call the elements of $K$ \emph{infinite} profinite natural numbers.
      
      To understand the structure of $K = \hat{\mathbb{N}} \setminus \mathbb{N}$ better, we also need to introduce the ring of \emph{profinite integers} $\hat{\mathbb{Z}}$. It is defined analogously to the semiring of profinite natural numbers but the metric is defined using cyclic groups instead of monogenic monoids (see \cite[Section~2.2]{alm20:short} for details). It turns out that $\hat{\mathbb{N}} \setminus \mathbb{N}$ is isomorphic to $\hat{\mathbb{Z}}$:
      \begin{proposition}\label{prop:isomorphismProfiniteNumbers}
        The map
        \begin{align*}
          \hat{\mathbb{N}} \setminus \mathbb{N} &\to \hat{\mathbb{Z}}\\
          \lim (n_k)_k &\mapsto \lim (n_k)_k
        \end{align*}
        is a well-defined isomorphism of rings.
      \end{proposition}
      
      This connection allows us to define profinite natural numbers of the form $\omega + \hat{z}$ for all $\hat{z} \in \hat{\mathbb{Z}}$: it is the pre-image of $z$ under the isomorphism from \autoref{prop:isomorphismProfiniteNumbers}. Note that this is compatible with the above definition of adding two profinite natural numbers for $\omega + \hat{n}$ where $\hat{n} \in \hat{\mathbb{N}}$. This definition shows $(\omega + \hat{x}) \cdot (\omega + \hat{y}) = (\omega + \hat{x}\hat{y})$. Alternatively, one can also show $\omega + \hat{z} = \lim (|z_k| \cdot (k!) + z_k)_k$ for all $\hat{z} = \lim (z_k)_k \in \hat{\mathbb{Z}} \setminus \{ 0 \}$ (where $|z_k|$ denotes the absolute value of $z_k$). With this, we also obtain, for example, $(\omega - n) + n = (\omega - n) + (\omega + n) = \omega$ and $n \cdot (\omega + \hat{z}) = \omega + n\hat{z}$ for all $n \in \mathbb{N}$ and $\hat{z} \in \hat{\mathbb{Z}}$.
      
      If $(n_k)_k$ is a Cauchy sequence, one can show that so is $(n_k \bmod m)_k$ for all $m \in \Np$ where $n \bmod m$ denotes the smallest non-negative representative of the residue class of $n$ modulo $m$. Thus, we may define $\hat{n} \bmod m = \lim (n_k \bmod m)_k$ for all profinite natural numbers $\hat{n} = \lim (n_k)_k$. In fact, $\hat{n} \mapsto \hat{n} \bmod m$ is a continuous morphism and we obtain $\omega \bmod m = 0$ (as the idempotent $\omega$ may only be mapped to the only idempotent $0$ in the cyclic group of order $m$). This allows us also to extend the notation $\hat{n} \equiv \hat{m} \bmod m$ to profinite natural numbers $\hat{n}, \hat{m} \in \hat{\mathbb{N}}$.
      
      For a Cauchy sequence $(n_k)_k$ (with respect to $d'$) of elements from $\mathbb{N}$ and a Cauchy sequence $(u_k)_k$ (with respect to the profinite metric $d$) of elements from $\Sigma^*$, it is straight-forward to verify that $(u_k^{n_k})_k$ is again a Cauchy sequence (with respect to $d$). Thus, for every profinite word $\freeProf \ni \alpha = \lim (u_k)_k$ and every profinite natural number $\hat{\mathbb{N}} \ni \hat{n} = \lim (n_k)_k$, the profinite word $\alpha^{\hat{n}} = \lim (u_k^{n_k})_k$ is well-defined. Please note that we now have two definitions for $\alpha^\omega$ and for $\alpha^{\omega - 1}$, respectively, but that they coincide in both cases. For infinite profinite natural numbers, we have the following evaluation in a finite monoid and we can distinguish any pair of distinct infinite profinite natural numbers in a finite cyclic group.
      \begin{fact}\label{fct:profiniteMod}
        Let $\hat{n}$ be an infinite profinite natural number and let $M$ be a finite monoid. Then we have
        \[
          \varphi(a^{\hat{n}}) = \varphi(a^{M! + (\hat{n} \bmod M!)})
        \]
        for every homomorphism $\varphi: \freeProf \to M$ and every $a \in \Sigma$.
      \end{fact}
      \begin{proof}
        Let $\hat{n} = \lim (n_k)_k$. There is some $k$ such that $\varphi(a^{\hat{n}}) = \varphi(a^{n_k})$ and $\hat{n} \bmod M! = n_k \bmod M!$. For $n = n_k$, we may write $n = M! \cdot \ell + (n \bmod M!)$. Together with the fact that $\hat{n}$ is infinite and that we, thus, have $\hat{n} = \omega + \hat{n}$, we obtain
        \[
          \varphi(a^{\hat{n}}) = \varphi(a^\omega a^{\hat{n}}) = \varphi(a^{M!} a^{M! \cdot \ell} a^{n \bmod M!}) = \varphi(a^{M! + \hat{n} \bmod M!}) \textbf{.}\qedhere
        \]
      \end{proof}
    
      \begin{fact}\label{fct:distinguishInfiniteNumbers}
        Let $\hat{m}, \hat{n}$ be infinite profinite natural numbers with $\hat{m} \neq \hat{n}$. Then there is some $m$ with $\hat{m} \not\equiv \hat{n} \bmod m$.
      \end{fact}
      \begin{proof}
        The distinct infinite profinite natural numbers $\hat{m}$ and $\hat{n}$ get mapped to distinct profinite integers by the isomorphism from \autoref{prop:isomorphismProfiniteNumbers}. Thus, there is a finite cyclic group in which they can be distinguished. The order of this group is $m$.
      \end{proof}
      
      On the other hand, if $(u_k)_k$ is a Cauchy sequence (with respect to $d$), it is also easy to see that $(|u_k|_a)_k$ is a Cauchy sequence (with respect to $d'$) for all $a \in \Sigma$. This means that we may continue $|\cdot|_a$ into a homomorphism $|\cdot|_a: \freeProf \to \hat{\mathbb{N}}$ by letting $|\lim (u_k)_k|_a = \lim (|u_k|_a)_k$. Furthermore, we have the following compatibility, which in particular holds for $\DAb$.
      \begin{fact}\label{fct:numberOfAsCompatibility}
        Let $\V$ be a class of monoids containing all finite monogenic monoids. Then, for $\alpha, \beta \in \freeProf$, we have that
        $\V \models \alpha = \beta$ implies $|\alpha|_a = |\beta|_a$ for all $a \in \Sigma$.
      \end{fact}
      \begin{proof}
        We show the statement by contraposition. Suppose we have $\hat{m} = |\alpha|_a \neq |\beta|_a = \hat{n}$ for some $\alpha, \beta \in \freeProf$ and $a \in \Sigma$. Then there is a finite monogenic monoid generated by $s$ with $s^{\hat{m}} \neq s^{\hat{n}}$. This implies $M \not\models \alpha = \beta$ (by considering the homomorphism $\freeProf \to M$ induced by $a \mapsto s$ and $b \mapsto 1$ for all $b \neq a$) and, thus, $\V \not\models \alpha = \beta$.
      \end{proof}
      
      \paragraph*{Infinite Blocks and the Normal Form.}
      In order to define the normal form, we fix an arbitrary linear ordering $<$ for $\Sigma$.
      \begin{definition}\label{def:normalForm}
        An \emph{infinite block} is a pseudoword $\alpha \in \freeProf$ of the form
        \[
          \alpha = (a_1 \dots a_\ell)^\omega a_1^{\hat{n}_{1}} \dots a_\ell^{\hat{n}_\ell}
        \]
        with $\ell > 0$ where $a_1, \dots, a_\ell \in \Sigma$ with $a_1 < \dots < a_\ell$ and where $\hat{n}_{1}, \dots \hat{n}_{\ell} \in \widehat{\mathbb{N}} \setminus \mathbb{N}$ are infinite proinfinite natural numbers.
        
        A pseudoword
        \[
          \freeProf \ni \alpha = \alpha_1 \dots \alpha_L
        \]
        is in \emph{normal form} (for $\DAb$) if each \emph{block} $\alpha_i$ (with $1 \leq i \leq L$) is either
        \begin{itemize}
          \item a letter (i.\,e.\ $\alpha_i \in \Sigma$) or
          \item an \emph{infinite block} and we have
            \begin{enumerate}[label=(c-\arabic*), leftmargin=1.5cm]
              \item\label{itm:normal:incomparable}
                $\alphabet \alpha_{i - 1} \setminus \alphabet \alpha_i \neq \emptyset$ (if $1 < i$) and $\alphabet \alpha_{i + 1} \setminus \alphabet \alpha_i \neq \emptyset$ (if $i < L$) and
              \item\label{itm:normal:omega}
                if $a \in \alphabet \alpha_{i - 1}$ for some $a \in A$ (if $1 < i$), then $|\alpha_i|_a = \omega \in \widehat{\mathbb{N}}$.
            \end{enumerate}
        \end{itemize}
      \end{definition}
      \begin{remark*}
        Instead of allowing only single letters as blocks in the normal form we could also allow entire finite words. That we chose the former over the latter is only a technicality to simplify some proofs later on.
      \end{remark*}
      \begin{remark*}
        If $\alpha = \alpha_1 \dots \alpha_L$ is in normal form and we have $\alpha_{i - 1} = a \in A$ and that $\alpha_i$ is an infinite block, then \ref{itm:normal:incomparable} implies $a \not\in \alphabet \alpha_i$ and there is no restriction on the $\hat{n}_j$ of $\alpha_i$ (apart from $\hat{n}_j \in \widehat{\mathbb{N}} \setminus \mathbb{N}$). Symmetrically, if $\alpha_{i + 1} = a \in A$ is a letter and $\alpha_i$ is again an infinite block, we also have $a \not\in \alphabet \alpha_i$.
        
        On the other hand, \ref{itm:normal:incomparable} also states that two adjacent infinite blocks must have incomparable alphabets (neither one is a subset of the other).
      \end{remark*}
    
      The infinite profinite natural numbers from the definition of an infinite block indicate the number of occurrences of the letter they belong to. We will often make use of this fact implicitly in the following.
      \begin{fact}
        For an infinite block
        \[
          \alpha = (a_1 \dots a_\ell)^\omega a_1^{\hat{n}_{1}} \dots a_\ell^{\hat{n}_\ell} \textbf{,}
        \]
        we have $|\alpha|_{a_i} = \hat{n}_i$.
      \end{fact}
      \begin{proof}
        We have $|\alpha|_{a_i} = \omega + \hat{n}_i = \hat{n}_i$ as $\hat{n}_i$ is from $\hat{\mathbb{N}} \setminus \mathbb{N}$ and, thus, of the form $\omega + \hat{n}'$.
      \end{proof}
    \end{subsection}
    \begin{subsection}{Every Element has a Normal Form}
      We need to derive a pseudoword $\alpha'$ in normal form an arbitrary pseudoword $\alpha$ with $\DAb \models \alpha = \alpha'$. For this, we will use a more general result that every pseudoword can be factorized into finite words and factors which are regular with respect to the variety $\V[DS]$, which consists of those finite monoids whose regular $\mathcal{D}$-classes form semigroups.
      
      \paragraph*{Regular Pseudowords.}
      A pseudoword $\alpha \in \hat{\Sigma}^*$ is \emph{$\V$-regular} for a class $\V$ of finite monoids if $\varphi(\alpha)$ is regular for all $M \in \V$ and all continuous homomorphisms $\varphi: \hat{\Sigma}^* \to M$.
      
      From \cite[Theorem~8.1.11]{alm94:short}, we obtain that every pseudoword can be factorized into $\DAb$-regular factors and finite words (since $\DAb \subseteq \V[DS]$).
      \begin{corollary}\label{cor:regularFiniteFactorization}
        Every pseudoword $\alpha \in \freeProf$ admits a factorization $\alpha = \alpha_1 \dots \alpha_L$ where each $\alpha_i$ (with $1 \leq i \leq L$) is either a finite word or $\DAb$-regular.
      \end{corollary}
    
      The finite factors correspond to blocks of letters in our normal form and the $\DAb$-regular factors will turn out to belong to infinite blocks. To see this, we first need to prove some identities for $\DAb$ though.
      
      \begin{fact}\label{fct:DAbEquations}
        Let $\alpha, \beta \in \freeProf$ with $\alphabet \beta \subseteq \alphabet \alpha$. Then $\DAb$ satisfies the equations
        \begin{enumerate}[label=(\arabic*)]
          \item\label{itm:DAbEquations:moveOmega}
            $\displaystyle \alpha^\omega \beta = \alpha^\omega \beta \alpha^\omega = \beta \alpha^\omega$ and
          \item\label{itm:DAbEquations:subsetOmega}
            $\displaystyle \alpha^\omega \beta^\omega = \alpha^\omega = \beta^\omega \alpha^\omega$.
        \end{enumerate}
        In particular, if $\alphabet \alpha = \alphabet \beta$, $\DAb$ satisfies
        \begin{enumerate}[resume, label=(\arabic*)]
          \item\label{itm:DAbEquations:sameAlphabetOmega}
            $\displaystyle \alpha^\omega = \beta^\omega$.
        \end{enumerate}
      \end{fact}
      \begin{proof}
        Let $M \in \DAb$ and consider a homomorphism $\varphi: \freeProf \to M$. We may write $\alpha = \lim (u_k)_k$ and $\beta = \lim(v_k)_k$ for some sequences $(u_k)_k$ and $(v_k)_k$ of elements from $\Sigma^*$. There is some $k$ with $\alphabet u_k = \alphabet \alpha$ and $\varphi(\alpha) = \varphi(u_k)$ and we let $u = u_k$. In the same way, there is some $k$ with $\alphabet v_k = \alphabet \beta$ and $\varphi(\beta) = \varphi(v_k)$ and we let $v = v_k$.
        To simplify things further, we identify every $\gamma \in \freeProf$ with its image $\varphi(\gamma)$. 
        
        For equation \ref{itm:DAbEquations:moveOmega}, we need to show $u^\omega v = u^\omega v u^\omega = v u^\omega$. We only prove the part on the left as the one on the right can be obtained symmetrically. Let $v = a_1 \dots a_n$ for $a_1, \dots, a_n \in \Sigma$. By \autoref{lem:groupByAlphabet}, we have
        \begin{align*}
          u^\omega v &= u^\omega (a_1 \dots a_n) = u^\omega a_1 u^\omega (a_2 \dots a_n) = \dots = (u^\omega a_1) \dots (u^\omega a_n)\\
          &= (a_1 u^\omega) \dots (a_n u^\omega) = u^\omega (a_1 u^\omega) \dots (a_n u^\omega) = \dots = u^\omega v u^\omega \text{.}
        \end{align*}
      
        Regarding equation \ref{itm:DAbEquations:subsetOmega}, we note that the first and the last term are equal by equation~\ref{itm:DAbEquations:moveOmega}. We need to show $u^\omega v^\omega = u^\omega$ and, by \autoref{lem:groupByAlphabet}, we have
        \[
          u^\omega \H u^\omega v^\omega = (u^\omega v)^\omega \textbf{.}
        \]
        Since the left hand side and the right hand side of the above equation are both idempotents in the same $\mathcal{H}$-class, they must be equal.
        
        Finally, equation \ref{itm:DAbEquations:sameAlphabetOmega} follows directly from equation~\ref{itm:DAbEquations:subsetOmega}.
      \end{proof}
      \begin{fact}\label{fct:regularNeutralElement}
        Consider a monoid $M \in \DAb$ and a homomorphism $\varphi: \freeProf \to M$. For every $\alpha \in \freeProf$ such that $\varphi(\alpha)$ is regular, we have:
        \[
          \varphi(\alpha) = \varphi(\alpha^\omega \alpha)
        \]
        Thus, if $\alpha \in \freeProf$ is $\DAb$-regular, we have:
        \[
          \DAb \models \alpha = \alpha^\omega \alpha
        \]
      \end{fact}
      \begin{proof}
        We may write $\alpha = \lim (u_k)_k$ and there is some $k$ with $\varphi(\alpha) = \varphi(u_k)$. Let $u = u_k$.
        By definition $\varphi(\alpha) = \varphi(u)$ is contained in a regular $\mathcal{D}$-class, which is a group (since $M \in \DAb$). Thus, every power of $\varphi(u)$ is in the same class and $\varphi(u)^{M!}$ is the neutral element of the group. This shows $\varphi(u) = \varphi(u)^{M!} \varphi(u)$ and, thus, the statement.
      \end{proof}
    
      If we combine \autoref{fct:regularNeutralElement} with \autoref{fct:numberOfAsCompatibility}, we immediately obtain that $|\alpha|_a$ is infinite (i.\,e.\ $|\alpha|_a \in \hat{\mathbb{N}} \setminus \mathbb{N}$) for all $a \in \alphabet \alpha$ if $\alpha \in \freeProf$ is $\DAb$-regular.
      Together with the following result, this shows that every $\DAb$-regular pseudoword is equivalent to an infinite block over $\DAb$.
      \begin{proposition}\label{prop:regularElements}
        Consider a monoid $M \in \DAb$ and a homomorphism $\varphi: \freeProf \to M$. If $\varphi(\alpha)$ is regular for some $\alpha \in \freeProf$ with $\alphabet \alpha = \{ a_1, \dots, a_n \}$, we have:
        \[
          \varphi(\alpha) = \varphi \left( (a_1 \dots a_n)^{\omega} a_1^{|\alpha|_{a_1}} \dots a_n^{|\alpha|_{a_n}} \right)
        \]
        
        Thus, if $\alpha \in \freeProf$ is $\DAb$-regular, we have:
        \[
          \DAb \models \alpha = (a_1 \dots a_\ell)^\omega a_1^{|\alpha|_{a_1}} \dots a_\ell^{|\alpha|_{a_\ell}}
        \]
      \end{proposition}
      \begin{proof}
        Let $\alpha = \lim (u_k)_k$. We may choose $k$ large enough such that
        \[
          \varphi(\alpha) = \varphi(u_k), \quad \alphabet \alpha = \alphabet u_k \quad \text{and} \quad \varphi \left( a_i^{|u_k|_{a_i}} \right) = \varphi \left( a_i^{|\alpha|_{a_i}} \right)
        \]
        for all $i = 1, \dots, \ell$. Let $u = u_k$ and identify $\gamma$ with $\varphi(\gamma)$ for all $\gamma \in \freeProf$.
        We have in $M$:
        \begin{flalign*}
          && \alpha = u &= u^\omega u & \text{(by \autoref{fct:regularNeutralElement})}\\
          && &= u^\omega a_1^{|u|_{a_1}} \dots a_n^{|u|_{a_n}} & \text{(by \autoref{fct:permuteInHClass})}\\
          && &= (a_1 \dots a_n)^\omega a_1^{|u|_{a_1}} \dots a_n^{|u|_{a_n}} & \text{(by \autoref{fct:DAbEquations}, \ref{itm:DAbEquations:sameAlphabetOmega})}\\
          && &= (a_1 \dots a_n)^{\omega} a_1^{|\alpha|_{a_1}} \dots a_n^{|\alpha|_{a_n}} &\qedhere
        \end{flalign*}
      \end{proof}

      Combining \autoref{cor:regularFiniteFactorization} and \autoref{prop:regularElements}, we obtain that, for every pseudoword $\alpha$, there is one $\alpha'$ equivalent over $\DAb$ (i.\,e.\ with $\DAb \models \alpha = \alpha'$) such that $\alpha'$ is a concatenation of letters and infinite blocks. Thus, the only remaining part to finally obtain the normal form is to satisfy the conditions \ref{itm:normal:incomparable} and \ref{itm:normal:omega} from \autoref{def:normalForm}.
      
      We first consider the case that there is a letter contained in the alphabet of an adjacent infinite block. In this case, we may simply move the letter to the corresponding $a_i^{\hat{n}_i}$ block.
      \begin{fact}\label{fct:normalForm:reduceLetterRegular}
        For an infinite block
        \[
          \alpha = (a_1 \dots a_\ell)^\omega a_1^{\hat{n}_1} \dots a_\ell^{\hat{n}_\ell} \textbf{,}
        \]
        we have
        \[
          \DAb \models a_i \alpha = (a_1 \dots a_\ell)^\omega a_1^{\hat{n}_1} \dots a_{i - 1}^{\hat{n}_{i - 1}} a_i^{\hat{n}_i + 1} a_{i + 1}^{\hat{n}_{i + 1}} \dots a_\ell^{\hat{n}_\ell} = \alpha a_i
        \]
        for all $1 \leq i \leq \ell$.
      \end{fact}
      \begin{proof}
        We can move the letter $a_i$ to the new position by \autoref{fct:permuteInHClass} and \autoref{fct:DAbEquations}, \ref{itm:DAbEquations:moveOmega}.
      \end{proof}
    
      It remains to handle two adjacent infinite blocks. If they share a common letter, we may move copies of it to the left one of the two blocks until there are exactly $\omega$ many occurrences in the right block (we need this to satisfy \ref{itm:normal:omega} from \autoref{def:normalForm}). Although this is the direction we need to get a pseudoword into normal form, we also have a symmetric version where we move the excessive occurrences of the shared letter to the right. We will use this shortly.
      \begin{fact}\label{fct:moveBetweenInfiniteBlocks}
        For two infinite blocks
        \[
          \alpha = (a_1 \dots a_k)^\omega a_1^{\hat{m}_1} \dots a_k^{\hat{m}_k}
          \quad \text{and} \quad
          \beta = (b_1 \dots b_\ell)^\omega b_1^{\hat{n}_1} \dots b_\ell^{\hat{n}_\ell}
        \]
        with $a_i = b_j$ (for some $1 \leq i \leq k$ and $1 \leq j \leq \ell$), $\DAb$ satisfies
        \[
          \alpha \beta = (a_1 \dots a_k)^\omega a_1^{\hat{m}_1} \dots a_{i - 1}^{\hat{m}_{i - 1}} a_i^{\hat{m}_i + \hat{n}_j} a_{i + 1}^{\hat{m}_{i + 1}} \dots a_k^{\hat{m}_k} \,
          (b_1 \dots b_\ell)^\omega b_1^{\hat{n}_1} \dots b_{j - 1}^{\hat{n}_{j - 1}} b_j^{\omega} b_{j + 1}^{\hat{n}_{j + 1}} \dots a_\ell^{\hat{n}_\ell} \textbf{.}
        \]
        and
        \[
          \alpha \beta = (a_1 \dots a_k)^\omega a_1^{\hat{m}_1} \dots a_{i - 1}^{\hat{m}_{i - 1}} a_i^{\omega} a_{i + 1}^{\hat{m}_{i + 1}} \dots a_k^{\hat{m}_k} \,
          (b_1 \dots b_\ell)^\omega b_1^{\hat{n}_1} \dots b_{j - 1}^{\hat{n}_{j - 1}} b_j^{\hat{m}_i + \hat{n}_j} b_{j + 1}^{\hat{n}_{j + 1}} \dots a_\ell^{\hat{n}_\ell} \textbf{.}
        \]
      \end{fact}
      \begin{proof}
        We have $\hat{m}_j = \omega + \hat{m}_j$ since $\hat{m}_j \in \hat{\mathbb{N}} \setminus \mathbb{N} = \{ \omega + \hat{n} \mid \hat{n} \in \hat{\mathbb{N}} \}$ (by the definition of an infinite block) and, thus, can replace $b_j^{\hat{m}_j}$ by $b_j^{\omega + \hat{m}_j}$. Afterwards, we may move the $b_j^{\hat{m}_j}$ block to the left (right next to the $a_i^{\hat{n}_i} = b_j^{\hat{n}_i}$ block) by \autoref{fct:permuteInHClass} and \autoref{fct:DAbEquations}, \ref{itm:DAbEquations:moveOmega}. The symmetric result follows from the same arguments.
      \end{proof}
    
      Finally, we have to make sure that two adjacent infinite blocks have incomparable alphabets. For this, we consider the converse situation where the alphabet of an infinite block is contained in that of an adjacent infinite block. We may use \autoref{fct:moveBetweenInfiniteBlocks} to ensure that the infinite block with the smaller alphabet (or an arbitrary one of the two if they have the same alphabet) has only $\omega$ exponents. Then we may absorb the entire infinite block into the other one.
      \begin{fact}\label{fct:normalForm:reduceRegularSubsetAlphabet}
        Consider an infinite block
        \begin{align*}
          \alpha &= (a_1 \dots a_k)^\omega a_1^{\hat{m}_1} \dots a_k^{\hat{m}_k}
        \shortintertext{and an infinite block}
          \beta &= (b_1 \dots b_\ell)^\omega b_1^\omega \dots b_\ell^\omega
        \end{align*}
        with $\alphabet \beta \subseteq \alphabet \alpha$ (and $|\beta|_b = \omega$ for all $b \in \alphabet \beta$). Then we have:
        \[
          \DAb \models \alpha \beta = \alpha = \beta \alpha
        \]
      \end{fact}
      \begin{proof}
        This holds by \autoref{fct:DAbEquations}, \ref{itm:DAbEquations:moveOmega}. First, we have $\DAb \models \alpha (a_1 \dots a_k)^\omega$ (to make the situation symmetric) and then we may absorb the $\omega$ blocks (again by \autoref{fct:DAbEquations}, \ref{itm:DAbEquations:moveOmega}).
      \end{proof}
    
      We have now all the necessary operations to re-write a pseudoword in normal form.
      \begin{proposition}\label{prop:normalFormExists}
        For every pseudoword $\alpha$, there is a pseudoword $\alpha'$ in normal form with $\DAb \models \alpha = \alpha'$.
      \end{proposition}
    \end{subsection}
    \begin{subsection}{The Normal Form is Unique}
      To show that the normal form is indeed a normal form, we have to show that two distinct pseudowords in normal form can be distinguished in $\DAb$. For this, we fix two distinct (i.\,e.\ $\alpha \neq \beta$) pseudowords
      \begin{align*}
        \alpha = \alpha_1 \dots \alpha_K \quad \text{and} \quad \beta = \beta_1 \dots \beta_L
      \end{align*}
      in normal form. We need to show that there is a monoid $M \in \DAb$ and a homomorphism $\varphi: \freeProf \to M$ with $\varphi(\alpha) \neq \varphi(\beta)$.
      
      \paragraph*{Distinguishing Normal Forms in $\V[J]$.}
      If we have $K \neq L$ or there is some $i$ with $\alpha_i \in \Sigma$ but $\beta_i \not\in \Sigma$ (or vice verse) or with $\alphabet \alpha_i \neq \alphabet \beta_i$, we can already distinguish $\alpha$ and $\beta$ over $\V[J]$, the variety of finite monoids whose $\mathcal{J}$-classes are trivial. Note that we have $\V[J] \subseteq \DAb$ since, in finite monoids, $\mathcal{J}$ and $\mathcal{D}$ coincide and, thus, every regular $\mathcal{D}$-class is trivial (it consists of a single idempotent) and is, therefore, a trivial (and, in particular, Abelian) group.
      
      For this approach, we first repeat the normal form for pseudowords over $\V[J]$ (see \cite[Section~8.2]{alm94:short}).
      \begin{definition}
        A pseudoword
        \[
          \freeProf \ni \gamma = \gamma_1 \dots \gamma_N
        \]
        is in $\V[J]$-normal form if each $\alpha_i$ (with $1 \leq i \leq N$) is either
        \begin{itemize}
          \item a letter (i.\,e.\ $\gamma_i \in \Sigma$) or
          \item $\gamma_i = (c_1 \dots c_n)^\omega$ for $c_1 < \dots < c_n$ and we have
          \begin{enumerate}[label=(c'-\arabic*), leftmargin=1.5cm]
            \item\label{itm:Jnormal:incomparable}
             $\alphabet \gamma_{i - 1} \setminus \alphabet \gamma_i \neq \emptyset$ (if $1 < i$) and $\alphabet \gamma_{i + 1} \setminus \alphabet \gamma_i \neq \emptyset$ (if $i < N$).
          \end{enumerate}
        \end{itemize}
      \end{definition}
    
      This already matches closely our definition for normal forms (for $\DAb$). In fact, to get a normal form (for $\DAb$) into $\V[J]$-normal form, we only have to reduce the infinite blocks:
      \begin{fact}\label{fct:JNormalOfInfiniteBlock}
        For an infinite block
        \[
          \gamma = (c_1 \dots c_m)^\omega c_1^{\hat{n}_1} \dots c_m^{\hat{n}_m} \text{,}
        \]
        we have $\V[J] \models \gamma = (c_1 \dots c_m)^\omega$.
      \end{fact}
      \begin{proof}
        This follows from the reduction rules for normal forms over $\V[J]$ \cite[p.~228, Section~8.2, rule rr.3)]{alm94:short}.
      \end{proof}
      
      \begin{proposition}
        Consider the pseudowords $\alpha = \alpha_1 \dots \alpha_K$ and $\beta = \beta_1 \dots \beta_L$ in normal form with $\alpha \neq \beta$. If we have
        \begin{itemize}
          \item $K \neq L$ or
          \item $K = L$ and there is some $i$ with $1 \leq i \leq K = L$ with \begin{itemize}
            \item $\alpha_i \in \Sigma$ but $\beta_i \not\in \Sigma$ (i.\,e.\ $\alpha_i$ is a letter but $\beta_i$ is an infinite block),
            \item $\alpha_i \not\in \Sigma$ but $\beta_i \in \Sigma$ (symmetric situation) or
            \item $\alphabet \alpha_i \neq \alphabet \beta_i$,
          \end{itemize}
        \end{itemize}
        then we have $\DAb \not\models \alpha = \beta$.
      \end{proposition}
      \begin{proof}
        If we reduce every infinite block $\alpha_i$ with $\alphabet \alpha_i = \{ a_1, \dots, a_k \}$ (for $a_1 < \dots < a_k$) to $(a_1 \dots a_k)^\omega$ (according to \autoref{fct:JNormalOfInfiniteBlock}), we obtain a pseudoword $\alpha'$ in $\V[J]$-normal form with $\V[J] \models \alpha = \alpha'$. In the same way, we can also obtain an equivalent pseudoword $\beta'$ in $\V[J]$-normal form from $\beta$. By hypothesis, we have $\alpha' \neq \beta'$ and, thus, there is a monoid $M \in \V[J]$ and a homomorphism $\varphi: \freeProf \to M$ with $\varphi(\alpha) = \varphi(\alpha') \neq \varphi(\beta') = \varphi(\beta)$ \cite[Section~8.2]{alm94:short}. Since we have $M \in \V[J] \subseteq \DAb$, this shows $\DAb \not\models \alpha = \beta$.
      \end{proof}
    
      \paragraph*{The Remaining Case.}
      From now on, we may assume $K = L$, $\alpha_i \in \Sigma \iff \beta_i \in \Sigma$ and $\alphabet \alpha_i = \alphabet \beta_i$ for all $1 \leq i \leq K = L$.
      
      We will use the ranker characterization from above to construct a monoid $M \in \DAb$ and a homomorphism $\varphi: \freeProf \to M$ with $\varphi(\alpha) \neq \varphi(\beta)$. For this, we define finite words that are equivalent to the individual blocks of our normal form.
      \begin{definition}
        For an infinite block
        \[
          \freeProf \ni \gamma = (c_1 \dots c_m)^\omega c_1^{\hat{n}_{1}} \dots c_m^{\hat{n}_m}
        \]
        and a monoid $M \in \DAb$, let
        \begin{align*}
          r_M(\gamma) &= (c_1 \dots c_m)^{M!} c_1^{\hat{n}_1 \bmod M!} \dots c_m^{\hat{n}_m \bmod M!}
        \shortintertext{%
          Furthermore, for $c \in \Sigma$, we let}
          r_M(c) &= c \text{.}
        \end{align*}
      \end{definition}
      \begin{fact}\label{fct:infiniteBlockAsFiniteWord}
        Let $M \in \DAb$ and $\varphi: \freeProf \to M$ be a homomorphism. For any infinite block $\gamma$ or $\gamma \in \Sigma$, we have:
        \[
          \varphi(\gamma) = \varphi \left( r_M(\gamma) \right)
        \]
      \end{fact}
      \begin{proof}
        There is nothing to show for $\gamma \in \Sigma$. Therefore, let $\gamma = (c_1 \dots c_m)^\omega c_1^{\hat{n}_{1}} \dots c_m^{\hat{n}_m}$. We have:
        \begin{flalign*}\allowdisplaybreaks
          && \varphi(\gamma) &= \varphi \left( (c_1 \dots c_m)^\omega c_1^{\hat{n}_{1}} \dots c_m^{\hat{n}_m} \right) \\
          && &= \varphi \Big( (c_1 \dots c_m)^{M!} \\
          && &\phantom{= \varphi \Big(} \qquad c_1^{M! + (\hat{n}_{1} \bmod M!)} \dots c_m^{M! + (\hat{n}_m \bmod M!)} \Big) & \text{(by \autoref{fct:profiniteMod})}\\
          && &= \varphi \Big( (c_1 \dots c_m)^{M!} & \text{} \quad\\
          && &\phantom{= \varphi \Big(} \qquad c_1^{\hat{n}_{1} \bmod M!} \dots c_m^{\hat{n}_m \bmod M!} \Big) & \text{\makebox[0pt][r]{(by \autoref{fct:permuteInHClass} \& \autoref{fct:DAbEquations}, \ref{itm:DAbEquations:subsetOmega})}} \\
          && &=\varphi \left( r_M(\gamma) \right) &\qedhere
        \end{flalign*}
      \end{proof}

      \begin{fact}\label{fct:numberOfAGammaRM}
        For any infinite block $\gamma$ or $\gamma \in \Sigma$ and all $a \in \Sigma$, we have
        \[
          |\gamma|_a \equiv |r_M(\gamma)|_a \mod M!
        \]
        for all $M \in \DAb$.
      \end{fact}
      \begin{proof}
        If we have $a \not\in \alphabet \gamma = \alphabet r_M(\gamma)$, there is nothing to show. Similarly, there is nothing to show for $\gamma \in \Sigma$. For $\gamma = (c_1 \dots c_m)^\omega c_1^{\hat{n}_{1}} \dots c_m^{\hat{n}_m}$ and $a = c_i$ for some $1 \leq i \leq m$, we have:
        \[
          |\gamma|_a = \omega + \hat{n}_i = \hat{n}_i \equiv \hat{n}_i \bmod M! \equiv M! + (\hat{n}_i \bmod M!) = |r_M(\gamma)|_a \mod M! \qedhere
        \]
      \end{proof}

      \begin{proposition}\label{prop:normalFormIsUnique}
        There is a monoid $M \in \DAb$ and a continuous homomorphism $\varphi: \freeProf \to M$ with $\varphi(\alpha) \neq \varphi(\beta)$.
      \end{proposition}
      \begin{proof}
        We are in the case $K = L$, $\alpha_i \in \Sigma \iff \beta_i \in \Sigma$ and $\alphabet \alpha_i = \alphabet \beta_i$ for all $1 \leq i \leq K = L$. The only remaining possibility for $\alpha \neq \beta$ is that there is some $I$ such that
        \begin{align*}
          \alpha_I = (a_1 \dots a_k)^\omega a_1^{\hat{m}_1} \dots a_k^{\hat{m}_k}
          \quad \text{and} \quad
          \beta_I = (a_1 \dots a_k)^\omega a_1^{\hat{n}_1} \dots a_k^{\hat{n}_k}
        \end{align*}
        with $\hat{m}_j \neq \hat{n}_j$ for some $1 \leq j \leq k$. Let $a = a_j$. By \autoref{fct:distinguishInfiniteNumbers}, there is some $m$ such that $|\alpha_I|_a = \hat{m}_j \not\equiv \hat{n}_j = |\beta_I|_a \bmod m$. Without loss of generality, we may assume $I$ to be minimal with this property. In particular, we have $|\alpha_i|_a = |\beta_i|_a$ for all $1 \leq i < I$.
        
        We choose $M = \Sigma^* / {\approx_{m, K + 1}}$ and $\varphi$ as the continuation of the natural projection sending $w \in \Sigma^*$ to the congruence class of $w$ with respect to $\approx_{m, K + 1}$. Note that $a^k$ is in a different congruence class for each $0 \leq k < m$ and that we, thus, have $m \leq |M|$ and $m$ divides $M!$. In particular, we also obtain $|\alpha_I|_a \not\equiv |\beta_I|_a \bmod M!$ but $|\alpha_i|_a \equiv |\beta_i|_a \bmod M!$ for all $1 \leq i < I$.
        
        We need to show $\varphi(\alpha) \neq \varphi(\beta)$. For this, we define some $u \in \Sigma^*$ with $\varphi(u) = \varphi(\alpha)$ and some $v \in \Sigma^*$ with $\varphi(v) = \varphi(\beta)$. Showing $\varphi(\alpha) \neq \varphi(\beta)$ then boils down to showing $\varphi(u) \neq \varphi(v)$ or, in other words, $u \not\approx_{m, K + 1} v$.
        
        In order to define $u$ and $v$, let $u_i = r_M(\alpha_i)$ and $v_i = r_M(\beta_i)$ again for all $1 \leq i \leq K$. Note that we now have $\varphi(\alpha_i) = \varphi(u_i)$ and $\varphi(\beta_i) = \varphi(v_i)$ by \autoref{fct:infiniteBlockAsFiniteWord} for all $1 \leq i \leq K$. Thus, we have $\varphi(\alpha) = \varphi(u)$ for $u = u_1 \dots u_K$ and $\varphi(\beta) = \varphi(v)$ for $v = v_1 \dots v_K$.
        
        Now, let $c_i \in \alphabet \alpha_i \setminus \alphabet \alpha_{i - 1} = \alphabet \beta_i \setminus \alphabet \beta_{i - 1}$ and $d_i \in \alphabet \alpha_i \setminus \alphabet \alpha_{i + 1} = \alphabet \beta_i \setminus \alphabet \beta_{i + 1}$ be arbitrary\footnote{We may choose $c_i$ and $d_i$ minimal with respect to the linear ordering of $\Sigma$, for example.} for all $1 \leq i \leq K$ (where we let $\alpha_0 = \alpha_{K + 1} = \beta_0 = \beta_{K + 1} = \varepsilon$ from now on to handle edge cases). With these, let $\rho_i = X_{c_1} \dots X_{c_i}$ and $\lambda_i = Y_{d_K} \dots Y_{d_i}$ for $1 \leq i \leq K$. Using a simple induction, it is easy to see that $\rho_i$ ends in $u$ at a position belonging to $u_i$ (in both cases: if $\alpha_i$ is an infinite block and if $\alpha_i \in \Sigma$). In fact, if $u_i = r_M(\alpha_i) = (b_1 \dots b_\ell)^{M!} b_1^{p_1} \dots b_\ell^{p_\ell}$ belongs to an infinite block, $\rho_i$ will end at a position belonging to the first copy of $(b_1 \dots b_\ell)$. Using similar arguments, we may observe that $\lambda_i$ also ends in $u$ at a position belonging to $u_i$. In fact, for an infinite block, $\lambda_i$ will either end at a position belonging to some $b_j^{p_j}$ block or at the last copy of $(b_1 \dots b_\ell)$. This shows that we have $\rho_i(u) < \lambda_i(u)$ if $\alpha_i$ is an infinite block and that we have $\rho_i(u) = \lambda_i(u)$ if $\alpha_i \in \Sigma$. In any case, we have $\rho_i(u) \leq \lambda_i(u)$. Finally, all these statements also hold analogously for $v$. An example is given in \autoref{fig:exampleUWithRankers}.

        Note that we have $|\rho_i| = i$ and $|\lambda_i| = K - i + 1$ for all $1 \leq i \leq K$. Thus, we have $|\rho_i| + |\lambda|_i = K + 1$. To handle the edge cases, we let $\rho_{K + 1} = \Xinfty$ and $\lambda_{K + 1} = 0$. Here, we have $|\rho_{K + 1}| + |\lambda_{K + 1}| = 1 \leq K + 1$.
        
        Let $J$ be minimal with $I < J$ and $a \not\in \alphabet \alpha_J$ ($J = K + 1$ for $\alpha_{K + 1} = \varepsilon$ is possible). We will show $u \not\sim_{m, \rho, \lambda} v$ for $\rho = \rho_J$ and $\lambda = \lambda_J$. This, in turn, shows $u \not\approx_{m, K + 1} v$ (as we have $|\rho| + |\lambda| \leq K + 1$) and we are done.
        
        As we have $\rho(u) \leq \lambda(u)$ (by the above argumentation), we may consider the $\rho$-$\lambda$-split of $u$. Its middle part is a (possibly empty) factor or $u_J$ and, thus, does not contain an $a$ (as we have $\alphabet u_J = \alphabet r_M(\alpha_J) = \alphabet \alpha_J \not\ni a$). The same also holds for the middle part of the $\rho$-$\lambda$-split of $v$.
        
        The left part of the $\rho$-$\lambda$-split of $u$ is
        \begin{align*}
          u' &= u_1 \dots u_{I - 1} u_I u_{I + 1} \dots u_{J - 1} u_J'
        \intertext{%
        where $u_J'$ is some prefix of $u_J$. Please note that this prefix cannot contain an $a$ (as $u_J$ itself does not contain any $a$). In the same way, the left part of the $\rho$-$\lambda$-split of $v$ is}
          v' &= v_1 \dots v_{I - 1} v_I v_{I + 1} \dots v_{J - 1} v_J'
        \end{align*}
        for some prefix $v_J'$ with $a \not\in \alphabet v_J' \subseteq \alphabet v_J = \alphabet \beta_J = \alphabet \alpha_J$.
        
        By the choice of $I$ and by \autoref{fct:numberOfAGammaRM}, we have
        \[
          |u_i|_a \equiv |\alpha_i|_a = |\beta_i|_a \equiv |v_i|_a \mod M!
        \]
        for all $1 \leq i < I$.
        
        For $I < i < J$, we observe that $\alpha_i$ must be an infinite block, which (by the choice of $J$) contains an $a$. Indeed assume $\alpha_i \in \Sigma$ for one such $i$. Without loss of generality, we may consider the smallest $i$ with this property. By the choice of $J$, we have $\alpha_i = a$. However, the block $\alpha_{i - 1}$ to the left of $\alpha_i$ must be an infinite block (as $i$ is minimal) with $a \in \alphabet \alpha_{i - 1}$ (by the choice of $J$). However, this contradicts condition \ref{itm:normal:incomparable} of the definition of normal forms.
        
        However, if all of these $\alpha_i$ (and, thus, also $\beta_i$) are infinite blocks with $a \in \alphabet \alpha_i$, we obtain $|\alpha_i|_a = |\beta_i|_a = \omega$ from condition \ref{itm:normal:omega} of the definition of normal forms. This yields $|u_i|_a \equiv |\alpha_i|_a = \omega \equiv 0 \bmod M!$ by \autoref{fct:numberOfAGammaRM} and, analogously, $|v_i|_a \equiv 0 \bmod M!$ for all $I < i < J$.
        
        Combining all of these results, we obtain
        \begin{align*}
          |u'|_a - |v'|_a &= |u_1 \dots u_{I - 1} u_I u_{I + 1} \dots u_{J - 1} u_J'|_a - |v_1 \dots v_{I - 1} v_I v_{I + 1} \dots v_{J - 1} v_J'|_a \\
          &= \left( \sum_{i = 1}^{I - 1} |u_i|_a \right) + |u_I|_a + \left( \sum_{i = I + 1}^{J - 1} \underbrace{|u_i|_a}_{{}\equiv 0} \right) + \underbrace{|u_J'|_a}_{{}=0} \\
          &\qquad -\left( \sum_{i = 1}^{I - 1} \smash[b]{\underbrace{|v_i|_a}_{{}\equiv |u_i|_a}} \right) -|v_I|_a -\left( \sum_{i = I + 1}^{J - 1} \smash[b]{\underbrace{|v_i|_a}_{{}\equiv0}} \right) -\underbrace{|v_J'|_a}_{{}=0} \\
          &\equiv |u_I|_a - |v_I|_a \not\equiv 0 \mod M!
        \end{align*}
        where the last congruence follows from $|u_I|_a \not\equiv |v_i|_a \bmod M!$. This show $|u'|_a \not\equiv |v'|_a \bmod M!$, i.\,e.\ that the left parts of the $\rho$-$\lambda$-splits of $u$ and $v$ do not have the same number of $a$ modulo $m$ although in the middle parts there is no $a$. This shows $u \not\sim_{m, \rho, \lambda}$ as desired.
      \end{proof}
      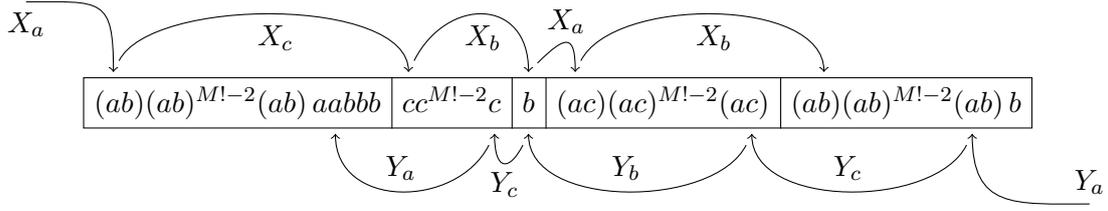
\begin{figure}\centering
        \begin{tikzpicture}[baseline=(ul.base)]
          \matrix [rectangle, draw, matrix of math nodes, ampersand replacement=\&, inner ysep=2pt, inner xsep=0pt, text height=.8em, text depth=.2em] (u) {
            \hspace*{4pt} ( \& a \& b) (ab)^{M! - 2} (ab) \, a \& a \& bbb \hspace*{4pt} \&
            \hspace*{4pt} \& c \& c^{M! - 2} \& c \& \hspace*{4pt} \&
            \hspace*{4pt} \& b \& \hspace*{4pt} \&
            \hspace*{4pt} ( \& a \& c) (ac)^{M! - 2} (a \& c \& ) \hspace*{4pt} \&
            \hspace*{4pt} (a \& b \& ) (ab)^{M! - 2} ( \& a \& b) \, b \hspace*{4pt}\\
          };
          \foreach \i in {5,10,13,18} {
            \draw ($(u.north -| u-1-\i.north east)$) -- ($(u.south -| u-1-\i.south east)$);
          };
          
          \draw[->] ($(u.north west) + (-0.75cm, 1cm)$) .. controls +(1cm, 0cm) and +(0cm,1cm) .. node[below, pos=0] {$X_a$} ($(u.north -| u-1-2.north)+(0pt,2pt)$);
          \draw[->, shorten <= 1ex] ($(u.north -| u-1-2.north)+(0pt,2pt)$) .. controls +(0.5cm, 1cm) and +(0pt,1cm) .. node[below] {$X_{c}$} ($(u.north -| u-1-7.north)+(0pt,2pt)$);
          \draw[->, shorten <= 1ex] ($(u.north -| u-1-7.north)+(0pt,2pt)$) .. controls +(0.5cm, 1cm) and +(0pt,1cm) .. node[below] {$X_{b}$} ($(u.north -| u-1-12.north)+(0pt,2pt)$);
          \draw[->, shorten <= 1ex] ($(u.north -| u-1-12.north)+(0pt,2pt)$) .. controls +(0.5cm, 0.5cm) and +(0pt,0.5cm) .. node[above] {$X_{a}$} ($(u.north -| u-1-15.north)+(0pt,2pt)$);
          \draw[->, shorten <= 1ex] ($(u.north -| u-1-15.north)+(0pt,2pt)$) .. controls +(0.5cm, 1cm) and +(0pt,1cm) .. node[below] {$X_{b}$} ($(u.north -| u-1-20.north)+(0pt,2pt)$);
          
          \draw[->] ($(u.south east) + (0.75cm, -1cm)$) .. controls +(-1cm, 0cm) and +(0cm,-1cm) .. node[above, pos=0] {$Y_a$} ($(u.south -| u-1-22.south)+(0pt,-2pt)$);
          \draw[->, shorten <= 1ex] ($(u.south -| u-1-22.south)+(0pt,-2pt)$) .. controls +(-0.5cm, -1cm) and +(0pt,-1cm) .. node[above] {$Y_{c}$} ($(u.south -| u-1-17.south)+(0pt,-2pt)$);
          \draw[->, shorten <= 1ex] ($(u.south -| u-1-17.south)+(0pt,-2pt)$) .. controls +(-0.5cm, -1cm) and +(0pt,-1cm) .. node[above] {$Y_{b}$} ($(u.south -| u-1-12.south)+(0pt,-2pt)$);
          \draw[->, shorten <= 1ex] ($(u.south -| u-1-12.south)+(0pt,-2pt)$) .. controls +(-0.25cm, -0.5cm) and +(0pt,-0.5cm) .. node[below] {$Y_{c}$} ($(u.south -| u-1-9.north)+(0pt,-2pt)$);
          \draw[->, shorten <= 1ex] ($(u.south -| u-1-9.north)+(0pt,-2pt)$) .. controls +(-0.5cm, -1cm) and +(0pt,-1cm) .. node[above] {$Y_{a}$} ($(u.south -| u-1-4.south)+(0pt,-2pt)$);
        \end{tikzpicture}
        \caption{An example word $u = u_1 \dots u_5$ for $u_1 = (ab)^{M!} a^2 b^3$, $u_2 = c^{M!} c^0$, $u_3 = b$, $u_4 = (ac)^{M!}a^0 c^0$ and $u_5 = (ab)^{M!} a^0 b^1$ with the rankers $\rho_i$ and $\lambda_i$}\label{fig:exampleUWithRankers}
      \end{figure}
    
      Together with \autoref{prop:normalFormExists}, we now have proved that the form from \autoref{def:normalForm} is indeed a normal form for $\DAb$.
      \begin{theorem}\label{thm:normalForm}
        For every pseudoword $\alpha$, there is a unique word $\alpha'$ in normal form with $\DAb \models \alpha = \alpha'$.
      \end{theorem}
    
      \paragraph*{Semigroups.}
      Again, we also obtain a version of the previous result for semigroups.
      A semigroup $S$ \emph{satisfies} an equation $\alpha = \beta$ of two pseudowords $\alpha, \beta \in \hat{\Sigma}^+ = \hat{\Sigma}^* \setminus \{ \varepsilon \}$ if we have $\varphi(\alpha) = \varphi(\beta)$ for all continuous homomorphisms $\varphi: \hat{\Sigma}^+ \to S$. A class $\mathcal{C}$ of semigroups \emph{satisfies} the equation $\alpha = \beta$ if every semigroup in $\mathcal{C}$ does. Note that a monoid satisfies an equation (as a monoid) if and only if it satisfies it as a semigroup.
      
      \begin{fact}\label{fct:semigroupEquationsVsMonoidEquations}
        Let $\alpha, \beta \in \hat{\Sigma}^+$ be non-empty pseudowords. We have:
        \[
          \DAb_S \models \alpha = \beta \iff \DAb \models \alpha = \beta
        \]
      \end{fact}
      \begin{proof}
        If $\DAb_S \models \alpha = \beta$, we have in particular $\DAb \models \alpha = \beta$ since $\DAb \subseteq \DAb_S$.
        
        For the other direction, consider $S \in \DAb_S$ and a continuous homomorphism $\varphi: \hat{\Sigma}^+ \to S$ and let $\DAb \models \alpha = \beta$. We have to show $\varphi(\alpha) = \varphi(\beta)$. We may extend $\varphi$ into a continuous homomorphism $\varphi': \hat{\Sigma}^* \to S^1$ by letting $\varphi'(\varepsilon) = 1$ and $\varphi'(\gamma) = \varphi(\gamma) \in S$ for all $\gamma \in \hat{\Sigma}^+$. Since $S^1$ is a monoid and in $\DAb$, we have $\varphi(\alpha) = \varphi'(\alpha) = \varphi'(\beta) = \varphi(\beta)$.
      \end{proof}
      
      \autoref{fct:semigroupEquationsVsMonoidEquations} immediately yields the following corollary of \autoref{thm:normalForm}.
      \begin{corollary}
        For every pseudoword $\alpha \in \hat{\Sigma}^+$, there is a unique word $\alpha' \in \hat{\Sigma}^+$ in normal form with $\DAb_S \models \alpha = \alpha'$.
      \end{corollary}
    \end{subsection}
  
    \begin{subsection}{The Normal Form of \texorpdfstring{$(\omega - 1)$}{(ω - 1)}-Terms can be Computed}
      As an $(\omega - 1)$-term arises from letters by using concatenation and $(\omega - 1)$-powers, we may consider them as finite words over an alphabet containing $\Sigma$, opening parenthesis $($ and closing parenthesis $)^{\omega - 1}$ with an $(\omega - 1)$-power. Equivalently, we can also compute a tree-like representation where the leaves are labeled by letters, we have inner nodes with two or more children labeled as concatenation nodes and inner nodes with a single child labeled by $(\omega - 1)$.
      
      These object can naturally be used as inputs or outputs for algorithms and we will show that we may compute the normal form of a given $(\omega - 1)$-term:
      \begin{theorem}
        The problem
        \compProblem{
          an $(\omega - 1)$-term $\alpha$
        }{
          the normal form of $\alpha$
        }\noindent
        is computable.
      \end{theorem}
    
      By computing the normal forms of the two input $(\omega - 1)$-terms and comparing them, we immediately also obtain an algorithm to decide whether two given $(\omega - 1)$-terms form an equation that is satisfied by $\DAb$:
      \begin{corollary}\label{cor:WPdecidable}
        The word problem for $(\omega - 1)$-terms over $\DAb$
        \decProblem{
          two $(\omega - 1)$-terms $\alpha$ and $\beta$
        }{
          does $\DAb \models \alpha = \beta$ hold?
        }\noindent
        is decidable.
      \end{corollary}
    
      To compute the normal form of a given $(\omega - 1)$-term, we first split the input term $\alpha$ into blocks of single letters and blocks of the form $(\beta)^{\omega - 1}$, i.\,e.\ we split $\alpha = \alpha_1 \dots \alpha_L$ where each $\alpha_i$ is either a letter or $\alpha_i = (\beta_i)^{\omega - 1}$ for some $(\omega - 1)$-term $\beta_i$.
      
      The letters are already in normal form and, for the blocks of the form $(\beta)^{\omega - 1}$, we compute an equivalent term in normal form. Afterwards we may use the rules given by \autoref{fct:normalForm:reduceLetterRegular}, \autoref{fct:moveBetweenInfiniteBlocks} and \autoref{fct:normalForm:reduceRegularSubsetAlphabet} to compute the normal form for the entire term.
      
      Note that a block of the form $(\beta)^{\omega - 1}$ is $\DAb$-regular (the image of $(\beta)^{\omega - 1}$ is $\mathcal{J}$-equivalent to the idempotent given by $(\beta)^\omega$ in any semigroup $S$ under any continuous morphism $\freeProf \to S$). Thus, we may use the connection given in \autoref{prop:regularElements} to compute the normal form. For this, we only need to compute the alphabet of $\beta$ (which clearly can be obtained for any given $(\omega - 1)$-term $\beta$) and $|(\beta)^{\omega - 1}|_a$ for all $a$ in this alphabet. This latter number can be computed as well (as the following result states), which completes the algorithm for computing the normal form.
      \begin{proposition}
        For an $(\omega - 1)$-term $\alpha$ and $a \in \Sigma$ we either have $|\alpha|_a = n \in \mathbb{N}$ or $|\alpha|_a = \omega + z$ for $z \in \mathbb{Z}$. Furthermore, which case applies and the values of $n$ and $z$ can be computed.
      \end{proposition}
      \begin{proof}
        The statement can be shown by an induction on the structure of $(\omega - 1)$-term $\alpha$. This induction immediately also yields an algorithm.
        
        For a single letter $b \in \Sigma$, we have $|b|_a = 1\in \mathbb{N}$ if $a = b$ and $|b|_a = 0 \in \mathbb{N}$ otherwise. If we have $\alpha = \alpha_1 \cdot \alpha_2$, we have $|\alpha|_a = |\alpha_1|_a + |\alpha_2|_a$. The rest of this case follows from the following observations:
        \begin{itemize}
          \item For $n, m \in \mathbb{N} \subseteq \hat{\mathbb{N}}$, we have $n + m \in \mathbb{N}$.
          \item For $n \in \mathbb{N}$ and $z \in \mathbb{Z}$, we have $n + (\omega + z) = (\omega + z) + n = \omega + (z + n) \in \omega + \mathbb{Z}$.
          \item For $x, y \in \mathbb{Z}$, we have $(\omega + x) + (\omega + y) = \omega + (x + y) \in \omega + \mathbb{Z}$ (as $\omega$ is idempotent).
        \end{itemize}
        Finally, for $\alpha = ( \beta )^{\omega - 1}$, we have $|\alpha|_a = (\omega - 1) \cdot |\beta|_a$. Since we have $(\omega - 1) \cdot n = \omega - n \in \omega + \mathbb{Z}$ for $n \in \mathbb{N}$ and $(\omega - 1) \cdot (\omega + z) = \omega - z \in \omega + \mathbb{Z}$, the statement follows.
      \end{proof}
    \end{subsection}
  \end{section}
  
  \section*{Acknowledgments}
The first author acknowledges partial support by CMUP (Centro de
Matemática da Universidade do Porto), member of LASI (Intelligent
Systems Associate Laboratory), which is financed by Portuguese funds
through FCT (Fundação para a Ciência e a Tecnologia, I. P.)
under the projects UIDB/00144/2020 and UIDP/00144/2020.

  Large parts of this work were developed while the third author was at CMUP (which is financed by national funds through FCT -- Fundação para a Ciência e Tecnologia, I.P., under the project with reference UIDB/00144/2020). Other parts were written while he was affiliated with the Politecnico di Milano where he was funded by the Deutsche Forschungsgemeinschaft (DFG, German Research Foundation) -- 492814705.
  Some parts were written under his current affiliation listed above, where this work was supported by the Engineering and Physical Sciences Research Council [grant number EP/Y008626/1].

  \bibliography{references}
\end{document}